\pdfoutput=1

\newif\ifsubmission
\submissiontrue
\submissionfalse
\newif\ifpersonalcopy
\personalcopyfalse
\personalcopytrue

\ifpersonalcopy
\documentclass[acmsmall,screen,nonacm]{acmart}
\else
\documentclass[acmsmall]{acmart}
\fi
\usepackage[T1]{fontenc}
\usepackage[utf8]{inputenc}
\newif\ifrevision
\revisiontrue
\revisionfalse
\newif\ifspacetrick
\spacetricktrue
\newif\ifwithcomments
\withcommentstrue
\withcommentsfalse
\newif\iflongversion
\longversiontrue
\longversionfalse

\usepackage{setup}
\definecolor{modification}{RGB}{35,178,44}
\definecolor{modification}{rgb}{0, 0, 0}

\setcopyright{rightsretained}
\acmDOI{10.1145/3632882}
\acmYear{2024}
\copyrightyear{2024}
\acmSubmissionID{popl24main-p201-p}
\acmJournal{PACMPL}
\acmVolume{8}
\acmNumber{POPL}
\acmArticle{40}
\acmMonth{1}
\received{2023-07-11}
\received[accepted]{2023-11-07}

\usepackage{booktabs}   
\usepackage{subcaption} 

\begin{document}

\title[]{Polymorphic Type Inference for Dynamic Languages}
\subtitle{Reconstructing Types for Systems Combining Parametric,
Ad-Hoc, and Subtyping Polymorphism}

\authorsaddresses{Authors'
addresses: \href{HTTPS://ORCID.ORG/0000-0003-0951-7535}{Giuseppe
Castagna}, \href{HTTPS://ORCID.ORG/0000-0003-1590-2392}{Mickaël
Laurent},  Institut de Recherche en Informatique Fondamentale (IRIF),
Université Paris Cité, CNRS, 8 place Aurélie Nemours, 75013 Paris, France; \href{HTTPS://ORCID.ORG/0000-0002-1729-870X}{Kim~Nguyen}, Laboratoire Méthodes Formelles, Université Paris-Saclay, CNRS, ENS Paris-Saclay,  91190, Gif-sur-Yvette, France}
\author{Giuseppe Castagna}
\orcid{0000-0003-0951-7535}
\affiliation{
  \department{Institut de Recherche en Informatique Fondamentale (IRIF)}
\institution{CNRS, Université Paris Cité}
  \country{France}
}

\author{Mickaël Laurent}
\orcid{0000-0003-1590-2392}
\affiliation{
  \department{Institut de Recherche en Informatique Fondamentale (IRIF)}
\institution{Université Paris Cité}
\country{France}
}

\author{Kim Nguy\~{\^e}n}
 \orcid{0000-0002-1729-870X}   
\affiliation{
  \department{Laboratoire Méthodes Formelles (LMF)}
  \institution{Université Paris-Saclay}
  \postcode{91190}
  \city{Gif-sur-Yvette}
  \country{France}
}

\begin{abstract}
We present a type system that combines, in a controlled way,
first-order polymorphism with intersection types, union types, and
subtyping, and prove its safety. We then define a type reconstruction
algorithm that is sound and terminating. This yields a system in which
unannotated functions are given polymorphic types (thanks to
Hindley-Milner) that can express the overloaded behavior of the
functions they type (thanks to the intersection introduction rule) and
that are deduced by applying advanced techniques of type narrowing
(thanks to the union elimination rule). This makes the system a prime
candidate to type dynamic languages.

\end{abstract}



\ifsubmission\else
\fi

\maketitle
\section{Introduction}
\label{sec:intro}
\ifpersonalcopy%
\begin{tikzpicture}[remember picture,overlay,shift={(current page.north east)}]
\node[anchor=north east,xshift=-1.6cm,yshift=-0.4cm]{\href{https://www.acm.org/publications/policies/artifact-review-and-badging-current}{\includegraphics[width=2cm]{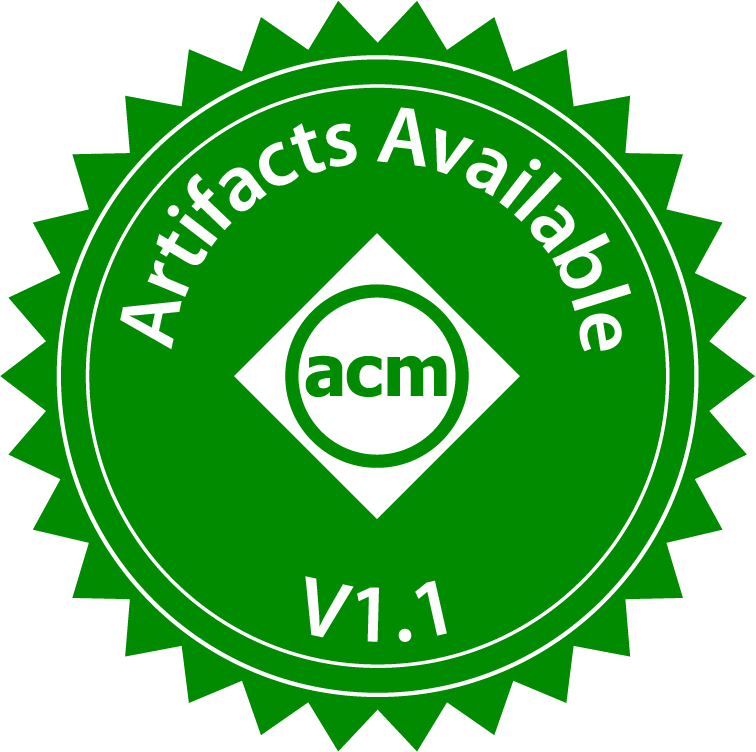}\includegraphics[width=2cm]{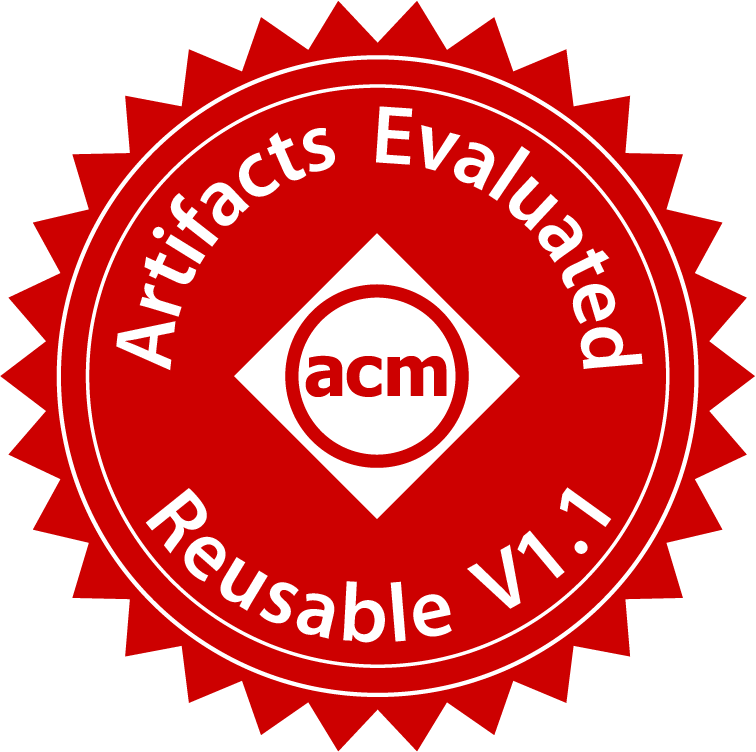}}};
\end{tikzpicture}\fi

Typing dynamic languages is a challenging endeavour
even for very simple pieces of code. For instance, JavaScript's 
logical or operator ``\jsinline{||}'' behaves like the following
function (also in JavaScript):\footnote{\color{modification}This definition does not
capture the short-circuit evaluation of ``\jsinline{||}''.}
\begin{javascript}
 function lOr (x, y) {                       !\label{ln:one}!
     if (x) { return x; } else { return y; }
 }                                           !\label{ln:two}!
\end{javascript}
A naive type for this function is
$(\Bool,\Bool)\to\Bool$, which states that \jsinline{lOr} is a function
that takes two Boolean arguments and returns a Boolean result. This
however is an
overly restrictive type, that does not account for the fact that in
JavaScript logical operators such as \jsinline{lOR} can be applied to
any pairs of arguments, not just to Boolean ones. JavaScript distinguishes two kinds of values: eight
``falsy'' values (i.e., \jsinline{false}, \jsinline{""}, \jsinline{0},
\jsinline{-0}, \jsinline{0n}, \jsinline{undefined}, \jsinline{null},
and \jsinline{NaN}) and the ``truthy'' values (all the others). The
expression \jsinline{if} executes the \jsinline{else} code if and only
if the tested value is falsy. If we want to change the previous type
to account for this fact, then we should give \jsinline{lOr} the
type $(\any,\any)\to\any$ (where $\any$ is the type of all values),
which is a rather useless type since it essentially states that
\jsinline{lOr} is a binary function. To give \jsinline{lOr} a more
informative type, we need union and intersection types (which are
already integrated in typed versions of JavaScript such as TypeScript~\cite{typescript} and Flow~\cite{Flow}): we define the type \Falsy as the following union type
$\jsinline{false}\vee\jsinline{""}\vee\jsinline{0}\vee\jsinline{-0}\vee\jsinline{0n}\vee\jsinline{undefined}\vee\jsinline{null}\vee\jsinline{NaN}$,
where each value denotes here the \emph{singleton type} containing
that value,
and the type \Truthy to be its complement, $\neg\Falsy$, that is, the
type of all
values that are not of type \Falsy. Then
we
can deduce for \jsinline{lOr} the following more precise type
\begin{equation}\label{monolor}
    ((\Truthy,\any)\to\Truthy)\wedge((\Falsy,\Truthy)\to\Truthy)\wedge((\Falsy,\Falsy)\to\Falsy)
\end{equation}
In this type, $\land$ is a type combinator denoting intersection and
meaning that the function has all the types given in the intersection:
that is, in words, if the first argument of a function of this type is
a \Truthy, then the function
returns a \Truthy regardless of the second argument (first arrow type),
while if the first argument is a \Falsy, then the result is of the same
type as the second argument's type (second and third arrow type). Notice how
the use of an intersection of arrow types corresponds to the typing
of an ``overloading'' behavior (also known as, \emph{ad hoc}
polymorphism~\cite{Str67}), insofar as the result of an application depends on the
\emph{type} of the input. 

In order to derive such a type, the type system must deduce that whenever the
condition tested by the \jsinline{if} holds, then \jsinline{x} is of type
\Truthy and, therefore, $(i)$ that all occurrences of \jsinline{x} in the
``then'' branch (here just one) have type \Truthy and $(ii)$ that all the
occurrences of the same variable \jsinline x in the branch ``else'' (here none)
have thus type \Falsy. This kind of deduction is usually referred as \emph{type
narrowing} or \textit{occurrence typing} since it requires to ``narrow'' the
type of a variable \jsinline x differently for its different occurrences. A type
system such as the one for Typed Racket---defined in \citep{THF10} where the
term \textit{occurrence typing} was first introduced---is able to \emph{check}
that \jsinline{lOr} has the type in \eqref{monolor}, meaning that the deduction
requires the programmer to explicitly specify the type in the
code.%
\iflongversion
\else
\ 
\fi
The system by
\citet{occtyp22} makes a step further, since not only it can check that
\jsinline{lOr}  has the type in \eqref{monolor}, but also it can reconstruct for
\jsinline{lOr} the intersection type
$((\Truthy,\any)\to\Truthy)\land((\Falsy,\any)\to\any)$ which, although it is
less precise a type than \eqref{monolor}, it is inferred from the code of
\jsinline{lOr} as is, without needing any type annotation. This latter work
constitutes the state of the art of this kind of inference, since it is the only
system that can reconstruct intersections of arrow types.

In this work we  go a step further, and show how to infer (i.e.,
reconstruct) intersections of \emph{polymorphic}
function types. In particular,  the system we present here reconstructs for
\jsinline{lOr} the
following first order polymorphic type (where $\alpha$ and $\beta$ are
type variables):\footnote{This type can be considered as an encoding of
$\forall(\alpha\leq\Truthy).\forall(\beta)\,.\,((\alpha,\any)\to\alpha)\wedge((\Falsy,\beta)\to\beta)$
a type expressed in so-called bounded polymorphism: see~\citet[Section
  2]{Cas22}.\label{fn:bounded}}
\begin{equation}\label{polylor}
\forall\alpha,\beta\,.\,((\alpha\Wedge\Truthy,\any)\to\alpha\Wedge\Truthy)\wedge((\Falsy,\beta)\to\beta)
\end{equation}
This type completely specifies the semantics of the function \jsinline{lOR}: it states that if the
first argument is a \Truthy, then the application of the function returns the first argument,\footnote{Strictly speaking, the type
states that the function returns a result of the same type as the
first argument, but by parametricity we can deduce that the result
will be the first argument. Likewise for the second argument. A simple
way to understand it, is by instantiating both type variables in \eqref{polylor}
with the singleton type of the (value result of the) argument.} otherwise
it returns the second argument. This type is more precise than the one
in~\eqref{monolor}, since it allows the system to deduce that, say, if
the first argument of \jsinline{lOR} is an object, then the result will be
an object of the same type (rather than just a truthy value).
Not only does the system we present here infers such a precise type,
but this kind of precision  is compositional, yielding an
accurate type also for the expressions in which the function is used.
For instance, if we define the following function:
\begin{javascript}
 function id (x) {    !\label{ln:ids}!
     return lOR(x, x)
 }                    !\label{ln:ide}!
\end{javascript}
then, as we explain later on, our system infers that \jsinline{id} has type
$\forall\alpha.\alpha\to\alpha$, viz., that \jsinline{id} is indeed the
polymorphic identity function.

This is clearly better than the current state of the art. Still, it does
not seem too hard a feat to deduce that if we are testing whether \jsinline
x is a truthy value, then when the test succeeds we can assume that \jsinline x
is of type \Truthy. To show the more advanced capabilities of our system let us
have a look at how ECMAScript specifies the semantics of JavaScript logical
operators, as defined in the 2021 version of the
specification~\cite[Section 13.13.1]{ecmascript12}. Since in JavaScript there
are no union or intersection types, then the falsy and truthy values are defined
via an (abstract) function \jsinline{ToBoolean} which simply checks whether its
argument is one of the 8 falsy values and returns \jsinline{false}, otherwise it
returns \jsinline{true} (see its definition in row 1 of
Table~\ref{tab:bench} in Section~\ref{sec:implementation}). In our system, \jsinline{ToBoolean} has type
$(\Truthy\to\jsinline{true})\wedge(\Falsy\to\jsinline{false})$. All logical
operators are then defined by ECMAScript in terms of this function: this has the
advantage that any change to the specification of falsy (e.g., the addition of a
new falsy value, like the addition of the built-in \jsinline{bigint} type and
its constant \jsinline{0n} in ES2020) requires only the modification of this
function, and is automatically propagated to all operators. So the actual
definition of \jsinline{lOr} for ECMAScript is the following one:
\begin{javascript}
 function lOr (x, y) {
     if (ToBoolean(x)) { return x; } else { return y; } !\label{ln:tobool}!
 }
\end{javascript}
If we feed this function to our system, then it infers for it the type in
\eqref{polylor}, that is, the same type it already deduced for the simpler
version of \jsinline{lOr} defined in lines~\ref{ln:one}-\ref{ln:two}. But here
the deduction needed to perform type narrowing is more challenging, since the
system must deduce from the type
$(\Truthy\to\jsinline{true})\wedge(\Falsy\to\jsinline{false})$ of
\jsinline{ToBoolean} that when the \emph{application} in line \ref{ln:tobool}
returns a truthy value, then the \emph{argument} of \jsinline{ToBoolean}  is of
type \Truthy, and it is of type \Falsy otherwise. More generally, we need a
system which, when a test is performed on an arbitrarily complex application,
can narrow the type of all the variables occurring in the application by
exploiting the information provided by the overloaded behavior of the functions
therein.  Achieving such a degree of precision is a hard feat but, we argue, it
is necessary if we want to reconstruct types for dynamic languages, that is, if
we want to type their programs as they are, without requiring the addition of
any type annotations. Indeed, the core operators of these languages (e.g.,
JavaScript's ``\jsinline{||}'', ``\jsinline{&&}'', ``\jsinline{typeof}'',
\ldots) are characterized by an ``overloaded'' behavior, which is then passed
over to the functions that use them. So for instance a simple use of JavaScript
logical or ``\jsinline{||}'' such as in \jsinline{(x => x || 42)} results in a
function whose precise type, as reconstructed by our system, is
$(\Keyw{Falsy}\rightarrow
\jsinline{42})\land(\Keyw{Truthy}\land\alpha\rightarrow\Keyw{Truthy}\land\alpha)$.
JavaScript functions also routinely perform dynamic checks against constants
(notably \jsinline{null} and \jsinline{undefined}), which our system also
handles as part of its more general approach to type narrowing of arbitrary
expressions. 

\subsection{Outline}\label{sec:outline}
\paragraph{Type System (Section~\ref{sec:types})}
So, how can we achieve all this? Conceptually, it is quite simple: we
just merge together three of the most expressive type systems studied
in the literature, namely the Hindley-Milner (HM) polymorphic types~\cite{Hin69,Mil78},
intersection types~\cite{CDV81}, and union types~\cite{MacQueen1986,BDD95}. We achieve it simply by putting
together in a controlled way the deduction rules characteristic of
each of these systems (see Figure~\ref{fig:declarative} in Section~\ref{sec:types})  and proving
that the resulting system is sound (cf., Theorem~\ref{th:typesound}).

More precisely, the type system we describe in Section~\ref{sec:types}
is pretty straightforward. Its core is a classic HM system with
first order polymorphism: a program is a list of let-bindings that
define polymorphic functions; these are typed by inferring a type for
the expressions that define them, this type is then
generalized, yielding a prenex polymorphic type for the function. As usual, the
deduction of the type of each of these expressions is performed in a type
environment that records the generic types for the previously-defined
polymorphic functions, and the type system can instantiate these types
differently for each use of the polymorphic functions in the
expression. The novelty of our system is that when
deducing the types of the expressions  that define the polymorphic
functions, the type system can use not only instantiations of
polymorphic types (rule \Rule{Inst} in Figure~\ref{fig:declarative}), but also intersection and union types.
More precisely, to type these expressions the type system can decide
to use the classic rules of
intersection introduction (rule \Rule{$\wedge$}) and union elimination
(rule \Rule{$\vee$})  given in Figure~\ref{fig:declarative}. For instance, the intersection introduction
rule is used by the system to deduce that since the function \jsinline{lOr}
(either versions) has both type
$((\alpha\Wedge\Truthy,\any)\to\alpha\Wedge\Truthy)$ and type
$((\Falsy,\beta)\to\beta)$, then it has their intersection, too; this
intersection type
is then generalized (when \jsinline{lOr} is defined at top-level)
yielding the polymorphic type in~\eqref{polylor}. The union
elimination rule is essentially used to fine-grainedly type branching
expressions and tests involving applications of overloaded functions:
for instance, to deduce that the function \jsinline{id} in lines
\ref{ln:ids}--\ref{ln:ide} has type $\alpha\to\alpha$, the system can
assume that \jsinline x has type $\alpha$ and separately infer the
type of the body for $\jsinline x : (\alpha\Wedge\Truthy)$ and for
$\jsinline x : (\alpha\Wedge\neg\Truthy)$; since the first
deduction yields  $(\alpha\Wedge\Truthy)$ and the second yields
$(\alpha\Wedge\neg\Truthy)$, then the system deduces that under the
hypothesis  $\jsinline x : \alpha$, the body has the union of these
two types, that is $\alpha$%
\iflongversion
\else%
.
\fi
Furthermore, as observed by~\citet{occtyp22}, the combination of the
union elimination with the rules of type-cases given in
Figure~\ref{fig:declarative} constitutes the essence of narrowing and
occurrence typing.

The declarative type system given in
Section~\ref{sec:types} is all well and good, but how can we define an
algorithm that infers whether a given expression can be typed in this
system? Rules such as union elimination and intersection introduction
are easy to understand, but they
do not easily lend themselves to an implementation. In order to arrive
to an effective implementation of the type system specified in
Section~\ref{sec:types} we proceed in two steps: $(i)$ the definition of an
algorithmic system and $(ii)$ the definition of a reconstruction algorithm.

\paragraph{Algorithmic System (Section~\ref{sec:algo})} 

The first step towards an effective implementation of our type system
is taken in Section~\ref{sec:algo} where we define an algorithmic system
that is sound and complete with respect to the system of
Section~\ref{sec:types}. The system is algorithmic since it is
composed only by syntax-directed and analytic rules\footnote{A rule is
\emph{analytic} (as opposed to \emph{synthetic}) when the input (i.e.,
$\Gamma$ and $e$) of the judgment at the conclusion is sufficient to
determine the inputs of the judgments at the premises
(cf.\ \cite{ML1994,types2019}).}  and, as such, is immediately
implementable. It is sound and complete since an expression is typable
in it if and only if it is typable in the system of
Section~\ref{sec:types}. To obtain this results the system is defined
on pairs formed by an MSC-form (Maximal Sharing Canonical form) and an
annotation tree. MSC-forms are A-normal forms~\cite{SF92} on steroids:
they are lists of bindings associating variables to expressions in
which every proper subexpression is a variable. Their characteristic
is that they encode expressions and preserve typability in the sense
that every expression is typable if and only if its unique MSC-form is
typable. MSC-forms were introduced by~\citet{occtyp22} to drastically
reduce the range of possible applications of the union elimination
rule; here we improve their definition to deal with our polymorphic
setting and use them for exactly the same reason as
in~\cite{occtyp22}. Annotation trees encode canonical derivations of
the system of Section~\ref{sec:types} for the MSC-form they are paired
with. They are a generalization of type annotations inserted in the
code. Instead of annotating directly an MSC-form with type-annotations
we used a separate annotation tree because of the union elimination
rule which types several times the same expression under different
type environments; this would, thus, require different annotations for
the same subexpressions, each annotation depending on the typing
context: this naturally yields to tree-shaped annotations in which each
branching corresponds either to the different deductions performed by
a union elimination rule or to the different deductions performed by
an intersection introduction rule. The soundness and
completeness properties of the algorithmic systems are thus stated
in terms of MSC-forms and annotation trees. They essentially state
that an expression $e$ has type $t$ in the declarative system of
Section~\ref{sec:types} if and only if there exists a tree annotation
for the (unique) MSC-form of $e$ that is typable in the algorithmic system
with (a subtype of) $t$: see Theorem~\ref{th:main}.

\paragraph{Reconstruction Algorithm (Section~\ref{sec:reconstruction})} The second of the two steps to
achieve an effective implementation for the type system of
Section~\ref{sec:types} is to define a reconstruction algorithm for
the previous algorithmic system, which we do in
Section~\ref{sec:reconstruction}. The statements of the soundness and
completeness properties of the algorithmic system clearly suggest what
this algorithm is expected to do: given an expression that defines a
polymorphic function, the algorithm must transform it into its unique
MSC-form and then try to reconstruct an annotation tree for it so that
the pair MSC-form and annotation tree is typable in the algorithmic
system of Section~\ref{sec:algo}.

The reconstruction is performed by a system of deduction rules that incrementally refines
an annotation tree (initially composed of a single node ``$\annotiinfer$'')
while exploring the list of bindings of the MSC-form of the expression to type. It mixes two independent
mechanisms: one that infers the domain(s) of $\lambda$-terms, and the other that performs type narrowing
when a typecase is encountered.

The first mechanism is inspired by the algorithm $\mathcal{W}$ by
\citet{DM82}: whenever the application of a destructor (e.g., a function application) is
encountered, an algorithm finds a substitution (if any) that makes
this application well-typed. In the context of a HM type system, the
algorithm at issue needs to solve a \emph{unification problem} (i.e.,
whether for two given types $s$ and $t$ there exists a substitution
$\sigma$ such as $s\sigma=t\sigma$) which, if solvable, has a principal
solution given by a single substitution~\cite{Rob65}. In our system,
which is based on subtyping, the algorithm at issue needs to solve a
\emph{tallying problem} (i.e., whether for two given types $s$ and $t$
there exists a substitution $\sigma$ such as $s\sigma\leq t\sigma$)
which, if solvable, has a principal solution given by \emph{a finite set} of
substitutions~\cite{polyduce2}. When multiple substitutions are found,
they are all considered and explored in different branches by adding
an intersection branching node in the current annotation tree.

The second mechanism gets inspiration from \citet{occtyp22} and
refines decompositions made by the union-elimination rule in order to
narrow the types of variables in the branches of a typecase
expression.  When the system encounters a typecase that tests whether some expression $e$ has
type $t$, then the type $s$ of the variable
bound to $e$ (recall that an MSC-form is a list of bindings) is split into $s\land t$ and $s\land \neg t$, and
these splits are in turn propagated recursively in order to generate
new splits for the types of the variables associated with the
subexpressions composing $e$.  For instance, when the algorithm
encounters the test ``\jsinline{if (ToBoolean(x))...}'' at
line~\ref{ln:tobool}, it splits the type of (the variable bound to)
\jsinline{ToBoolean(x)} in two, by intersecting it with
$\jsinline{true}$ and $\neg\jsinline{true}$, and this split in turn
generates a new split $\Truthy$ and $\Falsy$ for the type of the
variable \jsinline{x}.

The reconstruction algorithm we present in Section~\ref{sec:reconstruction} is
sound: if it returns an annotation tree for an
MSC-form, then the pair is typable in the algorithmic system, whose
soundness implies that the expression at the origin of the MSC-form is
typable in the system of Section~\ref{sec:types}.  At this point,
however, it should be pretty obvious that such a reconstruction
algorithm cannot be complete. Our system merges three well known
systems: first-order parametric polymorphism, intersection types,
union elimination. Now, even if parametric polymorphism is decidable,
in our system we can encode (and type, via intersection types)
polymorphic fixed-point combinators, yielding a system with polymorphic
recursion whose inference has been long known to be
undecidable~\cite{Hen93,KTU93}.
Worse, our system
includes union elimination, which is one of the most problematic rules
from an algorithmic viewpoint, not only because it is neither syntax
directed nor analytic, but also because determining an inversion
(a.k.a., generation) lemma for this rule is considered by experts the
most important open problem in the research on union and intersection
types~\cite{DezaniPC}, and an inversion lemma is somehow the first
step to define a type-inference algorithm, since it tells us when and
how to apply the rule. \textcolor{modification}{We discuss in detail
the reasons and implications of incompleteness
in Section~\ref{sec:propreco}.}

Despite being incomplete, our reconstruction algorithm is powerful
enough to handle both complicated typing use-cases and common programming
patterns of dynamic languages.  For instance, for the $Z$ fixed-point
combinator for strict languages $Z =\lambda f.(\lambda x.f(\lambda
v.xxv))\ (\lambda x.f(\lambda v.xxv))$ our algorithm
reconstructs the type $\forall \alpha, \beta,
\gamma. ((\alpha\To\beta)\to((\alpha\To\beta)\Wedge\gamma))\to((\alpha\To\beta)\Wedge\gamma))$
(i.e., in bounded polymorphic notation $\forall (\alpha)(\beta)
(\gamma\leq\alpha\To\beta). ((\alpha\To\beta)\to \gamma)\to\gamma$,
cf.\ Footnote~\ref{fn:bounded}). The combinator can then be used as
is, to define and infer the type of classic polymorphic functions such as \code{map},
\code{fold}, \code{concat}, \code{reverse}, etc., often yielding types more precise than in HM:
for instance if we use $[\alpha*]$ to denote the type of the lists whose elements
have type $\alpha$, then the type inferred for (a curried version of)
\code{fold\_r} is
$\forall\alpha,\beta,\gamma.((\alpha\To\beta\To\beta)\to\beta\to[\alpha*]\to\beta)\wedge(\any\to\gamma\to[\,]\to\gamma)$
where the second type in the intersection states that if the third
argument is an empty list, then the result will be the second argument,
whatever the type of the first argument is. Finally, we designed
our algorithm so that it can take into account explicit type annotations
to help it in the inference process. As an example, our
algorithm can check that the classic \code{filter} function has type
$\forall\alpha,\beta,\gamma.((\alpha\Wedge\beta\To\Bool)\wedge(\alpha\Wedge\neg\beta\To\False))\to[\alpha*]\to[(\alpha\Wedge\beta)*]$,
stating that if we pass to \code{filter} a predicate that returns false for the
elements of $\alpha$ that are not in $\beta$, then filtering a list of
$\alpha$'s will return only elements also in $\beta$.

\smallskip
Sections~\ref{sec:types}, \ref{sec:algo}, and \ref{sec:reconstruction} outlined above
constitute the core of our
contribution. Section~\ref{sec:implementation} presents our implementation.
In Section~\ref{sec:related}
we discuss related work and Section~\ref{sec:conclusion} concludes our
presentation. For space reasons we omitted in the main text
some rules of the algorithmic and reconstruction systems, as well as
all proofs: they are all given in the appendix,
\ifsubmission
provided as supplemental material for the submission.
\else
available on line as supplemental material.
\fi

\subsection{Discussion, Contributions, and Limitations}\label{sec:discon}

\paragraph{Intersections vs.\ Hindley-Milner}
It is a truth universally acknowledged that intersection type systems
are more powerful than HM systems: for that, one does not even need
full intersections, since Rank 2 intersections suffice. Rank 2
intersection types are types that may contain intersections only to
the left of a single arrow and the system of Rank 2 intersection types
is able to type all ML programs (i.e., all program typable by HM), has principal typings, decidable type
inference, and the complexity of type inference is of the same order
as in ML~\cite{Leivant83}.

However, intersection type systems are not compositional, and this
hinders their use in a modular setting. A program that uses the
polymorphic identity function to apply it to, say, an integer and a
Boolean, type checks since we can infer that the polymorphic identity
function has type $(\Int\To\Int)\Wedge(\Bool\To\Bool)$. But if we want
to \emph{export} this polymorphic identity function to use it in other
unforeseen contexts, then we need for it a type that covers all its
admissible usages, without the need of retype-checking the function
every time it is applied to an argument of a new type.  In other
words, in a modular usage, parametric polymorphism has an edge over
intersection/ad-hoc polymorphism despite being less powerful, since a
type such as $\forall \alpha.\alpha\To\alpha$ synthesizes the
infinitely many combinations of intersection types that can be deduced
for the identity function; however in a local setting, everything that
does not need to be exported can be finer-grainedly typed by 
intersection types.  This division of roles and responsibilities is at
the core of our approach. As we show in the next section, programs are
lists of bindings from variables to expressions. These expressions are
typed in a type environment (generated by the preceding bindings)
which binds variables to polymorphic types. These expressions are
typed by using instantiation, intersection introduction, and union
elimination, but \emph{not} generalization. Generalization is
performed only at top level, that is at the level of programs and
reserved to expressions to be used in other contexts.

\paragraph{Parametricity vs.\ type cases} A
parametric polymorphic function (a.k.a., a generic function) is a function that behaves
uniformly for all possible types of its arguments, that is, whose
behavior does not depend on the type of its arguments. A common way to
characterize a generic function is that it is a function that cannot inspect the parametric parts of its input,
that is, those parts that are typed by a type variable: these parts can
only be either returned as they are, or discarded, or passed to another generic
function. Our approach suggests refining this characterization by
shifting the attention from inputs as a whole to some particular
values among all the possible inputs. This can be seen by comparing the
following two function definitions:\\[1mm]
\centerline{
$\lambda x.\tcase x \Int x x\qquad\qquad\qquad\lambda x.\tcase x \Int
{x+1} x$}

\noindent
Both functions test whether their input is an integer. The function on
the right-hand side returns the successor of the argument when this is
true and the argument itself otherwise; the one on the left-hand side
returns its argument in both cases, that is, it is the identity
function.

Our system deduces for the function on the left the type
 $\alpha\to\alpha$.\footnote{Precisely, it deduces for it the type
 $(\alpha{\wedge}\Int\to\alpha{\wedge}\Int)\wedge(\beta{\setminus}\Int\to\beta{\setminus}\Int)$. Instantiating
 $\beta$ to $\alpha$ yields a subtype of $\alpha\to\alpha$.}  For the
 function on the right it returns the type
 $(\Int\to\Int)\wedge(\alpha{\setminus}\Int\to\alpha{\setminus}\Int)$,
 where $s{\setminus}t$ denotes the set-theoretic difference of the two
 types, that is, $s{\wedge}\neg t$.  These two types suggest how we
 can refine the intuitive characterization of parametricity. A generic function
 can inspect the parametric part of its input {\color{modification}
 (as the function on the left-hand side shows) and its output can
 depend on this inspection (as the function on the right-hand side
 shows), but the parts of its output that are typed by a type
 variable---i.e., the ``parametric'' parts---cannot depend on it. We
 can speak of ``partial'' parametricity, and say that a function is
 parametric ``only'' for the inputs (or parts thereof) that are either
 returned unchanged or discarded: the type variables in its type
 describe such inputs. For instance, the domain of both the functions
 above is $\any$: they both can be applied to any argument. But the
 first function is parametric for all possible inputs, since the
 result of the inspection is not used to produce any particular output
 (it has type $\forall\alpha.\alpha\to\alpha$), while the second
 function is parametric only for the values of its domain that are not
 in \Int, since it uses the result of the inspection to generate the
 result for the integer inputs (by subsumption, the second function
 has type
 $\forall\alpha.\,\alpha\to\Int{\vee}(\alpha{\setminus}\Int)$:
 parametricity holds only for the arguments not in \Int).  }

\paragraph{Contributions}
The general contribution of our work is twofold. First, it proposes a
way to mix parametric and intersection/ad-hoc polymorphism which, in
hindsight, is natural: parametric polymorphism  for
everything defined at top-level and that can thus be used in other
contexts (modularity); intersection polymorphism for everything that
remains local (for which we can thus use more precise non-modular typing). Second, it proposes an effective
way to implement this type discipline by defining a reconstruction
algorithm; with respect to that, a fundamental role is
played by the analysis of the (type-)tests performed by the
expressions, since they drive the way in which types are split:
externally, to split the domain of functions yielding intersection of
arrows (intersection introduction); internally, to split the type of
tested expressions, yielding a precise typing of branching (union elimination). In doing
so, it provides the first system that reconstructs types that oncombine
parametric and \emph{ad hoc} polymorphism.

The technical contributions of the work can be summarized as follows:
\begin{enumerate}[topsep=0pt]
\item We define a type system that combines parametric polymorphism
with union and intersection types for a functional
calculus with type-cases and prove its soundness.
\item We define an algorithmic system that we prove sound and complete
with respect to the previous system.
\item We define an algorithm to reconstruct the type annotations of
the previous algorithmic system and prove it sound and terminating.
\end{enumerate}
The reconstruction algorithm is fully implemented. A prototype
which also implements optional type annotations and pattern matching 
(presented in Appendix~\ref{sec:extensions-appendix})
is available on-line at \url{https://www.cduce.org/dynlang}, and whose
sources are on Zenodo:~\cite{CLN24proto}.

{\color{modification}
\paragraph{Limitations}

The system we present here has some limitations.
Foremost, the reconstruction algorithm of
Section~\ref{sec:reconstruction} uses backtracking, and at each of its
passes it may try to type the same piece of code several
times. Backtracking is inherent to our algorithm, since it proceeds by
successively refining in different passes, the annotation tree of an
MSC-form.  The checking of a same piece of code several times at each
pass is inherent to the use of unions and intersections: the
union-elimination rule repeatedly type-checks the same expression, using different type
hypotheses for a given sub-expression; the intersection-introduction
rule verifies that an expression has all the types of an intersection,
by checking each of them separately. Both features are
very penalizing in terms of performance, and any naive implementation
of the reconstruction described in Section~\ref{sec:reconstruction} would yield
type-inference times that grow exponentially with the size
of the program. Clearly, this is an issue that must be addressed if we
want to apply our system to real-world dynamic languages, and further
work is needed to frame and/or constrain the current
system so that its performance becomes acceptable. Fortunately, the room
for improvement is significant: our prototype is an unoptimized proof
of concept whose implementation was defined to faithfully simulate the
reconstruction inference rules, rather than to obtain an efficient
execution; but the simple addition of textbook memoization
techniques improved its performance by an order of magnitude
(cf.\ Section~\ref{sec:implementation}).

A second limitation of our system is that it is
not sound in the presence of side-effects. The algorithm transforms an
initial expression into its Maximal Sharing Canonical form, which is a list of
bindings, one for each sub-expression of the initial expression. As we explain in Section~\ref{sec:mscf},
these forms are called ``maximal sharing'' since all equivalent sub-expressions
(in the sense
stated by Definition~\ref{def:order}) of the initial expression must be bound by the same
variable, so that any refinement of the type of one sub-expression
(e.g., as a consequence of a type-case) is passed-through to all
equivalent sub-expressions. However, this is sound only if all evaluations
of equivalent sub-expressions return results that have the same 
types. While this is true for pure expressions, this can be invalidated
by the presence of side-effects. In Section~\ref{sec:conclusion} we
suggest some research directions on how to modify the equivalence relation of
Definition~\ref{def:order} to make our system work in the presence of
side-effects. Nevertheless, the work presented here is closer to be
adapted/adaptable to pure functional languages such as Erlang and
Elixir, than to languages such as JavaScript or Python.

Finally, it may be worth pointing out that our approach works only for
strict languages, since it uses a semantic subtyping relation that is
unsound for call-by-name evaluation strategies~\cite{types18}.
}

\section{Source Language and Type System}
\label{sec:types}
\subsection{Syntax and Semantics}
Our core language is fully defined in Figure~\ref{fig:syntax}.
\begin{figure}[h]
\input{fig-syntax}\ifspacetrick\vspace{-1mm}\fi
\caption{Syntax and semantics of the source language\label{fig:syntax}}
\end{figure}
Expressions are an untyped $\lambda$-calculus with
constants $ c $, pairs $ (e, e) $, pair projections $ \pi_i{e} $, and
type-cases.
A type-case $ \tcase{e_0}\tau{e_1}{e_2} $ is a dynamic type test that
first evaluates $ e_0 $
and, then, if $ e_0 $ reduces to a value $ v $,
evaluates $ e_1 $ if $ v $ has type $ \tau $ or $ e_2 $ otherwise.
Type-cases cannot test arbitrary types but just ground types (i.e.,
types without type variables occurring in them) of the form $
\tau $ where the only arrow type that can occur in them is $ \Empty
\to \Any$, the type of all functions.
This means that type-cases can distinguish functions from other values
but they cannot distinguish, say,
functions that have type $ \Int \To \Int $ from those that do not.

Programs are sequences of top-level definitions, ending with an expression that can be seen as the main entry.
This notion of program is useful to capture
the modularity of our type system. Indeed, top-level definitions are typed sequentially:
the type we obtain for a top-level definition is considered definitive and will not be
challenged by a later definition.

The reduction semantics for expressions is the one of call-by-value
pure $\lambda$-calculus with products and with a type-case expression,
together with the context rules that implement a leftmost outermost
reduction strategy. We use the standard substitution operation
$e\subs x{e'}$ that denotes the capture avoiding substitution of $e'$
for $x$ in $e$, whose definition we recall in Appendix~\ref{sec:sem-appendix}.
The relation $v\in\tau$ determines whether a \emph{value}
is of a given type or not and holds true if and only
if $\typof v\leq \tau$, where $\leq$ is the subtyping relation defined
by~\citet{CX11} (we recall its definition in Appendix~\ref{sec:subtyping-appendix}).
Note that $ \typof v $ maps every $\lambda$-abstraction
to $ \Empty \to \Any $ and, thus, dynamic type tests do not depend on static type inference.
This approximation is allowed by the restriction on arrow types in typecases.
Finally, the reduction semantics for programs sequentially reduces top-level definitions,
together with a context rule that allows reducing the expression of the first definition.

\subsection{Types}

Types are those by~\citet{CX11} who add type variables to the semantic subtyping framework of~\citet{FCB02,FCB08}.
\begin{definition}[Types]\label{def:types}
The set of types \types{} is formed by the terms $t$ coinductively produced by the grammar:\ifspacetrick\vspace{-3mm}\fi
\[
\begin{array}{lrcl}
\types & t,s & ::= & b\alt \polyvar\alt\monovar\alt t\to t\alt t\times t\alt t\vee t \alt \neg t \alt \Empty 
\end{array}
\]
and that satisfy the following conditions:
\iflongversion
\else
 $(i)$ every term has a finite number of different sub-terms
(regularity) and $(ii)$ every infinite branch of a term contains an infinite number of occurrences of the
arrow or product type constructors (contractivity).
\fi
\end{definition}
\noindent We use the abbreviations
$
  t_1 \land t_2 \eqdef \lnot (\lnot t_1 \lor \lnot t_2)
$, $
  t_1 \setminus t_2 \eqdef t_1 \land (\lnot t_2) $, and $ \Any \eqdef \lnot
\Empty $. Basic types (e.g., \Int, \Bool) are ranged over by $b$, $\Empty$ and $\Any$
respectively denote the empty (that types no value) and top (that types all
values) types. Coinduction accounts for
recursive types and the condition on infinite branches bars out
ill-formed types such as 
$t = t \lor t$ (which does not carry any information about the set
denoted by the type) or $t = \neg t$ (which cannot represent any set).

For what concerns type variables, we choose \emph{not} to use
type-schemes but rather distinguish two kinds of type variables.
\emph{Polymorphic type variables} ranged over by $\polyvar$, are type variables
that have been  generalized and can therefore be instantiated. In a more
traditional presentation, such variables are bound by the $\forall$ of a
type-scheme ; the set of polymorphic variables is $\polyvars$. \emph{Monomorphic
type variables}, ranged over by $\monovar$ (with bold font), are variables that are not
generalized and therefore cannot be instantiated; the set of monomorphic
variables is $\monovars$. Types that only contain monomorphic
variables are dubbed monomorphic types%
\footnote{\color{modification}The term \emph{polytypes} and \emph{monotypes} can be found
(albeit inconsistently) in the literature: in particular, \citet
{Mil78} uses the latter to denote
types with no type variables and the former when he wishes to imply that
a type may, or does, contain a variable. We avoided using
them to prevent any confusion with our
\emph{monomorphic types}. While our \emph{types} are indeed polytypes, our
\emph{monomorphic types} are not monotypes: monotypes do not have type
variables while monomorphic types may have type variables,
though only monomorphic ones. So we used instead \emph{types}
(which may have type variables), \emph{ground types} (which cannot have any type variable), and  \emph{monomorphic types}
(which may have monomorphic type variables, only). }:
\[
  \begin{array}{lrclr}
    \textbf{Monomorphic types}  &\mt,\ms &::=& b\alt \monovar \alt \mt\to \mt\alt \mt\times \mt\alt \mt\vee \mt \alt \neg \mt \alt \Empty\\
  \end{array}
\]
Our choice of using two disjoint sets for polymorphic and monomorphic type variables, instead of
the classical approach of using type schemes $\forall\alpha_1...\alpha_n. t$, is justified by two reasons.
First, type schemes are expected to be equivalent modulo $\alpha$-renaming. In our case however,
we do not want polymorphic type variables to be freely renamed because of the use, in the algorithmic type system
of Section~\ref{sec:algo}, of external annotations containing explicit substitutions over some polymorphic type variables
of the context.
Secondly, introducing type schemes would require redefining many of
the usual set-theoretic type-related definitions,
such as the subtyping relation $\leq$, and the type operators for
application $\circ$ and projections $\pi_i$. Instead, we obtain a more streamlined theory by making
subtyping and these operators ignore whether a variable is polymorphic or monomorphic in the current context,
and by explicitly performing instantiations in the type system when required.

The subtyping relation for these types, noted $\leq$, is the one defined
 by~\citet{CX11},
to which the reader may refer for the formal
definition (cf.\ Appendix~\ref{sec:subtyping-appendix}). 
For this presentation, it suffices to consider that
ground types (i.e., types with no variables) are interpreted as sets of \emph{values}
that have that type, and that subtyping is set
containment (i.e., a type $s$ is a subtype of a type $t$ if and only if $t$
contains all the values of type $s$). In particular, $s\to t$
contains all $\lambda$-abstractions that when applied to a value of
type $s$, if their computation terminates, then they return a result of
type $t$ (e.g., $\Empty\to\Any$ is the set of all
functions and $\Any\to\Empty$ is the set
of functions that diverge on every argument). Type connectives
(i.e., union, intersection, negation) are interpreted as the
corresponding set-theoretic operators.
For what concerns non-ground types (i.e., types with
variables occurring in them) all the reader needs to know for this
work is that the subtyping relation of \citet{CX11} is preserved
by type-substitutions. Namely, if $s \leq t$, then $s\sigma \leq t\sigma$ for every
type-substitution $\sigma$.
We use $\simeq$ to denote the symmetric closure of $\leq$,
thus $s\simeq t$ (read, $s$ is equivalent to $t$) means that $s$ and $t$
denote the same set of values and, as such, they are semantically the same type.

\subsection{Type System}
Our type system is given in full in Figure~\ref{fig:declarative}. 
The typing rules for expressions are, to some extent, the usual ones.
Constants and variables are typed by the corresponding axioms \Rule{Const} and
\Rule{Ax}. The arrow and product constructors have introduction and elimination
rules. Notably, in the case of rule~\Rule{$\rightarrow$I} the type of the
argument is monomorphic. The rules for intersection (\Rule{$\wedge$}) and subtyping (\Rule{$\leq$}) are the classical
ones, and so is the rule for instantiation (\Rule{Inst}) where $\sigma$
denotes a substitution from polymorphic variables to
types. The type-case construction is handled by three rules:
\Rule{$\Empty$}; \Rule{$\in_1$}; \Rule{$\in_2$} . Rule \Rule{$\Empty$}
handles the case
where the tested expression is known to have the empty type. The other two are
symmetric and handle the case when the tested expression is known to have either
the type $\tau$ or its negation, in which case the corresponding branch is
typed. These rules work together with
Rule~\Rule{$\vee$}, which we describe in detail next.
\begin{figure}[t]
\input{fig-declarative}\ifspacetrick\vspace{-4mm}\fi
\caption{Declarative Type System\label{fig:declarative}\ifspacetrick\vspace{-1mm}\fi}
\end{figure}

 At first sight, the formulation of rule \Rule{$\vee$} seems odd, since
the $\vee$ connector does not appear in it. To understand it, 
consider the classic union elimination rule by~\citet{MacQueen1986}:
\begin{mathpar}
  \Infer[$\vee$E]
      {     \Gamma \vdash e':s_1{\vee} s_2 \quad
        \Gamma, x:s_1\vdash e:t\quad \Gamma, x:s_2\vdash e:t
      }  
      {
      \Gamma\vdash e\subs x {e'}  : t
      }{}
\end{mathpar}
Rule \Rule{$\vee$E} types an expression that contains 
      occurrences of an expression $e'$ that has a union type $s_1\vee
      s_2$; the rule substitutes in this expression \emph{some}
      occurrences of $e'$ by the variable $x$ yielding an expression
      $e$, and then types $e$ first under the hypothesis that $x$ has
      type $s_1$ and then under the hypothesis that $x$ has type
      $s_2$. If both succeed, then the common
type is returned for the expression at issue. As shown
      by~\citet{occtyp22}, this rule, together with
the rules for type-cases, allows the system to perform occurrence typing. For
instance, consider the expression $\tcase{f y}\Int {(f y) +
1}{\jsinline{false}}$, in the context where $f$ has type $\any{\to}\any$
      and $y$ is of type $\any$. This
expression can be typed thanks to the rule $[\vee\textrm{E}]$, by considering the
sub-expression $f y$. This sub-expression has
      type $\any$, which can be seen as the union type $\any \simeq \Int{\vee}\lnot\Int$. We
can then replace $x$ for $f y$ and type, using $[\in_1]$, the expression $\tcase{x}\Int{x +
1}{\jsinline{false}}$, with $x:\Int$. This yields a type $\Int$ (rule~$[\in_1]$
ignores the second branch) and by subtyping, the expression has type
$\Int{\vee}\False$. Likewise for the choice $x:\lnot\Int$, using rule~$[\in_2]$
the second branch has type $\False$ and therefore $\Int{\vee}\False$ (again via
subtyping). The whole expression has thus the desired type $\Int{\vee}\False$.

A key element is that rule~$[\vee\textrm{E}]$ guessed how to split
the type $\any$ of $f y$ into $\Int\vee\lnot\Int$. In a non-polymorphic setting,
this is perfectly fine.
But in
a type-system featuring polymorphism, particular care must be taken when
introducing (fresh) type variables. As it is stated,
MacQueen et al.'s $[\vee\textrm{E}]$ rule could choose to split, say, $\Any$ into a union
$\alpha\vee\lnot\alpha$, with $\alpha$ a polymorphic type variable. If so, then
the rule becomes unsound. As a matter of fact, the premises of the \Rule{$\vee$E} behave
as in rule~\Rule{$\to$I}, in that they introduce in the typing environment a
fresh type whose variables must \emph{not} be instantiated. In our example,
however, in one premise, the rule introduces $x:\alpha$ in the typing
environment which can, for instance, be instantiated by the $[\textrm{Inst}]$ rule. In
the second premise, it introduces $x:\lnot\alpha$ which can also be instantiated
in a \emph{completely different way}. In other words, the correlation between the two
occurrences of the same variable $\alpha$ is lost, which amounts to commuting
the (implicit) universal quantification with the $\vee$ type connective,
yielding a non-prenex polymorphic type
$(\forall\alpha.\alpha)\vee(\forall\alpha.\lnot\alpha)$. To avoid this
unsound situation, we need to ensure that when a type is
split between two components of a union, no polymorphic variable is introduced.
This is achieved by rule $[\vee]$ which requires the type $s$ of $e'$ to be
split as $s\equiv (s\wedge\mt)\vee(s\wedge\lnot\mt)$ (here is our hidden union).

The top-level definitions of a program are typed sequentially by two
specific rules:
\begin{mathpar}
  \Infer[TopLevel-Expr]
  { \Gamma\vdash e : t }
  { \Gamma\vdashp e : t\grenaming }
  { \grenaming\disjoint\Gamma }
  \qquad
  \Infer[TopLevel-Let]
  { \Gamma\vdashp e : t\\\Gamma,x:t\vdashp p:t' }
  { \Gamma\vdashp \tletexp{x}{e}{p} : t' }
  { }
\end{mathpar}
where $\grenaming$ denotes a generalization, that is a substitution transforming
monomorphic variables into polymorphic ones and where
$\grenaming\disjoint\Gamma{\iffdef}\dom\grenaming\cap\vars{\Gamma}=\emptyset$.

After typing an expression used for a definition, its type is generalized
(Rule~\Rule{Toplevel-Expr}) before being added in the environment
(Rule~\Rule{Toplevel-Let}). Note that this is the only place where
generalization takes place: no rule in the type system for expressions (Figure~\ref{fig:declarative})
allows the generalization of a type variable.  As explained at the
beginning of Section~\ref{sec:discon}, this is not a limitation, since intersection types are more powerful than HM
polymorphism, and top-level generalization is of practical importance
since it is necessary to the modularity of type-checking.
Nevertheless, the core of our inference system is given only by the rules
in Figure~\ref{fig:declarative} for expressions: the above
``\textsc{TopLevel}'' rules are only useful to inhabit variables of
the typing environments used in the rules for expressions, and this makes
it possible to close the expressions being typed. For instance, if a typing
derivation for an expression $e$ is deduced, say, under the hypothesis
$x:\alpha{\to}\alpha$ (with $\alpha$ polymorphic), then it is possible
to obtain a closed program by inhabiting $x$ by a definition like
$\tletexp x {\lambda y.y\text{ ;}...}{e}$.  This is the reason why,
henceforth, we mainly focus on the typing of expressions.

The type system is sound (all proofs for this work are given in Appendix~\ref{sec:proofs-appendix}):
  \begin{theorem}[Soundess]\label{th:typesound} If $\varnothing\vdashp p:\tau$, then either
    $p$ diverges or $p\reducesprog v$ with $v\in \tau$.
  \end{theorem}

\section{Algorithmic System}
\label{sec:algo}
As discussed in the introduction, the declarative type system is not syntax directed and some rules are not analytic.
In order to make it algorithmic, we first introduce in Section~\ref{sec:msc}
a \textit{canonical form} for expressions that adds syntactic constructions (\textit{bindings})
to indicate when to apply the union elimination rule and on which sub-expression.
Then, in Section~\ref{sec:algosystem}, we define a fully algorithmic type system
that takes a \textit{canonical form} together with an \textit{annotation}
and produces a type.

\subsection{MSC Forms}\label{sec:msc}

\subsubsection{Canonical Forms}\label{sec:canonical}
The \Rule{$\vee$} rule is not syntax directed since it can be applied on any expression
and can split the type of any of its subexpressions. If we want an algorithmic type system,
we need a syntactic way to determine when to apply this rule, and which subexpression to split.
In order to achieve this, we represent our terms with a syntax
called \textit{Maximal Sharing Canonical Form} (MSC Form)
introduced by \citet{occtyp22}. Let us start by defining
the \emph{canonical forms}, which are expressions produced by the following grammar:
\begin{equation*}
  \begin{array}{lrclr}
    \textbf{Atomic expressions} &a &::=& c\alt x\alt \lambda x.\kappa\alt (\xx,\xx)\alt \xx \xx\alt \pi_i \xx\alt \tcase{\xx}{\tau}{\xx}{\xx}\\
    \textbf{Canonical Forms} & \kappa &::=& \xx\alt \bindexp{\xx}{a}{\kappa} \\
  \end{array}
\end{equation*}
Canonical forms, ranged over by $\kappa$, are \textit{binding variables} (noted $\xx$, $\yy$, or $\zz$)
possibly preceded by a list of definitions (from binding variables to atoms). Atoms are either
a variable from a $\lambda$-abstraction (noted $x$, $y$, or $z$), or a
constant, or a $\lambda$-abstraction whose
body is a canonical form, or any other expression in which \emph{all} proper
sub-expressions are binding variables.
An expression in canonical form without any free binding variable can be transformed into an expression of the source
language using the unwinding operator $\eras{.}$ that basically inlines bindings:
$\eras {\bindexp{\xx}{a}{\kappa}} = \eras{\kappa}\subs{\xx}{\eras{a}}$
(see Appendix~\ref{app:unwind} for the full definition).
The inverse direction, that is, producing from a source language
expression a canonical form that unwinds to it, is straightforward (see
Appendix~\ref{app:sl2cf}). However for each expression of the source
language there are several canonical forms that unwind to it. For our
algorithmic type system we need to associate each source language
expression to a unique canonical form, as we define next.

\subsubsection{Maximal Sharing Canonical Forms}\label{sec:mscf}

We define a congruence on canonical forms and atoms:
\begin{definition}[Canonical equivalence]\label{def:order}
  We denote by $\,\eqcan$ the smallest congruence on canonical forms and atoms
  that is closed by $\alpha$-conversion and such that\\[.5mm]
\hspace*{4mm}
$\bindexp{\xx_1}{a_1}{\bindexp{\xx_2}{a_2}{\kappa}}
\ \eqcan\ 
\bindexp{\xx_2}{a_2}{\bindexp{\xx_1}{a_1}{\kappa}}
  \hfill\xx_1{\notin}\fv{a_2},\xx_2{\notin}\fv{a_1}$
\end{definition}
\noindent
To infer types for the source
language, we single out canonical forms satisfying three properties:
\begin{definition}[MSC Forms]\label{def:maximalsharing}
  A \emph{maximal sharing canonical form} (abbreviated as
  \emph{MSC-form}) is (any canonical form $\alpha$-equivalent to) a canonical form $\kappa$ such
  that:
  \begin{enumerate}[leftmargin=12pt,labelsep=2pt]
    \item if  $\bindexp{\xx_1}{a_1}{\kappa_1}$ and
      $\bindexp{\xx_2}{a_2}{\kappa_2}$ are distinct
      sub-expressions of $\kappa$, then $a_1\not\eqcan a_2$
    \item if $\lambda x.\kappa_1$ is a sub-expression of
      $\kappa$ and
      $\bindexp{\yy}a{\kappa_2}$ a sub-expression of $\kappa_1$, then
      $\fv a\,{\not\subseteq}\;\fv{\lambda x.\kappa_1}$
    \item  if $\bindexp{\xx}a{\kappa'}$ is a sub-expression of
      $\kappa$, then $\xx\in\fv{\kappa'}$.
  \end{enumerate}
\end{definition}
MSC-forms are defined modulo
$\alpha$-conversion.\footnote{For
instance, both $\lambda x.\bindexp{\xx}{x}{\bindexp{\zz}{\xx\yy}{\bindexp{\zz'}{\zz\yy}{\zz'}}}$ and
$\lambda x.\bindexp{\xx}{x}{\bindexp{\zz}{\xx\yy}{\zz}}$ are two distinct atoms that can occur in the same
MSC-form, even though the atom $\xx\yy$ appears in both: an
$\alpha$-renaming of $\xx$ makes the first MSC-property hold.}
The first condition states that distinct variables denote different
definitions,
\textcolor{modification}{that is, it enforces the \emph{maximal sharing} of common
sub-expressions.}
The second condition requires bindings to extrude $\lambda$-abstractions whenever possible.
The third condition states that there is no useless bind (bound variables must occur in the body of the binds).

The key property of MSC-forms is that given an expression $e$
of the source language, all its MSC-forms (i.e., all MSC-form whose
unwinding is $e$) are equivalent:
\begin{proposition}\label{th:mscfequiv}
  If $\kappa_1$ and $\kappa_2$ are two MSC-forms and $\eras {\kappa_1}\aequiv\eras{\kappa_2}$, then $\kappa_1\eqcan\kappa_2$.
\end{proposition}
We denote the unique MSC-form whose unwinding is $e$ by $\MSCF e$.
It is easy to transform a canonical form into a MSC-form that has the same unwinding.
The reader can refer to Appendix~\ref{sec:msc-appendix} for a set of rewriting rules implementing this operation.

\subsection{Algorithmic Typing Rules}\label{sec:algosystem}

MSC-forms tell us when to apply the \Rule{$\vee$} rule: a term
$\bindexp\xx a\kappa$ means (roughly) that it must be typed by applying the union rule to
the expression $\kappa\subs \xx a$. Putting an expression into its
MSC-form to type it, thus corresponds to applying
the \Rule{$\vee$} rule on every occurrence of every subexpression of the
original expression. This is a step toward a
syntax-directed type system. However, there are still two issues
to solve before obtaining an algorithmic type system:
$(i)$ rules \Rule{$\land$}, \Rule{Inst}, and \Rule{$\leq$} are still not syntax-directed, and
$(ii)$ rules \Rule{$\vee$},  \Rule{Inst},  \Rule{$\to$I},
and \Rule{$\leq$} are not analytic,
meaning that some of their premises cannot be deduced just by looking at the conclusion:
the \Rule{$\vee$} rule requires guessing a type decomposition (i.e.,
the monomorphic type $\mt$  in the premises), the \Rule{Inst} rule
requires guessing a substitution, the \Rule{$\to$I} rule  requires  guessing
the domain $\mt$ of the function, and the  \Rule{$\leq$} rule requires
guessing the type $t'$ to subsume to.

The issue of \Rule{Inst} and \Rule{$\leq$} not being syntax directed
can be solved by
embedding them in some structural rules (in
particular, in the rules for destructors).
Moreover, as we will see later, the rules in which we embed \Rule{$\leq$} can be made analytic by using some type operators.
As for rule \Rule{$\land$}, making it syntax-directed, is trickier. Indeed, the usual
approach of merging rules \Rule{$\to$I} and \Rule{$\land$} does not work here, since
terms in MSC-forms may hoist a bind definition outside the function where they
are used, causing rule \Rule{$\land$} to be needed on a term that is not,
syntactically, a $\lambda$-abstraction.
Lastly, there is no easy way to guess the substitutions used by \Rule{Inst} rules,
or the domain used in \Rule{$\to$I} rules, or
the decompositions performed by \Rule{$\vee$} rules. To tackle these issues,
our algorithmic type system will not only take a canonical form as input, but also
an annotation that will $(i)$ indicate when to apply an intersection, and $(ii)$ indicate
which type decomposition (for \Rule{$\vee$} rules) and which type
substitutions (for \Rule{Inst} rules) to use.
Formally, our algorithmic system uses judgements of the form
$\Gamma\vdashA\ea \kappa \annot : t$ for a canonical form $\kappa$, and
$\Gamma\vdashA\ea a \annota : t$ for an atom $a$ where $\annot$ and
$\annota$ are respectively \emph{form annotations} and \emph{atom annotations} defined as follows:
\begin{equation*}\label{eq:algolang}
  \hspace*{-2mm}  \begin{array}{lrclr}
      \textbf{Atom annots}  &\annota &::=& \annotconst 
      \alt \annotlambda{\mt}{\annot}\alt\annotpair\renaming\renaming\alt\annotapp\Sigma\Sigma\alt\annotproj\Sigma
      \alt\annotempty\Sigma\alt\annotthen\Sigma\alt\annotelse\Sigma\alt\annotinter{\{\annota,...\,,\annota\}}\\
      \textbf{Form annots}  &\annot &::=& \annotvar\renaming\alt
      \annotbind{\annota}{\{(\mt,\annot), \dots, (\mt,\annot)\}}\alt\annotskip{\annot}\alt\annotinter{\{\annot,\dots,\annot\}}
      \end{array}
  \end{equation*}
where $\renaming$ ranges over \emph{renamings} of polymorphic variables, that
is, injective substitutions from $\polyvars$ to $\polyvars$, and
$\Sigma$ ranges over instantiations, that is, sets of substitutions
from $\polyvars$ to $\types$.
We chose to keep annotations separate from the terms, instead of embedding them in the
canonical forms, since in the latter case it would be more complicated
to capture the tree structure of the derivations.

The algorithmic system is defined by the rules given in
Appendix~\ref{app:algosys}. Below we comment the most interesting
rules (we just omit the rules for constants, variables and two rules
for type-cases). Essentially, there is one typing rule for each annotation,
the only exception being the $\varnothing$ annotation that is used
both in the rule to type constants and in the two rules for variables.
\begin{mathpar}
  \Infer[$\to$I\AR]
  {\Gamma,x:\mt\vdashA \ea\kappa{\annot}:t}
  {\Gamma\vdashA\ea{\lambda x.\kappa}{\annotlambda{\mt}{\annot}}:\arrow{\mt}{t}}
  { }
\end{mathpar}
To type the atom $\lambda x.\kappa$, the annotation
$\annotlambda{\mt}{\annot}$ provides the domain $\mt$ of the function,
and the annotation $\annot$ for its body.
\begin{mathpar}
\hspace*{15mm}\Infer[$\to$E\AR]
{ }
{\Gamma \vdashA \ea{{\xx_1}{\xx_2}}{\annotapp{\Sigma_1}{\Sigma_2}}: t_1\circ t_2 }
{ \stackedguards{t_1=\Gamma(\xx_1)\Sigma_1,\,\,\,t_2=\Gamma(\xx_2)\Sigma_2}
  {t_1\leq\Empty\to\Any,\,\,\,t_2\leq\dom{t_1}} }
\end{mathpar}
To type an application one  must apply an instantiation and a
subsumption to both the type of the function and the type of the
argument. Instantiations (i.e., $\Sigma_1$ and $\Sigma_2$) are sets of
type substitutions; their application to a type $t$ is defined as
$t\Sigma\eqdeftiny \textstyle\bigwedge_{\sigma\in\Sigma}
t\sigma$.
Since they
cannot be directly
guessed, they are given by the annotation. Subsumption instead  is embedded
in two type operators. A first operator, $\dom{}$, computes the domain
of the arrow and is used to check that the application is
well-typed. A second type operator, $\circ$, computes the type of the result
of the application.  These type operators are defined as follows:
\(\dom t \eqdef \max \{ u \alt t\leq u\to \Any\}\) and
\(\apply t s \eqdef\min \{ u \alt t\leq s\to u\}\).
\begin{mathpar}
  \Infer[$\times$E$_1$\AR]
  { ~}
  {\Gamma \vdashA \ea{\pi_1 \xx}{\annotproj\Sigma}:\bpi_1(t)}
  { \stackedguards{t=\Gamma(\xx)\Sigma}
    {t\leq(\Any{\times}\Any)} }
  \quad
  \Infer[$\times$E$_2$\AR]
  { }
  {\Gamma \vdashA \ea{\pi_2 \xx}{\annotproj\Sigma}:\bpi_2(t)}
  { \stackedguards{t=\Gamma(\xx)\Sigma}
    {t\leq(\Any{\times}\Any)} }
\end{mathpar}
The rules for projections \Rule{$\times$E$_1$\AR} and \Rule{$\times$E$_2$\AR}
follow the same idea as the rule for application \Rule{$\to$E\AR},
with the use of two type operators \(\bpl t\eqdef \min \{ u \alt t\leq \pair u\Any\}\) and
\(\bpr t\eqdef \min \{ u \alt t\leq \pair \Any u\}\).
All these type operators can be effectively computed (cf.\
Appendix~\ref{sec:typeop-appendix}).

\begin{mathpar}
  \Infer[$\times$I\AR]
  { }
  {\Gamma \vdashA \ea{(\xx_1,\xx_2)}{\annotpair{\renaming_1}{\renaming_2}}:\pair {t_1} {t_2}}
  { t_1=\Gamma(\xx_1)\renaming_1,\,\,\,t_2=\Gamma(\xx_2)\renaming_2 }
\end{mathpar}
To type a pair $(\xx_1,\xx_2)$ it is not necessary to instantiate
$\Gamma(\xx_1)$ or $\Gamma(\xx_2)$. However, to avoid unwanted
correlations, it is necessary to rename the
polymorphic type variables of its components.
For instance, when typing the pair $(\xx,\xx)$ with $\xx:\polyvar\to\polyvar$,
it is better to type it with $(\polyvar\to\polyvar,\polyvarb\to\polyvarb)$ rather than
$(\polyvar\to\polyvar,\polyvar\to\polyvar)$, since the former type has strictly more instances
than the latter.

\begin{mathpar}
  \Infer[$\in_1$\AR]
  { }
  {\Gamma\vdashA \ea{\tcase {\xx} \tau {\xx_1}{\xx_2}}{\annotthen\Sigma}: \Gamma(\xx_1)}
  { \Gamma(\xx)\Sigma\leq \tau }
\end{mathpar}
To type type-cases, the annotation indicates which of the three rules
must be applied (here \Rule{$\in_1$}) and how to instantiate the
polymorphic type variables occurring in the type of the tested expression, so that
it satisfies the side condition of the applied rule (see
also \Rule{$\in_2$\AR} and  \Rule{$\Empty$\AR} in Appendix~\ref{app:algosys}). 
\begin{mathpar}
  \Infer[Bind$_1$\AR]
  {\Gamma\vdashA \ea{\kappa}{\annot}:t}
  {
  \Gamma\vdashA\ea{\bindexp {\xx} {a} {\kappa}}{\annotskip \annot}:t
  }
  { \xx\not\in\dom\Gamma }
\end{mathpar}
In rule \Rule{Bind$_1$\AR} the annotation indicates to skip the definition of
the current binding. This rule is used when the binding variable is not required
for typing the body $\kappa$ under the current context $\Gamma$. For instance,
this is the case when $\xx$ only appears in a branch of a typecase that cannot
be taken under $\Gamma$. The side condition $x\not\in\Gamma$ prevents a
potential unsound name conflict between binding variables: as occurrences of
$\xx$ in $\kappa$ denote the $\xx$ binding variable that is being skipped,
having the type of a former binding variable $\xx$ in our environment when
typing $\kappa$ would be unsound.
\begin{mathpar}
  \Infer[Bind$_2$\AR]
  {\Gamma\vdashA \ea a{\annota}:s\\
  {\small(\forall i\in I)}\quad\Gamma,\xx:s\wedge\mt_i\vdashA \ea\kappa{\annot_i}:t_i
  }
  {
    \Gamma\vdashA\ea{\bindexp {\xx} {a} {\kappa}}{\annotbind {\annota} {\{(\mt_i,\annot_i)\}_{i\in I}}} : \textstyle\bigvee_{i\in I}t_i
  }
  { \textstyle\bigvee_{i\in I}\mt_i\simeq\Any }
\end{mathpar}
This rule tries to type the bound atom and then decomposes its
type according to the annotation. This decomposition corresponds to an
application of the \Rule{$\vee$} rule of the declarative type system
with the only difference that the type $s$ of the atom is split in several summands
by intersecting it with the various $\mt_i$ (instead of just two summands as in
the rule \Rule{$\vee$}) whose union covers $\Any$.

Finally, two annotations indicate when and how
to apply rule \Rule{$\wedge$} to atoms and canonical forms:
\begin{mathpar}
  \Infer[$\land$\AR]
  {
  {\small(\forall i\in I)}\quad\Gamma\vdashA \ea a{\annota_i}:t_i
  }
  {
    \Gamma\vdashA\ea{a}{\annotinter{\{\annota_i\}_{i\in I}}} : \textstyle\bigwedge_{i\in I}t_i
  }
  { I\neq\emptyset }
  \quad
  \Infer[$\land$\AR]
  {
  {\small(\forall i\in I)}\quad\Gamma\vdashA \ea\kappa{\annot_i}:t_i
  }
  {
    \Gamma\vdashA\ea{\kappa}{\annotinter{\{\annot_i\}_{i\in I}}} : \textstyle\bigwedge_{i\in I}t_i
  }
  { I\neq\emptyset }
\end{mathpar}
An expression $e$ is typable if and only if its unique (modulo $\eqcan$) MSC-form
is typable, too:
\begin{theorem}[Soundness and Completeness]\label{th:main}
  For every term $e$ of the source language
  \begin{align*}
    \hspace*{4cm}
    \vdash e:t\quad &\Rightarrow\quad\exists\annot\,\vdashA\ea{\MSCF e}{\annot}:t'\leq t&\text{(completeness)}\\
    \vdash e:t\quad &\Leftarrow\quad\;\vdashA\ea{\MSCF e}{\annot}:t&\text{(soundness)}
  \end{align*}
\end{theorem}
It is easy to generate
the unique MSC-form associated to a source language expression $e$
(cf. Appendix~\ref{sec:msc-appendix}).
Theorem~\ref{th:main} states that this MSC-form
is typable if and only if $e$ is: we reduced the problem of typing $e$
to the one of finding an annotation that makes the unique MSC-form of $e$
typeable with the algorithmic type system.\footnote{%
\color{modification}As stated by Theorem~\ref{th:main} the transformation
of an expression into its MSC-form preserves typing. However,
intuitively, it does not preserve the reduction semantics, since
bindings regroup different occurrences of some sub-expression that in
the original expression might be evaluated at different stages of the
reduction or not evaluated at all. We said ``intuitively'' because
no operational semantics is defined for canonical forms (this was
alredy the case for~\cite{occtyp22}).} 
\begin{figure*}
\input{fig-ex-annot.tex}\ifspacetrick\vspace{-3mm}\fi
  \caption{MSC-form and two annotations for $\lambda x. \tcase {f x}\Int{g (f x)} {x}$
  \label{fig:exannot}}
  \end{figure*}
Figure~\ref{fig:exannot} gives an example of an MSC-form and two possible 
annotations for it. The term  ``$\lambda x. \tcase {f x}\Int{g (f x)}
  {x}$'' (where $f:\forall\alpha.\alpha\To\alpha$ and $g:\Int\To\Int$) is
put in MSC-form (on the left). In the first annotation, the function is typed
with a single $\lambda$ annotation (line 3). The interesting part is the annotation of
the binding for $\bvar{u}$ (line 5): the corresponding \texttt{keep} annotation
represents an application of the union elimination rule on the occurrences of
the expression $f x$ whose type $\monovarb$ is split into $\monovarb{\wedge}\Int$
(line 6) and $\monovarb{\setminus}\Int$ (line 9). Each subcase is annotated
  accordingly. Notice in the second subcase that
the annotation for $\bvar{v}$ is \texttt{skip} (line 10) 
which indicates that this
particular variable must not be used (as $g(fx)$ cannot be typed since in the
``else'' branch, $f x$ has type $\lnot\Int$). A different annotation, yielding a
better type, is the one on the right. This intersection annotation (line 2) separates the domain of
the $\lambda$-abstraction into two cases, each typed independently, yielding for
the whole function an intersection type.

{\color{modification}
The example in Figure~\ref{fig:exannot} also shows why the condition of maximal sharing for our
forms is necessary, not only for their uniqueness, but also for
the completeness of the algorithmic system: if the two occurrences of $f x$ in ``$\lambda x. \tcase {f x}\Int{g (f x)}
  {x}$'' were not bound by the same variable (as in the leftmost
  column of line 5 in
  Figure~\ref{fig:exannot}), viz., if the sharing were not maximal,
  then it would not be possible to deduce that $g(f x)$ is well typed:
  $g$ expects an integer, but without maximal sharing it is not
  possible to deduce that the occurrence of $f x$ in the first branch
  is indeed of type $\Int$.

The problem of inferring an annotation for an MSC-form as the above---in particular
the rightmost (more precise) annotation in Figure~\ref{fig:exannot}---is tackled in the next section.
}

\section{Reconstruction}
\label{sec:reconstruction}

This section describes an algorithm to find an annotation for an
expression in MSC-form, such that the pair  expression and annotation is
typable in the algorithmic system.
Though this algorithm is not complete,
it is sound and terminating \textcolor{modification}{(see
Section~\ref{sec:propreco} for the formal statements and a discussion
about incompleteness)}. Experimental results are presented in Section~\ref{sec:implementation}.

The annotation reconstruction algorithm is composed of two systems of deduction rules:
the \textit{main reconstruction algorithm} (Section~\ref{sec:mainreconstruction})
which produces intermediate annotations containing information about the domains of $\lambda$-abstractions
and the type decompositions to use in bindings, and the \textit{auxiliary reconstruction algorithm}
(Section~\ref{sec:auxreconstruction}) which converts these intermediate annotations
into annotations for the algorithmic type system, by computing instantiations
$\Sigma$ for destructors.

\subsection{The Tallying Algorithm}\label{sec:tallying}

One key ingredient used by the reconstruction algorithm is the \textit{tallying} algorithm.
Roughly, \emph{tallying} is the equivalent of the \textit{unification}
used in algorithm $\mathcal{W}$~\cite{DM82}, but for a type system with subtyping.
The tallying algorithm was introduced by~\citet{polyduce2} to solve the following problem:
given a set of pairs $\{(t_i,t_i')\}_{i\in I}$ and a set of type variables $\Delta$ representing the monomorphic type variables,
find all substitutions $\sigma$ whose domain is disjoint from $\Delta$ (noted $\sigma\disjoint\Delta$) and that satisfy $\forall i\in I.\ t_i\sigma\leq t_i'\sigma$.
\citet{polyduce2} show that this problem is decidable and give an algorithm to characterize all solutions.
As for unification, for each instance of the tallying problem there is either no solution or
several substitutions each of which is a solution of the problem. The difference
is that while with unification all solutions are characterized by a
principal substitution, with tallying they are characterized by a
principal finite \emph{set} of substitutions.\footnote{This is due to
the presence of the empty type. For instance, the principal solution of unifying
$\alpha\times\beta$ with $s\times t$ is the substitution
$\{\alpha\tvarto s,\beta\tvarto t\}$, while all substitutions that make the former type
become a subtype of the latter are characterized by a set containing three
distinct substitutions: $\{\alpha\tvarto\Empty\}$, $\{\beta\tvarto\Empty\}$, and $\{\alpha\tvarto s,\beta\tvarto t\}$.}
More precisely, all substitutions that are  solutions to a tallying instance are characterized
by a \emph{principal} set  $\Sigma$ of substitutions,
such that every $\sigma\in\Sigma$ is a solution, and for any solution $\sigma$, we have
$\exists\sigma_1{\in}\Sigma.\ \exists \sigma_2.\ \sigma_2\disjoint\Delta \text{ and }\sigma\simeq\sigma_2\circ\sigma_1$,
where $\circ$ denotes the composition of substitutions and $\simeq$ is
pointwise type equivalence.

In this work, all tallying instances use a single constraint, and we will note
$\tallyf{t_1}{t_2}$ the set of substitutions $\Sigma$ characterizing all the solutions
of the tallying instance $\{(t_1,t_2)\}$, where $\Delta=\monovars$
and thus $\forall\sigma\in\Sigma.\ \dom\sigma\subseteq\polyvars$ \textcolor{modification}{(we
use the symbol $\dot\leq$, rather than $\leq$ to stress that it denotes a
constraint to be solved, rather than a pair in the subtyping relation)}.

The tallying function $\tallyff$ finds substitutions for polymorphic type variables,
but in order to infer the domain of $\lambda$-abstractions, we may need to find substitutions for
monomorphic type variables. We thus introduce an additional tallying function,
$\tallyinferf{t_1}{t_2}$:

\begin{definition}
Let $\restr{\sigma}{X}$ denote the restriction of the substitution $\sigma$ to the domain $X$. We define
\[\tallyinferf{t_1}{t_2}=\{\restr{(\sigma\circ\sigma'\circ\grenaming)}{\monovars} \,\alt\, \sigma'\in\tallyf{\fresh({t_1})\grenaming}{\fresh({t_2})\grenaming}\}\]
where $\fresh(t)$ denotes the type $t$ where polymorphic type
  variables have been substituted by fresh ones; $\grenaming$ is a
  renaming from $(\vars{t_1}\cup\vars{t_2})\cap\monovars$ to fresh
  polymorphic variables; and $\sigma$ is a substitution mapping polymorphic variables
  appearing in the image of $\sigma'\circ\grenaming$ to fresh monomorphic variables.
\end{definition}
In a nutshell, polymorphic type variables in $t_1$ and $t_2$ are refreshed in order to
decorrelate them, and monomorphic type variables are generalized using $\grenaming$ so that
$\tallyff$ is allowed to find solutions for them.
Each solution $\sigma'$ is composed with $\grenaming$
in order to restore the connection with the initial monomorphic type variables,
and the polymorphic type variables in the image of the resulting substitution are
transformed into monomorphic ones by composing $\sigma$ with it. Finally,
the substitution is restricted to $\monovars$ (i.e., to the domain of
$\grenaming$).

For example, an instance such as $\tallyinferf{\Int{\land}\polyvar\to\Int{\land}\polyvar}{\monovarb\to\polyvar}$
can be generated during reconstruction, when a function of type $\Int{\land}\polyvar\to\Int{\land}\polyvar$
is applied to an argument of type $\monovarb$, but the $\polyvar$ on the
right-hand of $\dot\leq$ is unrelated to the one on the left-hand
side. Decorrelating them yields a unique solution
$\{ \monovarb\tvarto \monovarb'{\land}\Int \}$, that is, $\monovarb$ must be substituted by
$\monovarb'{\land}\Int$ in our context for the  application to be typeable.

\subsection{Main Reconstruction Algorithm}\label{sec:mainreconstruction}

The \emph{main reconstruction algorithm}, defined in this section,
infers the domains of $\lambda$-abstractions and the
decompositions of types into disjoint unions to use for bindings.
It works by successively refining \emph{intermediate annotations} defined below.
These intermediate annotations store information about the domains of $\lambda$-abstractions and the decompositions of bindings.
However, the instantiations $\Sigma$ used to type destructors (i.e.,
applications, projections, and typecases) in the algorithmic type system
are not stored in intermediate annotations, because they might get invalidated as the reconstruction progresses:
when  new information is found about the domain of a lambda or the decomposition of a binding,
the algorithm will retype some intermediate definitions of the MSC-form,
thus invalidating the instantiations $\Sigma$ of later definitions.
Thus, these instantiations $\Sigma$ will be recomputed whenever needed, using the auxiliary system
(Section~\ref{sec:auxreconstruction})
that converts intermediate annotations into annotations for the algorithmic type system.

Atom and form \emph{intermediate annotations} are defined by the grammar below:\ifspacetrick\vspace{-.5mm}\fi
\begin{equation*}
  \begin{array}{lrclr}
  \textbf{Split annotations}  &\sanns &::=& \{(\mt,\banns),\dots,(\mt,\banns)\}\\
  \textbf{Atom intermediate annot.}  &\lanns &::=& \annotiinfer
  \alt\annotiuntyp\alt\annotityp
  \alt\annotiinter{\{\lanns,\dots,\lanns\}}{\{\lanns,\dots,\lanns\}}\\&&&
  \alt\annotithen\alt\annotielse\alt\annotilambda{\mt}{\banns}\\
  \textbf{Form intermediate annot.}  &\banns &::=& \annotiinfer\alt\annotiuntyp\alt\annotityp
  \alt\annotiinter{\{\banns,\dots,\banns\}}{\{\banns,\dots,\banns\}}\\&&&
  \alt\annotitryskip\banns\alt\annotitrybind{\lanns}{\banns}{\banns}\\&&&
  \alt\annotiprop{\lanns}{\Gammas}{\sanns}{\sanns}
  \alt\annotiskip{\banns}\alt\annotibind{\lanns}{\sanns}{\sanns}
  \end{array}
\end{equation*}
%
In the following, we use the metavariable $\kappaa$ to range over both atoms and
expressions (i.e., $\kappaa ::= a\alt\kappa$).
Similarly, the metavariable $\aannot$ ranges over atom annotations $\annota$ and form annotations $\annot$
 (i.e., $\aannot ::= \annota\alt\annot$); while
the metavariable $\aanns$ ranges over atom intermediate annotations $\lanns$ and form intermediate annotations $\banns$
 (i.e. $\aanns ::= \lanns\alt\banns$).

Let $\msubst$ range over \emph{monomorphic substitution}, that is, substitutions from $\monovars$
to monomorphic types, and $\Msubst$  range over finite sets of monomorphic
substitutions ($\Msubst ::= \{\msubst, \dots, \msubst\}$).
The main reconstruction algorithm is presented as a deduction rule system,
for judgments of the form $\Gamma\raavdash\epa{\kappaa}{\aanns}\refines\results$,
where $\results$ is a result defined as follows:
\begin{equation*}
  \begin{array}{lrclr}
    \textbf{Result}  &\results &::=& \resok{\aanns}\alt\resfail{}
    \alt\respart{\Gamma}{\aanns}{\aanns}
    \alt\ressubst{\Msubst}{\aanns}{\aanns}
    \alt\resvar{\xx}{\aanns}{\aanns}
  \end{array}
\end{equation*}
Let us see what each result for $\Gamma\raavdash\epa{\kappaa}{\aanns}$ means:
\begin{description}[leftmargin=10pt]
\item[$\resok{\aanns'}$:] the reconstruction was successful
  and  $\kappaa$ can be typed by the algorithmic type system using the annotation $\aanns'$
  (after converting it into an annotation $\aannot$ using the auxiliary reconstruction system).
  This result is terminal (i.e., it is a definitive answer that cannot
  be further refined).
\item[$\resfail$:] the reconstruction has failed. The algorithm was
  not able to find an annotation that makes $\kappaa$ typable with the
  algorithmic system.
  This result is terminal.
  \item[$\ressubst{\Msubst}{\aanns_1}{\aanns_2}$:] the reconstruction
  found a set of substitutions  $\Msubst$ that if applied to $\Gamma$
  may make $\kappaa$ typable.
  In practice, for each substitution $\msubst\in\Msubst$,
  the reconstruction will be called again on the environment $\Gamma\msubst$ and annotation $\aanns_1\msubst$.
  However, this does not necessarily mean that the reconstruction will fail on the current environment $\Gamma$: $\kappaa$ might still be typeable but with a less precise type (e.g., it could yield an arrow type
  with a smaller domain). Thus, 
  this \textit{default} case which does not instantiate $\Gamma$
  is also explored, using the annotation $\aanns_2$ instead of $\aanns_1$.
\item[$\respart{\Gamma'}{\aanns_1}{\aanns_2}$:] the reconstruction
  found some splits for the variables in $\dom{\Gamma'}$ that if applied to $\Gamma$
  may make $\kappaa$ typable. In practice, the system  generates
  several new environments: one is obtained by (pointwise) intersecting 
  $\Gamma$ with $\Gamma'$ and then it is used to retype $\kappaa$ with
  the annotation $\aanns_1$; the others are obtained by intersecting
  $\Gamma$ with all the possible
  pointwise negations of $\Gamma'$ and then they are  used to retype $\kappaa$ with
  the annotation $\aanns_2$. 
\item[$\resvar{\xx}{\aanns_1}{\aanns_2}$:] the reconstruction found
  that in order to type $\kappaa$,
  the definition of the bind-abstracted variable $\xx$ should be typed.
  Any branch that successfully types it continues with the annotation $\aanns_1$,
  otherwise it continues with the annotation $\aanns_2$.
\end{description}
Initially, any form or atom $\kappaa$ is annotated with $\annotiinfer$, and this annotation
is then refined until it yields a terminal result (i.e., either $\resok{}$ or
$\resfail$).
The rules below are presented by decreasing priority (i.e., the first rule that applies is used).
Some rules have been omitted for concision, but the reader can find the full reconstruction system
in Appendix~\ref{sec:mainreconstruction-appendix}.

There are two different forms of judgments: $\Gamma\raavdash\epa{\kappaa}{\aanns}\refines\results$
and $\Gamma\aavdash\epa{\kappaa}{\aanns}\refines\results$.
We first define rules for the judgment $\aavdash$ for every canonical form and atom.
The results of these judgments are not necessarily terminal and, therefore, it may be necessary
to call the reconstruction again in order to refine them.
This is the purpose of $\raavdash$ judgments which call repetitively $\aavdash$ judgments
when relevant, so that in the end we get a terminal result.
Let us first focus on $\aavdash$ judgments.

\begin{mathpar}
  \Infer[Ok]
  { }
  {\Gamma\aavdash \epa \kappaa\annotityp\refines\resok{\annotityp}}
  { }
  \qquad
  \Infer[Fail]
  { }
  {\Gamma\aavdash \epa \kappaa\annotiuntyp\refines\resfail}
  { }
\end{mathpar}
If a canonical form or atom $\kappaa$ is annotated with $\annotityp$,
then reconstruction is finished for $\kappaa$, and it is typeable
in the current context $\Gamma$.
The annotation $\annotityp$ is never used on $\lambda$-abstractions and bindings because
the system needs to store more information for them.
Likewise, if a form or atom $\kappaa$ is annotated with $\annotiuntyp$,
then reconstruction is finished for $\kappaa$ by failing
in the current context.
\begin{mathpar}
  \Infer[AxOk]
  { x\in\dom\Gamma }
  {\Gamma\aavdash \epa x{\annotiinfer}\refines\resok{\annotityp}}
  { }
  \qquad
  \Infer[AxFail]
  { }
  {\Gamma\aavdash \epa x{\annotiinfer}\refines\resfail}
  { }
\end{mathpar}
If a $\lambda$-abstracted variable $x$
is in the environment, then it is typeable and thus the algorithm returns $\resok{\annotityp}$.
Otherwise, $x$ is undefined and $\resfail$ is returned.
\begin{mathpar}
  \Infer[AppVar$_i$]
  { \xx_i\not\in\dom\Gamma }
  {\Gamma\aavdash \epa{\xx_1 \xx_2}{\annotiinfer}\refines
  {\resvar{\xx_i}{\annotiinfer}{\annotiuntyp}}}
  {}
\end{mathpar}
To type the application $\xx_1 \xx_2$, we must first ensure that
  $\{\xx_1,\xx_2\}\subseteq\dom\Gamma$. If it is not the case, then
  the two rules \Rule{AppVar$_i$} (for $i=1,2$) try to remedy it by returning $\resvar{\xx_i}{\annotiinfer}{\annotiuntyp}$, 
  which is the result that asks the system to try to type the atom bound to
  $\xx_i$ for $\xx_i\not\in\dom\Gamma$. If the attempt is successful,
then the algorithm will continue the reconstruction for the
  application with the annotation $\annotiinfer$ and $\xx_i\in\dom\Gamma$,
otherwise it will continue with the annotation $\annotiuntyp$ making
  the reconstruction fail on this application.
\begin{mathpar}
  \Infer[AppInfer]
  {
      \tallyinfer{\Msubst}{\Gamma(\xx_1)}{\Gamma(\xx_2)\to \polyvar}
  }
  {\Gamma\aavdash \epa{{\xx_1}{\xx_2}}{\annotiinfer}\refines
  {\ressubst{\Msubst}{\annotityp}{\annotiuntyp}}}
  { \polyvar\in\polyvars\text{ fresh} }
\end{mathpar}
If $\{\xx_1,\xx_2\}\subseteq\dom\Gamma$, then the rule \Rule{AppInfer} tries
to find all instances of the current context in which the
application $\xx_1 \xx_2$ is typeable, by subsuming $\Gamma(\xx_1)$
(the type of the function) to $\Gamma(\xx_2)\To\alpha$ (a function
type whose domain is the type of the argument). For that,
it calls the tallying algorithm which returns a set of substitutions $\Msubst$. Then,
$\ressubst{\Msubst}{\annotityp}{\annotiuntyp}$ is returned,
meaning that this application should be typeable under every instance $\Gamma\msubst$
of the current context $\Gamma$ (with $\msubst\in\Msubst$).
The default case (i.e., when the current context is unchanged, for
example, when $\Msubst=\varnothing$) cannot be
typed,
so it is annotated with $\annotiuntyp$ (see
rule \Rule{Iterate$_2$} later on).
The rules for pairs are similar and have been omitted.
\begin{mathpar}
  \Infer[CaseSplit]
  { \Gamma(\xx)\not\leq\tau\\\Gamma(\xx)\not\leq\neg\tau }
  {\Gamma\aavdash \epa{\tcase {\xx} \tau {\xx_1}{\xx_2}}{\annotiinfer}\refines
  {\respart{\envsingl \xx \tau}{\annotiinfer}{\annotiinfer}}}
  { }
\end{mathpar}
The key rule for type-cases is \Rule{CaseSplit},
corresponding to the case where $\xx$ is in $\Gamma$, but with a type
that does not allow the selection of a specific branch.
Thus, we need to partition the type of $\xx$ in two, one part being a subtype of $\tau$ and the other
a subtype of $\neg \tau$. This is achieved by returning $\respart{\{(\xx:\tau)\}}{\annotiinfer}{\annotiinfer}$:
this result is backtracked up to the binding of $\xx$, where it will split the associated type, accordingly.
\begin{mathpar}
  \Infer[CaseThen]
  {
    \Gamma(\xx)\leq\tau\\
    \tallyinfer{\Msubst}{\Gamma(\xx)}{\Empty}
  }
  {\Gamma\aavdash \epa{\tcase {\xx} \tau {\xx_1}{\xx_2}}{\annotiinfer}\refines\ressubst{\Msubst}{\annotityp}{\annotithen}}
  { }
  \\
  \Infer[CaseElse]
  {
    \Gamma(\xx)\leq\neg\tau\\
    \tallyinfer{\Msubst}{\Gamma(\xx)}{\Empty}
  }
  {\Gamma\aavdash \epa{\tcase {\xx} \tau {\xx_1}{\xx_2}}{\annotiinfer}\refines\ressubst{\Msubst}{\annotityp}{\annotielse}}
  { }
  \\
  \Infer[CaseVar$_i$]
  {
    \xx_i\not\in\dom\Gamma
  }
  {\Gamma\aavdash \epa{\tcase {\xx} \tau {\xx_1}{\xx_2}}{\annotithenelse}\refines
  {\resvar{\xx_i}{\annotityp}{\annotiuntyp}}}
  { }
\end{mathpar}
When the type of $\xx$ allows the selection of a branch, then either the
rule \Rule{CaseThen} or the rule \Rule{CaseElse} applies. If we are in
the
case of \Rule{CaseThen}, that is $\Gamma(\xx)\leq\tau$, then we have
to determine whether we will apply the algorithmic
rule \Rule{$\Empty$\AR} or the algorithmic
rule  \Rule{$\in_1$\AR}. To determine it, the \Rule{CaseThen} rule calls
$\tallyinferf{\Gamma(\xx)}{\Empty}$ which returns the set of contexts
$\Gamma\msubst$ (for $\msubst\in\Msubst$)  under which the algorithmic rule \Rule{$\Empty$\AR} is to be applied,
that is, the contexts under which the tested expression $\xx$ has an empty type.
The default case, corresponding to the case in which the type of
$\Gamma(\xx)$ is not guaranteed to be empty 
and, thus, in which the algorithmic
rule \Rule{$\in_1$\AR} must be applied, is annotated with $\annotithen$. 
This annotation is handled by the rule \Rule{CaseVar$_1$} which forces
the system to type $\xx_1$, the binding variable
associated to the first branch. The case for \Rule{CaseElse}
and \Rule{CaseVar$_2$} is analogous.

We omitted the remaining rules for type-cases since they are
straightforward: the rule for $\xx\not\in\dom\Gamma$, which triggers a
$\resvar{\xx}{\annotiinfer}{\annotiuntyp}$ result; two rules
similar to \Rule{CaseVar$_i$}, but where $\xx_i\in\dom\Gamma$, which simply
return $\resok\annotityp$.
\begin{mathpar}
  \Infer[LambdaInfer]
  {
      \Gamma\aavdash \epa{\lambda x. \kappa}{\annotilambda{\monovar}{\annotiinfer}} \refines\results
  }
  {
      \Gamma\aavdash \epa{\lambda x. \kappa}{\annotiinfer} \refines\results
  }
  { \monovar\in\monovars\text{ fresh} }
  \\
  \Infer[Lambda]
  {
      \Gamma,x:\mt\raavdash \epa{\kappa}{\banns} \refines\results
  }
  {
      \Gamma\aavdash \epa{\lambda x. \kappa}{\annotilambda{\mt}{\banns}} \refines
      {\mapres\results{X}{\annotilambda{\mt}{X}}}
  }
  { }
\end{mathpar}
The rules for $\lambda$-abstractions mimic algorithm $\mathcal W$. Rule \Rule{LambdaInfer}
transforms the initial $\annotiinfer$ annotation into a
$\annotilambda{\monovar}{\annotiinfer}$ annotation. As in $\mathcal{W}\!$, $\lambda$-abstracted variables are initially given a fresh type variable,
which will then be substituted as needed while reconstructing the type
of the
body; here we use a fresh monomorphic variable, but
$\tallyinferff$ will transform it into a polymorphic---thus,
instantiable---one, just for the reconstruction in the body.
Rule \Rule{Lambda} adds the $\lambda$-abstracted variable to the
environment with the type specified in the annotation,
recursively calls reconstruction on the body, and reestablishes the variable type annotation on the result.
The notation $\mapres\results{X\!\!}{\!\!f(X)}$ denotes the result $\results$ where $f$ has been applied
to every annotation $X$.
\begin{mathpar}
  \Infer[BindInfer]
  {
      \Gamma\aavdash\epa{\bindexp \xx a \kappa}{\annotitryskip{\annotiinfer}}\refines\results
  }
  {
      \Gamma\aavdash\epa{\bindexp \xx a \kappa}{\annotiinfer}\refines\results
  }
  { }
\end{mathpar}
The \Rule{BindInfer} rule transforms an initial $\annotiinfer$
annotation into a $\annotitryskip{\annotiinfer}$ annotation which skips
the binding and annotates the body $\kappa$ with $\annotiinfer$. We do not try to type the definition of a binding until
it is actually used, because its variable might appear only in unreachable
positions (e.g., in an unreachable branch of a type-case). In other
words, we implement a lazy typing discipline for bind-abstracted variables.
If the variable is used at some point, then an attempt to type it will be initiated by the \Rule{BindTrySkip$_1$} rule below:
\begin{mathpar}
  \Infer[BindTrySkip$_1$]
  {
      \Gamma\raavdash \epa\kappa\banns\refines\resvar{\xx}{\banns_1}{\banns_2}\\
      \Gamma\aavdash \epa{\bindexp \xx a \kappa}{\annotitrybind{\annotiinfer}{\banns_1}{\banns_2}}\refines\results
  }
  {
      \Gamma\aavdash\epa{\bindexp \xx a \kappa}{\annotitryskip{\banns}}\refines\results
  }
  { }
\end{mathpar}
This rule tries to type the body of the binding, starting with the annotation $\banns$ (initially, $\annotiinfer$).
If the result is  $\resvar{\xx}{\banns_1}{\banns_2}$, then it means
that the current binding is used in the body $\kappa$ and,
thus, the system should try to type it.
Consequently, the annotation for the current binding is changed into a $\annotitrybind{\annotiinfer}{\banns_1}{\banns_2}$
so that, at the next iteration, its definition will be reconstructed.

If typing the body of the binding yields a result different from $\resvar{\xx}{\banns_1}{\banns_2}$,
then this result is just propagated as in \Rule{Lambda} (the corresponding rules have been omitted).
\begin{mathpar}
  \Infer[BindTryKeep$_1$]
  {
      \Gamma\raavdash \epa a\lanns\refines\resok{\lanns'}\\
      \Gamma\aavdash\epa{\bindexp \xx a \kappa}{\annotibind{\lanns'}{\{(\Any,\banns_1)\}}{\emptyset}}\refines\results
  }
  {
      \Gamma\aavdash\epa{\bindexp \xx a \kappa}{\annotitrybind{\lanns}{\banns_1}{\banns_2}}\refines\results
  }
  { }
  \\
  \Infer[BindTryKeep$_2$]
  {
      \Gamma\raavdash \epa a\lanns\refines\resfail\\
      \Gamma\aavdash\epa{\bindexp \xx a \kappa}{\annotiskip{\banns_2}}\refines\results
  }
  {
      \Gamma\aavdash\epa{\bindexp \xx a \kappa}{\annotitrybind{\lanns}{\banns_1}{\banns_2}}\refines\results
  }
  { }
\end{mathpar}
As expected, if the current annotation for the binding is a $\annotitrybind{\lanns}{\banns_1}{\banns_2}$,
then the system tries to reconstruct the annotation for the definition. If it succeeds,
then it becomes possible to type the definition and to continue the reconstruction of the body using $\banns_1$.
This is what \Rule{BindTryKeep$_1$} does by changing the current annotation to
$\annotibind{\lanns'}{\{(\Any,\banns_1)\}}{\emptyset}$ (more details below).
If the reconstruction of the definition fails (rule \Rule{BindTryKeep$_2$}), then we have no choice but
to skip this definition and use the default annotation $\banns_2$ to type the body.


In an annotation $\annotibind{\lanns}{\sanns}{\sanns'}$ for
the binding of a variable $\xx$, $\lanns$ is the
annotation for typing the definition of $\xx$, while the two other
arguments describe the type decomposition to use for $\xx$ and, for each part of the decomposition, the annotation to use for the body. More precisely,
$\sanns$ contains the parts of the type decomposition
that have yet to be explored, and $\sanns'$ contains the parts that have already been fully explored.
In particular, the annotation
$\annotibind{\lanns'}{\{(\Any,\banns_1)\}}{\emptyset}$ used in rule \Rule{BindTryKeep$_1$}
means that the type of the definition does not need to be partitioned:
there is only one part, covering $\Any$, associated with an annotation $\banns_1$ for typing the body.

\begin{mathpar}
  \Infer[BindOk]
  { }
  {
    \Gamma\aavdash\epa{\bindexp \xx a \kappa}{\annotibind{\lanns}{\emptyset}{\sanns}}
    \refines\resok{\annotibind{\lanns}{\emptyset}{\sanns}}
  }
  { }
  \end{mathpar}
  If all the parts of the type decomposition have already been
  explored (i.e., $\varnothing$ in the annotation in the rule above),
  then the reconstruction is successful. Otherwise,
the following rules are applied:
  \begin{mathpar}
    \Infer[BindKeep$_1$]
    {
      \Gamma\paavdash \epa a{\lanns}\refines\annota\\\Gamma\vdashA \ea a{\annota}:s\\
      \Gamma, \xx:s\land\mt\raavdash \epa\kappa\banns\refines\resok{\banns'}\\
      \Gamma\aavdash\epa{\bindexp \xx a \kappa}{\annotibind{\lanns}{\sanns}{\{(\mt,\banns')\}\cup\sanns'}}\refines\results
    }
    {
      \Gamma\aavdash\epa{\bindexp \xx a \kappa}{\annotibind{\lanns}{\{(\mt,\banns)\}\cup \sanns}{\sanns'}}
      \refines{\results}
    }
    { }
\end{mathpar}\begin{mathpar}
  \Infer[BindKeep$_2$]
  {
    \Gamma\paavdash \epa a{\lanns}\refines\annota\\\Gamma\vdashA \ea a{\annota}:s\\
    \Gamma, \xx:s\land\mt\raavdash\epa{\kappa}{\banns}\refines\respart{\Gamma'}{\banns_1}{\banns_2}\\
    \xx\in\dom{\Gamma'}\\
    \Gamma\eaavdash\ety{a}{\neg (\mt\land\Gamma'(\xx))}\refines\Gammas_1\\
    \Gamma\eaavdash\ety{a}{\neg (\mt\setminus\Gamma'(\xx))}\refines\Gammas_2
  }
  {
    \Gamma\aavdash\epa{\bindexp \xx a \kappa}{\annotibind{\lanns}{\{(\mt,\banns)\}\cup \sanns}{\sanns'}}
    \refines\respart{\Gamma'\setminus\xx}{\banns_1'}{\banns_2'}
  }
  { }
\end{mathpar}
with, in the last rule,
$\banns_1'=\annotiprop{\lanns}{\Gammas_1\cup\Gammas_2}{\{(\mt\land\Gamma'(\xx),\banns_1),\,
(\mt\setminus\Gamma'(\xx),\banns_2)\}\cup\sanns}{\sanns'}$ and $\banns_2'=\annotibind{\lanns}{\{(\mt,\banns_2)\}\cup \sanns}{\sanns'}$

\medskip
In both rules, the definition of the binding is typed using the annotation $\lanns$.
For that, it is first converted into an annotation $\annota$ of
the algorithmic type system, using the deduction rules for the judgment $\Gamma\paavdash \epa
a{\lanns}\refines\annota$, defined in Section~\ref{sec:auxreconstruction}.
Then, the type $s$ obtained for the definition is
intersected with one of the parts of the type decomposition, according to the second argument
of the \texttt{keep}() annotation (i.e., $\{(\mt,\banns)\}\cup \sanns$ in both rules),
and the corresponding annotation for the body is reconstructed recursively.
Note that, since split annotations are sets, then the order in which the parts are explored is arbitrary.

The rule \Rule{BindKeep$_1$} for an annotation $\annotibind{\lanns}{\sanns}{\sanns'}$  is responsible for moving a branch from $\sanns$ to $\sanns'$
when the result for the branch is $\resok{}$. If instead the
reconstruction of the body requires to further split the type of $\xx$,
then the rule \Rule{BindKeep$_2$} splits the current branch into two branches.
However, before exploring these two branches, some information about the split needs to be propagated,
to ensure that when a split is explored, it is under a context as precise as possible.

Let us explain this by an example.
Assume  we have a polymorphic primitive function $\texttt{id}$ of type $\polyvar\to\polyvar$
and an initial environment $\Gamma=\{ x:\Bool \}$. We want  to type the following canonical form,
and deduce for it the type $\True$ (since $\xx$ and $\yy$ are always bound
to the same value):
\[ \bindexp{\xx}{x}{\bindexp{\yy}{\texttt{id } \xx}{\bindexp{\zz}{\tcase{\yy}{\True}{\xx}{\true}}{\zz}}} \]
At some point, the partition associated to $\yy$ will change from $\{\Any\}$ to $\{\True,\neg\True\}$
because of the type-case (rule \Rule{CaseSplit}).
However, if the case corresponding to $(\yy:\True)$ is immediately explored, it will yield for the body
the type $\Bool$, because $\xx$ still has the type $\Bool$ in the environment. In order to obtain
the more precise type $\True$, we must deduce, before exploring the case $(\yy:\True)$,
that when $\texttt{id}\ \xx$ (the definition of $\yy$) has type
$\True$, then $\xx$ also has type $\True$. Knowing that,
the type of $\xx$ should be split accordingly into $\{\True,\neg\True\}$.

This mechanism of backward propagation of splits is initiated in the \Rule{BindKeep$_2$} rule
with the two premises $\Gamma\eaavdash\ety{a}{\neg (\mt\land\Gamma'(\xx))}\refines\Gammas_1$ and
$\Gamma\eaavdash\ety{a}{\neg (\mt\setminus\Gamma'(\xx))}\refines\Gammas_2$.
This auxiliary judgment $\Gamma\eaavdash\ety{a}{\mt}\refines\Gammas$, defined in Appendix~\ref{sec:envrefinement-appendix}, can be read as follows:
\textit{refining the current environment $\Gamma$ with one of the $\Gamma'\in\Gammas$ ensures that the atom $a$ will have type \textup{$\mt$}}.
The refinements we obtain are stored in the annotation of the binding, using an annotation
$\annotiprop{\lanns}{\Gammas}{\sanns}{\sanns'}$.
This annotation is handled by two other rules (omitted here) whose role is to propagate these refinements
one after the other using successive $\respart{\Gamma'}{\banns_1}{\banns_2}$ results (with $\Gamma'\in\Gammas$),
before finally restoring a $\annotibind{\lanns}{\sanns}{\sanns'}$ annotation.

\begin{mathpar}
  \Infer[InterEmpty]
  { }
  {
      \Gamma\aavdash \epa{\kappaa}{\annotiinter{\emptyset}{\emptyset}} \refines \resfail
  }
  { }
  \quad
  \Infer[InterOk]
  { }
  {
      \Gamma\aavdash \epa{\kappaa}{\annotiinter{\emptyset}{S}} \refines
      \resok{\annotiinter{\emptyset}{S}}
  }
  { }
  \end{mathpar}\begin{mathpar}
  \Infer[Inter$_1$]
  {
      \Gamma\raavdash \epa\kappaa{\aanns}\refines\resok{\aanns'}\\
      \Gamma\aavdash \epa{\kappaa}{\annotiinter{S}{\{\aanns'\}\cup S'}}\refines\results
  }
  {
      \Gamma\aavdash \epa{\kappaa}{\annotiinter{\{\aanns\}\cup S}{S'}} \refines\results
  }
  { }
\end{mathpar}\begin{mathpar}
  \Infer[Inter$_2$]
  {
      \Gamma\raavdash \epa\kappaa{\aanns}\refines\resfail\\
      \Gamma\aavdash \epa{\kappaa}{\annotiinter{S}{S'}}\refines\results
  }
  {
      \Gamma\aavdash \epa{\kappaa}{\annotiinter{\{\aanns\}\cup S}{S'}} \refines\results
  }
  { }
\end{mathpar}\begin{mathpar}
  \Infer[Inter$_3$]
  {
      \Gamma\raavdash \epa\kappaa\aanns\refines\results
  }
  {
      \Gamma\aavdash \epa{\kappaa}{\annotiinter{\{\aanns\}\cup S}{S'}} \refines
      \mapres{\results}{X}{{(\annotiinter{\{X\}\cup S}{S'})}}
  }
  { }
\end{mathpar}

Intersection annotations are introduced
by the $\raavdash$ judgments defined below.
In an intersection annotation $\annotiinter{S}{S'}$,
the annotations in $S'$ are fully processed (i.e., the associated
reconstruction returned $\resok{}$),
while the annotations in $S$ are not: they still have to be refined one after the other (rule \Rule{Inter$_3$}).
If one of them becomes fully processed, it is moved in $S'$ (rule \Rule{Inter$_1$}).
Conversely, if one of them fails, it is removed (rule \Rule{Inter$_2$}).
The process stops when $S$ is empty: then, the reconstruction fails if $S'$ is empty (rule \Rule{InterEmpty}),
and succeed otherwise (rule \Rule{InterOk}).

Finally, we formalize the rules for the judgments $\raavdash$. As said earlier,
the purpose of $\raavdash$ is to repeatedly call $\aavdash$ judgments
so that, in the end, we obtain a terminal result.
\begin{mathpar}
  \Infer[Iterate$_1$]
  { \Gamma\aavdash\epa{\kappaa}{\aanns}\refines\respart{\Gamma'}{\aanns_1}{\aanns_2}\\
    \Gamma\raavdash\epa{\kappaa}{\aanns_1}\refines\results
  }
  { \Gamma\raavdash\epa{\kappaa}{\aanns}\refines\results }
  {\Gamma'=\emptyenv}
\end{mathpar}\begin{mathpar}
  \Infer[Iterate$_2$]
  { 
    \Gamma\aavdash \epa\kappaa\aanns\refines\ressubst{\{\msubst_i\}_{i\in I}}{\aanns_1}{\aanns_2}\\
    \Gamma\raavdash\epa{\kappaa}{\annotiinter{\{{\aanns_1\msubst_i}\}_{i\in I}\cup\{\aanns_2\}}{\emptyset}}\refines\results
  }
  { \Gamma\raavdash\epa{\kappaa}{\aanns}\refines\results }
  {\forall i\in I.\ \msubst_i\disjoint\Gamma }
\end{mathpar}
The iteration continues as long as it yields non-terminal results that are
immediately usable, that is, either  they return a trivial split (i.e.,
$\Gamma'=\emptyenv$) as in rule \Rule{Iterate$_1$}, or they return
substitutions that do not affect the current environment (i.e.,
$\msubst_i\disjoint\Gamma$) as in rule \Rule{Iterate$_2$}. For the
latter rule, the iteration may need to introduce an intersection
annotation (useless when $I$ is empty) in order to explore all the cases of a
$\ressubst{\{\msubst_i\}_{i\in I}}{\aanns_1}{\aanns_2}$ result (where
$\aanns\msubst$ is the intermediate annotation $\aanns$ in which the
substitution $\msubst$ has been applied recursively to every type in it).
An important special case of
the \Rule{Iterate$_2$} rule is when $I=\varnothing$: in that case the
iteration continues by trying to type $\kappaa$ with the default
annotation $\aanns_2$ and the current environment $\Gamma$. For
instance, this special case triggers
a \Rule{CaseVar$_1$} after a \Rule{CaseThen} and a \Rule{CaseVar$_2$}
after a \Rule{CaseElse}.

If the result is already terminal or if it is not immediately usable,
then it is directly returned:
\begin{mathpar}
  \Infer[Stop]
  { \Gamma\aavdash\epa{\kappaa}{\aanns}\refines\results }
  { \Gamma\raavdash\epa{\kappaa}{\aanns}\refines\results }
  { }
\end{mathpar}
In particular, if $\results = \respart{\Gamma'}{\aanns_1}{\aanns_2}$
where $\Gamma'\not=\varnothing$ (i.e., \Rule{Iterate$_1$} does not
apply), then \Rule{Stop} backtracks until $\Gamma'$ becomes empty;
likewise if  $\results =\ressubst{\{\msubst_i\}_{i\in
I}}{\aanns_1}{\aanns_2}$ and $\Gamma\psi_i\not\simeq\Gamma$ for some $i$ (i.e. \Rule{Iterate$_2$} does not
apply), then \Rule{Stop} backtracks until it exits the scope of
the binders of the
variables that make the side condition of  \Rule{Iterate$_2$} fail.

\subsection{Auxiliary Reconstruction Algorithm}\label{sec:auxreconstruction}

The auxiliary reconstruction algorithm defined in this section
converts an intermediate annotation of the main reconstruction system into an annotation
for the algorithmic type system.
For that, it needs to retrieve the polymorphic substitutions $\Sigma$
needed to type the atoms.

Formally, the algorithm takes as input an environment $\Gamma$, an atom or canonical form $\kappaa$,
and an intermediate annotation $\aanns$, and produces an annotation $\aannot$ for the algorithmic type system.
It is presented as a deduction rule system for  judgments of the form $\Gamma\paavdash\epa{\kappaa}{\aanns}\refines\aannot$.
Some rules have been omitted for concision (they can be found in
Appendix~\ref{sec:auxreconstruction-appendix}): for instance,
%
%
the rules for constants and axioms are omitted since straightforward, as they just  transform an intermediate annotation $\annotityp$
into an annotation $\annotvara$ for the algorithmic type system; 
likewise, the rules for $\lambda$-abstractions and intersections are straightforward
and have been omitted, since they just proceed recursively
on their children annotations. The most important rule for this system
is the one for applications:
\begin{mathpar}
  \Infer[App]
  {
      t_1=\Gamma(\xx_1)\\t_2=\Gamma(\xx_2)\\
      \renaming_1=\arenaming{t_1}\\\renaming_2=\arenaming{t_2}\\
      \tally{\Sigma}{t_1\renaming_1}{t_2\renaming_2\to \polyvar}
  }
  {\Gamma\paavdash \epa{{\xx_1}{\xx_2}}{\annotityp}\refines
  \annotapp{\{\sigma\circ\renaming_1\,\alt\,\sigma\in\Sigma\}}{\{\sigma\circ\renaming_2\,\alt\,\sigma\in\Sigma\}}}
  {\begin{array}{l}\Sigma\neq\emptyset\\[-1mm] \polyvar\in\polyvars\text{ fresh}\end{array} }
\end{mathpar}
where $\arenaming{t}$ returns a renaming from $\vars t\cap\polyvars$ to fresh polymorphic variables.

For applications, an annotation of the form $\annotapp{\Sigma_1}{\Sigma_2}$ must be produced.
In order to find some instantiations $\Sigma_1$ and $\Sigma_2$ (for
$\xx_1$ and $\xx_2$ respectively) that make the application typable,
the \Rule{App} rule solves the tallying instance $\tallyf{t_1\renaming_1}{t_2\renaming_2\to \polyvar}$.
The purpose of $\renaming_1$ and $\renaming_2$ is to decorrelate type variables in $\Gamma(\xx_1)$
and in $\Gamma(\xx_2)$. For instance, assume we want to reconstruct
the instantiations for the atom ``$\xx\,\xx$''
with $\Gamma(\xx)=\polyvarb\to\polyvarb$.
{\color{modification}
The tallying instance  $\tallyf{\polyvarb\to\polyvarb}{(\polyvarb\to\polyvarb)\to\polyvar}$ yields only a very specific, uninteresting solution (i.e., $\polyvar=\polyvarb=\mu X.X\to X$)\footnote{\color{modification}The solution is not interesting since it is the one that allows any simply typed system to type all pure lambda terms. We do not need this recursive type to type, say, the application of the polymorphic identity function to itself.} 
because of the use of the same type variable $\polyvarb$ on both sides of $\dot\leq$. But each occurrence of $\xx$ has a polymorphic type that can be instantiated independently. Thus, we remove this useless and constraining dependency by refreshing the generic type variables yielding
 $\tallyf{\polyvarb'\to\polyvarb'}{(\polyvarb\to\polyvarb)\to\polyvar}$ which has interesting solutions, in particular  
$\{ \polyvarb'\leadsto\polyvarb\to\polyvarb \;;\; \polyvar\leadsto\polyvarb\to\polyvarb \}$.}
The side-condition $\Sigma\neq\emptyset$ ensures that the tallying instance has at least one solution
(otherwise the annotation produced would be invalid).
%
%
%
The rules for projections, pairs, and type-cases are similar and, thus, omitted.
%
%
%
%
\begin{mathpar}
  \Infer[BindKeep]
  {
  \Gamma\paavdash \epa a \lanns\refines\annota\qquad
  \Gamma\vdashA \ea a \annota:s\qquad 
  {\small(\forall i\in I)}\ \ \Gamma,\xx:s\land\mt_i\paavdash \epa\kappa{\banns_i}\refines \annot_i\\
  }
  {
  \Gamma\paavdash\epa{\bindexp {\xx} {a} {\kappa}} {\annotibind \lanns{\emptyset}{\{(\mt_i,\banns_i)\}_{i\in I}}}
  \refines \annotbind \annota {\{(\mt_i,\annot_i)\}_{i\in I}}
  }{(*)}
\end{mathpar}
(where $(*)$ is $\tbvee_{i\in I}\mt_i\simeq\Any$).
The rule \Rule{BindKeep} takes as input an intermediate annotation $\annotibind{\lanns}{\sanns}{\sanns'}$,
with $\sanns=\varnothing$, since all branches must have been fully explored by the main reconstruction algorithm.
The rule recursively transforms the intermediate annotation $\lanns$ for the definition $a$ into an annotation $\annota$
for the algorithmic type system, and uses it to type $a$.
It can then update the environment and proceed recursively on the body $\kappa$,
for each branch in $\sanns'$.


{\color{modification}
\subsection{Properties of the Reconstruction Algorithm}\label{sec:propreco}
As recalled at the beginning of the section, reconstruction
is sound, terminating, but incomplete.

\begin{theorem}[Soundness]
  If $\Gamma\paavdash\epa{\kappa}{\banns}\refines\annot$, then $\exists t.\ \Gamma\vdashA\ea{\kappa}{\annot}:t$.
\end{theorem}

\begin{theorem}[Termination]\label{rec_term}
    The deduction rules $\raavdash$ and $\aavdash$ define a terminating algorithm:
    it can either fail (if no rule applies at some point) or return a result $\results$.
\end{theorem}

The incompleteness of the reconstruction algorithm is inherent to our
system and derives from the lack of principal typing. A simple example
is the curried function \code{map} defined in the third row of Table~\ref{tab:bench} in the next section. Our reconstruction deduces for it the
type $(\alpha\to\beta)\to [\alpha*]\to[\beta*]$ (actually, a slightly
more precise type),
where $[\alpha*]$ is the type of the lists of elements of type
$\alpha$. This states that an application of \code{map} yields a
function that maps lists of $\alpha$'s into lists of $\beta$'s. But for
every natural number $n$, the declarative system can also deduce that
the result maps lists of $\alpha$'s of length $n$ into lists of
$\beta$'s of the same length $n$.
Our algorithm can \emph{check}
each of these types, but none of them can be deduced from the type
reconstructed by the algorithm. And since we do not have dependent
types or infinite intersections, then the declarative system cannot
have a principal type expressing all these different derivations.
In other terms, incompleteness stems from the fact that the
declarative system can use all the infinitely many decompositions of
unions in the union elimination rule, and the infinitely many
decompositions of the domain of a
function when reconstructing its type as an intersection of arrows. The
algorithmic counterpart of this, is that there are infinitely many
annotations that the algorithmic system can use to type these
expressions and that these infinite choices cannot be summarized by a
notion of principal annotation: the reconstruction chooses one
particular annotation, and therefore it will miss some solutions.


There is a second source of incompleteness for reconstruction, which is
not inherent to the system, but a design choice, instead: the fact that reconstruction does not perform the so-called
``expansion'' of intersection types. This is shown by the
rule \Rule{App} in Section~\ref{sec:auxreconstruction}, where tally is
applied without expanding the types in the constraint (e.g., if
      $\tallyf{t_1\renaming_1}{t_2\renaming_2\to \polyvar}$
fails we can expand the type of the function and try $\tallyf{t_1\renaming_1 \wedge t_1\renaming_3}{t_2\renaming_2\to \polyvar}$,
and so on and so forth by alternating expansions on the function and on the
argument types: see~\cite[Section~3.2.3]{polyduce2} for more details). The consequence of this is that if you take the
definition of the function \code{filter} given in row 4 of
Table~\ref{tab:bench}, and you remove all type annotations, then the
type reconstructed by the algorithm for it is less precise than the
one specified by the annotations, which could have been reconstructed
if the algorithm had instead expanded the type of the parameter \code{f}.

Despite incompleteness, the declarative rules of
Figure~\ref{fig:declarative} form a reliable guide to which
programs are accepted, provided we bear in mind that the algorithm
approximates data structures according to the tests performed
on them. So, typically, the type reconstructed for a function on
lists, will probably differentiate the cases for empty and not-empty
lists, but not for, say, lists of size 42, unless the function
contains an explicit test for it. This (and to a lesser extent,
expansion) is essentially the main difference with the declarative
system, which has the liberty to deduce the type for the case of lists
of size 42, even if this property is not tested in the body of the
function. In that case, the programmer can still use an explicit type
annotation to \emph{check} that the specific type works.

}

\section{Implementation}
\label{sec:implementation}

We have implemented the reconstruction algorithm presented in
Section~\ref{sec:reconstruction}, using the CDuce \cite{cduce} API for
the subtyping and the tallying algorithms. The prototype
is 4500 lines of OCaml code and features several extensions
such as optional type annotations, pattern matching
(cf. Appendix~\ref{sec:extensions-appendix}), records, and a more
user-friendly syntax. It implements some optimizations,
{\color{modification} briefly discussed at the end of this section},
for instance {\color{modification}memoization and a mechanism} to avoid typing redundant branches
when inferring the domains of $\lambda$-abstractions.
\lstset{language=[Objective]Caml,columns=fixed,basicstyle=\linespread{0.43}\ttfamily\scriptsize\color{darkblue},
morekeywords={is},aboveskip=-0.5em,belowskip=-1em,xleftmargin=-0.5em}
%
%
\begin{table}[t]
\caption{Types inferred by the implementation (times are in ms)}\ifspacetrick\vspace{-1mm}\fi
\label{tab:bench}
   {\scriptsize
  \begin{tabular}{|@{\,}c@{\,}|p{0.42\textwidth}|p{0.445\textwidth}|p{0.045\textwidth}|}
\hline
  & Code & Inferred type & Time\\
\hline
1 &
\begin{lstlisting}
type Falsy = False | "" | 0
type Truthy = ~Falsy

let toBoolean x =
    if x is Truthy then true else false

let lOr (x,y) =
    if toBoolean x then x else y

let id x = lOr (x,x)
\end{lstlisting}
&
~\newline
~\newline
$(\Falsy\to\False)\land(\Truthy\to\True)$\newline
~\newline
$(((\alpha\land\Truthy)\times\Any)\to\alpha\land\Truthy)\land$\newline
$((\Falsy\times\beta)\to\beta)$\newline
~\newline
$\alpha\to\alpha$
&
~\newline
~\newline
\hspace*{\fill}3.42\newline
~\newline
\hspace*{\fill}13.31\newline
~\newline
\hspace*{\fill}8.46
\\\hline
2 &
\begin{lstlisting}
let fixpoint = fun f ->
    let delta = fun x ->
        f ( fun  v -> ( x x v ))
    in delta delta
\end{lstlisting}
&
~\newline
$((\beta \to \alpha) \to (\beta \to \alpha) \land \gamma) \to (\beta \to \alpha) \land \gamma$\newline
&
~\newline
\hspace*{\fill}15.37\newline
\\\hline
3 &
\begin{lstlisting}
let map_stub map f lst =
    if lst is Nil then nil
    else (f (fst lst), map f (snd lst))

let map = fixpoint map_stub
\end{lstlisting}
&
~\newline
$\dots$\newline
~\newline
$((\alpha \to \beta) \to ([ \alpha* ] \to [ \beta* ])) \land(\Any \to [  ] \to [  ])$
&
~\newline
\hspace*{\fill}33.03\newline
~\newline
\hspace*{\fill}84.75
\\\hline
4 &
\begin{lstlisting}
let filter_stub filter
        (f: ('a->Any) & ('b -> ~True))
        (l: [('a|'b)*]) =
    if l is Nil then nil
    else if f(fst(l)) is True
    then (fst(l),filter f (snd(l)))
    else filter f (snd(l))

let filter = fixpoint filter_stub
\end{lstlisting}
&
~\newline
~\newline
~\newline
$\dots$\newline
~\newline
~\newline
~\newline
$((\alpha\to\Any){\wedge}(\beta\to\neg\True))\to[(\alpha\vee\beta)*]{\to}[(\alpha{\setminus}\beta)*]$
&
~\newline
~\newline
~\newline
\hspace*{\fill}21.19\newline
~\newline
~\newline
~\newline
\hspace*{\fill}13.83
\\\hline
5 &
\begin{lstlisting}
let rec flatten x = match x with
 | [] -> []
 | h::t -> concat (flatten h) (flatten t)
 | _ -> [x] 
\end{lstlisting}
&
~\newline
{\color{modification}(Tree $\;\to\;$ [($\alpha\setminus$[$\Any*$])$*$]) $\;\;\;\wedge\;\;\;$ ($\beta\setminus$[$\Any*$] $\;\to\;$ [ $\beta\setminus$[$\Any*$] ])
\newline where Tree \;=\; [Tree$*$] $\;\vee\;$ $(\alpha\setminus$[$\Any+$])}
&
~\newline
\hspace*{\fill}374.41
\\\hline
\end{tabular}
}\ifspacetrick\vspace{-2mm}\fi
\end{table}

We give in Table~\ref{tab:bench} the code of several functions, using a syntax similar
to OCaml, where uppercase identifiers (e.g., \Keyw{True}, \Keyw{Truthy}) denote types
and lowercase identifiers denote variables or constants. For each function we
report its inferred type and the time used to infer it. To enhance
readability we manually curated the types which, thus, may be syntactically
different from (but are semantically equivalent to) the types printed by the prototype.
The experiments were performed on an Intel Core i9-10900KF 3.70GHz CPU.
The code was compiled natively using OCaml 4.14.1. All these examples (and more)
can be tested on the web-based interactive prototype
\ifsubmission
submitted as supplementary material (the examples are preloaded).
\else
hosted at
\url{https://www.cduce.org/dynlang}.
\fi
The web version is compiled to JavaScript using js\_of\_ocaml \cite{jsofocaml},
and is about 8 times slower than the native version.

Code 1 features the examples used in the introduction.

Code 2 implements Curry's fix-point combinator in a call-by-value setting.
Though it is traditionally given the type $((\beta \to \alpha) \to (\beta \to \alpha)) \to (\beta \to \alpha)$,
our prototype infers a slightly more precise type by intersecting the co-domain of the
argument with $\gamma$.

Code 3 shows how to use the fix-point combinator to type recursive functions.
The \texttt{map\_stub} function implements a step of the traditional \texttt{map}
function. The type inferred for this function has been omitted for simplicity.
Then, \texttt{map} is obtained by applying the fixed-point combinator to \texttt{map\_stub}.
Note that $[\alpha*]$ denotes a list of elements of type $\alpha$, and $[]$ denotes an empty list.
The branch $\Any\to []\to []$ may be surprising, but it is correct since the map function
does not use its first argument if the second argument is an empty list.

Code 4 shows how type annotations (cf.\ Appendix~\ref{sec:tconstr})
can be used to infer more precise types: when the \code{filter}
function is applied to a characteristic
function for the set
$\alpha\vee\beta$ whose type precises that the elements in $\beta$ do not satisfy the
predicate, then the inferred type has these elements removed from the
type of the result.


{\color{modification} The grammar for expressions in Figure~\ref{fig:syntax} does not
include recursive functions, since \emph{from a theoretical viewpoint} they are useless: \citet[page 356]{Mil78} justifies the
addition of a ``$\textsf{fix}\,x.e$'' expression by the
fact that his system cannot type Curry's fixpoint combinator, but, as
explained in Section~\ref{sec:outline} and shown by Code 2 above, our system
can. However, \emph{from a practical viewpoint}, the use
of \code{let\,rec} definitions instead of fixed-points combinators may
dramatically improve the speed of reconstruction, which is why the previous
definitions of \code{map} and \code{filter} with a fixed-point combinator must be considered  just as  stress tests for our
reconstruction algorithm. For recursive functions we implemented
classic \code{let\,rec} definitions, for which the reconstruction
takes the arity of the function into account.
Code 5 shows the use of \code{let\,rec} and of pattern matching
(cf. Appendix~\ref{sec:pat}) and is an example of the improvement
brought by \code{let\,rec} definitions: reconstruction for the same
definition but with        a fixpoint
combinator is four times slower.
}
The code defines the deep flatten function that
transforms arbitrary nested lists into the list of their elements (where
\code{concat} is a function of type $[\alpha*]\to[\beta*]\to[(\alpha*)
(\beta*)]$, the result being the type of lists starting with
$\alpha$ elements and ending with $\beta$ ones).
\citet{gre19} considers this function to be the ultimate test for
any type system: as he explains, this simple polymorphic
function
defies all type systems since of all existing languages, none can
reconstruct a type for it and only a couple of languages can check its explicitly typed version:
CDuce  and
Haskell (the latter by resorting to complex metaprogramming
constructions). Our system reconstructs a  precise type
for \code{flatten} as shown by the
first arrow in its intersection type, which states that \code{flatten} is a
function that takes a tree (i.e.,
{\color{modification}either a list of elements that are trees, or a value different from a list})
and returns the list
of elements of the tree that are not lists;
{\color{modification}the other arrow of the intersection states that
when flatten is applied to an element different from a list, then it returns
the list containing only that element}.

{\color{modification} Our prototype focuses on proximity with the
inference system for reconstruction, rather than on performance: we
used it mainly to explore and test our system, which is why it is
implemented in a purely functional style with persistent data
structures (so as to simulate the reconstruction inference rules).
Nonetheless, a few optimizations were implemented in order to mitigate
the cost of backtracking and branching.  One source of inefficiency comes from
the intersection nodes that are generated when a destructor is
reconstructed. This generation can lead to an explosion of the number
of branches to explore, even though many of these branches are
redundant. In the prototype, this is mitigated by recording, for each
$\lambda$-abstraction, the domains already explored for it, and by
trimming branches that do not explore new combinations of domains.

Another source of inefficiency comes from the type decompositions performed
after each binding. Although these type decompositions are usually
small (e.g., the type of a binding is seldom split in more than two parts),
it becomes an issue when typing large
expressions with multiple type-cases.
For instance, a preliminary and unoptimized
implementation of the reconstruction algorithm took about 40 seconds to type the
\code{bal}(ance) function used in the module \code{Map} of the OCaml standard
library, that contains 6 different pattern matches and 4
type-cases~\cite{ocamlmap}. Adding a simple memoization mechanism that prevents the
reconstruction from retyping an atom several times for equivalent contexts,
decreased the inference time down to 4 seconds.

While these simple optimizations significantly improve  performance,
they are still far from what would be considered acceptable for  real applications.
To be used in mainstream languages, the type system will have to be adapted
and restricted so as to ensure better and uniform performance. 
To this purpose, we believe that some more language-oriented
optimization techniques could be of help.
An example is what the development team of Luau \cite{Luau}
did on the occasion of its recent switch to semantic subtyping \citep{luausemsub}.
The developers did this switch by implementing a two-phase approach:
first, a sound syntactic system, fast but imprecise, is used to try to prove subtyping,
and only if it fails, the computationally expensive semantic subtyping inference is used.
We think not only that such a staged approach could be applied in our case,
but also that the partial results of the first phase could be used to improve the performance of the later phases,
as in the case of the \code{let\,rec}, where knowing the arity of the defined function
improves the performance of the reconstruction.
This could be further coupled with slicing, meaning that our type reconstruction
could be applied to very delimited regions that would bound the possibility of backtracking.
These techniques are language-dependent, and quite different from the
algorithmic aspects developed here,
though they will completely rely on it. We plan to explore them in
future work.
}

\section{Related work}
\label{sec:related}

This work can be seen as a polymorphic extension of \cite{occtyp22} from which
it borrows some key notions, such as $(i)$ the combination of the union
elimination rule (from \cite{BDD95}) with three rules for type-cases, in order
to capture the essence of occurrence typing (\cite{THF08}), $(ii)$ the use of
MSC forms to drive the application of the union elimination rule, and $(iii)$
the use of annotations in the algorithmic type system. However, the introduction
of polymorphic types greatly modifies the meta-theory. Besides its influence on
the union elimination rule, the interplay between intersection, union
elimination and instantiation suggests a different style of type
annotations, to be amenable to type inference. We use external annotations while
\cite{occtyp22} annotates terms. Further, the presence of type
variables  imposes to use tallying in an inference algorithm inspired by
$\mathcal{W}$ by \citet{DM82} and from \cite{polyduce2}, where tallying was
first introduced to type polymorphic applications. This yields a clear
improvement over \cite{occtyp22} which is unable to infer higher-order types
for function arguments, while our algorithm is able to do so even for recursive
functions.

The use of trees to annotate calculi with full-fledged intersection types is
common. In the presence of explicitly-typed overloaded
functions, one must be able to precisely describe how the types of nested
$\lambda$-abstractions relate to the various ``branches'' of the outermost
function. The work most similar to ours is \cite{LiquoriR07}, since the
deductions are performed on pairs of marked term and proof term. A marked term
is an untyped term where variables are marked with integers and a proof term is
a tree that encodes the structure of the typing derivation and relates marks to
types. Other approaches, such as \cite{Ronchi02,WellsDMT02,BVB08}, duplicate the
term typed with an intersection, such that each copy corresponds exactly to one
member of the intersection. Lastly, the work of \cite{WellsH02} does not
duplicate terms but rather decorate $\lambda$-abstractions with a richer concept
of \emph{branching shape} which essentially allows one to give names to the
various branches of an overloaded function and to use these names in the
annotations of nested $\lambda$-abstraction.
Note that none of these works features type
reconstruction, which was our main motivation to eschew annotations within
terms, since the backtracking nature of our reconstruction would imply
rewriting terms over and over.

Inference for ML systems with subtyping, unions, and intersections has been studied
in MLsub \cite{mlsub} and extended  with richer
types and a limited form of negation in MLstruct \cite{mlstruct}. Both works trade expressivity for
principality. They define a lattice of types and an algebraic subtyping relation
that ensures principality, but forbids the intersection of arrow types. This
precludes them from expressing overloaded functions, but allows them to define a
principal polymorphic type inference with unions and intersections. We justify
our choice of set-theoretic types, with no type principality and a complex
inference, by our aim to type dynamic languages, such as Erlang or JavaScript,
where overloading plays an important role. We favour the expressivity necessary
to type many idioms of these languages, and rely on user-defined annotations
when necessary to compensate for the incompleteness of type inference. Lastly,
both works implement some form of type simplifications (e.g., \citet{mlsub} use automata
techniques to simplify types), a problem of practical
importance that we did not tackle, yet.

\citet{APF20} provide a principal type inference for a type system with rank-2 intersection types.
In their work, overloaded behaviors are expressible using intersection
types, but they are limited by the rank-2 restriction.
Union types are not supported, nor are equi-recursive types
(actually, it does not feature a general notion of subtyping between two arbitrary types).
Their inference does not require backtracking: it generates a set of constraints that are then solved using
a \textit{set unification algorithm}. This approach for inference has some similarities with the one by~\citet{CPN16}
improved and further developed by \citet{tommasophd} in a context with set-theoretic types,
where the \textit{set unification algorithm} is replaced by \textit{tallying} in the presence of subtyping.
However, while \cite{tommasophd} does support intersection types with no ranking limitation, it is not able to infer
intersection types for overloaded functions. Our work aims to improve this aspect, as well as providing a more precise typing of type-cases
(occurrence typing).

Work by \citet{disjointinter} and \citet{fbow} study disjoint intersection and union types.
They allow expressing overloaded behaviors by a general deterministic merge operator.
In our work, we do not have a general merge operator: overloaded behaviors only emerge
through the use of type-case expressions (or the application of an overloaded function).
Our work can be extended with pattern-matching, in which case the first matching branch is selected.
This is a different approach than the one used with disjoint intersection types, where branches are disjoint
and have no priority and where ambiguous programs are rejected using a notion of mergeability and distinguishability,
allowing to define a general merge operator and to support nested composition,
which may be useful in some contexts such as compositional programming \cite{compositional}.

{\color{modification} \citet{Jim00} presents a polar type system which features
intersections and parametric polymorphism. In Jim's type system, quantifiers may
appear only in positive positions in types, while intersections may only appear in
negative positions. This yields a system that is more expressive than rank-2 intersection
types, and therefore more expressive than ML. Furthermore, the system features
principal types, and a decidable type inference. Some aspects of this work are
similar to ours, in particular the use of MGS, an algorithm to compute the most
general solution of a (syntactic) sub-typing problem, that plays the same role as
our tallying algorithm. Despite these similarities, the approaches differ in the
kind of programs they handle: in \cite{Jim00}, intersections are only deduced by applying
higher-order function parameters to arguments of distinct types within the body of a function,
while in our approach, they can also be caused by a type-case.}

Finally, set-theoretic types are starting to be
integrated into \textit{real-world languages}, for instance by \citet{erlang}
for Erlang, by \citet{luausemsub} for Luau, and by \citet{elixir} for Elixir. We
believe that, in the future, our work could be used in these systems in order to
benefit from a more precise typing of type-cases and pattern matching, as well
as by providing an optional type inference that can be used in conjunction with
explicit type annotations.

\section{Conclusion}
\label{sec:conclusion}

This work aims at providing a formal and expressive type system for dynamic languages,
where type-cases can be used to give functions an overloaded behavior.
It features a type inference that mixes both parametric polymorphism (for modularity)
and intersection polymorphism (to capture overloaded behaviors).
{\color{modification}
In that sense, our work is more than  a simple study on typability:
as a
matter of fact, monomorphic intersection and union types are
sufficient to type a closed program where all function applications
are known (cf., Section~\ref{sec:discon}), but this would be
bad from a language design point of view, and it is the reason why
people program using ML-style programming languages rather than
intersection based ones. Separate compilation and modular
definitions are requirements of any reasonable programming
language. The essence of this work is thus to challenge the limits of how
much precision one can obtain (through intersection types)---ideally
precise enough to type idioms of dynamic languages---while preserving
modularity (thanks to let-polymorphism).}

While we believe our work to be an important step towards a better static typing of
dynamic languages, several key features are still missing.
{\color{modification} First, the presence of side effects may invalidate
our approach: if the \Rule{$\vee$} rule in
Figure~\ref{fig:declarative} is applied to two different occurrences of an
expression $e'$ that is not pure, then the rule may type an expression that
yields a run-time type error. This can be  seen on the
algorithmic system, where the transformation into an MSC-form binds the two
occurrences of $e'$ to the same variable, thus wrongly assuming that they both yield
the same result. Strictly speaking, our algorithmic approach does not require expressions to be pure; it
just needs that when two occurrences of an expression may produce two distinct
values that may change the result of a dynamic test, then these two occurrences
must be bound by two different binds. Having only pure expressions is a
straightforward way to satisfy this property. Having each subexpression bound to
a distinct variable  (i.e., no sharing, that is, a less precise
system, in which the union rule is never used) is a way to retain safety
in the presence of side-effects. But between these two extrema, there is a whole
palette of less coarse solutions that make it possible to apply our approach in the
presence of side-effects, and that we plan to study in future work.
This poses two main challenges: $(i)$ how to
separate problematic expressions from non-problematic ones (e.g., a \code{gen\_id:\,Unit$\to$Int} function 
performs side-effects, but if its result is tested
only against \code{Int}, then it is sound to have all occurrences
of \code{gen\_id} bound by the same bind  during typing) which, in terms of the type system,
corresponds to characterize a class of subexpressions $e'$ that
can be safely used in rule \Rule{$\vee$}; 
and $(ii)$~how to do so \emph{before}  our
type inference,
at a point when type information is not available, yet.
}

Second,
while the performance of our prototype is reasonable, it can certainly be
improved by using more sophisticated implementations techniques and
heuristics on the lines we outlined at the end of Section~\ref{sec:implementation}.

Third, the interactions between code that is exported and
code that is local must be better studied and understood: using
intersection for local polymorphic functions and generalization for
global ones, may not always be entirely satisfactory since
the types of the global functions may be ``polluted'' by the types of the
local applications, yielding less a precise reconstruction for the former. One solution can be to hoist the definition of polymorphic
functions at toplevel whenever possible.

Lastly, an important future work is the support of row-polymorphism: while
records can be easily added to the present work, the precise typing of
functions
operating on records requires row-polymorphism. This is especially
important for dynamic languages where records are seamlessly used to
encode both objects and dictionaries. A first step in that direction may be to integrate the work by
\citet{Cas23}, which unifies dictionaries and records.

\begin{acks}                            
We warmly thank the POPL reviewers: their careful reading
and suggestions allowed us to improve the presentation
significantly. A special thank the reviewers of the POPL artifact
evaluation for their detailed and insightful reviews.

This work was partially supported by the \emph{Chaire Langages Dynamiques
pour les Données} of the \emph{Fondation Université Paris-Saclay},
by the \emph{SECUREVAL} ANR project n.\ ANR-22-PECY-0005 and by a
CIFRE PhD.\ grant with Remote Technology.
\end{acks}

\bibliography{main}

\newpage
\appendix

\section{Extensions}
\label{sec:extensions-appendix}
In this appendix, we present some extensions for the source language,
in particular let-bindings (not to be coufounded with the top-level definitions
composing a program: the let-bindings presented in this section can be used
anywhere in an expression and do not generalize the type of their definition)
and pattern matching.

This section gives an overview of these extensions together with some explanations,
but the full semantics and typing rules can be found in the next appendices.

\subsection{Let Bindings}\label{sec:let}

\subsubsection{Declarative Type System}

Let bindings can be added to the syntax of our language:
\begin{equation*}
\begin{array}{lrclr}
    \textbf{Expressions} &e &::=& \cdots\alt \letexp{x}{e}{e}
\end{array}
\end{equation*}
with the following notion of reduction:
\begin{eqnarray*}
  \letexp{x}{v}{e} &\reduces& e\subs x v
\end{eqnarray*}

At first sight, we could think of adding this typing rule to the declarative type system:
\begin{mathpar}
    \Infer[Let]
    { \Gamma \vdash e_1: t_1\\ \Gamma,x:t_1 \vdash e_2: t_2 }
    { \Gamma \vdash \letexp{x}{e_1}{e_2}: t_2 }
    { }
\end{mathpar}

However, this extension of the declarative type system has one issue:
let-bindings can introduce aliasing,
preventing in some cases the \Rule{$\vee$} rule from applying.
For instance, consider the following expression:

\[ \lambda x.\ \letexp{y}{x}{\tcase{f\ x}{\Int}{f\ y}{42}} \]

with $f:\Any\to\Any$.

Though for any argument $x$ this function yields an integer, it is not possible to derive for it the type
$\Any\to\Int$ using this extension of the declarative type system. Indeed, $f\ x$ and $f\ y$ are not
syntactically equivalent and thus the \Rule{$\vee$} rule can only decompose their types independently,
loosing the correlation between these two expressions.

One way to fix this issue is to remove this kind of aliasing before applying the declarative type system.
For that, we can introduce an intermediate language featuring an alternative version of let-bindings:

\begin{equation*}
  \begin{array}{lrclr}
      \textbf{Expressions} &e &::=& \cdots\alt \evalexp{e}{e}
  \end{array}
\end{equation*}

Let-bindings of the source language can be transformed into this alternative version
using a transformation $\transl{.}$ defined as follows (the other cases are straightforward):
\begin{align*}
  \transl{\letexp{x}{e_1}{e_2}} & = \evalexp{\transl{e_1}}{\transl{e_2}\subs{x}{\transl{e_1}}}
\end{align*}

Finally, the declarative type system can be extended with this rule:
\begin{mathpar}
  \Infer[Let]
  { \Gamma \vdash e_1: t_1\\ \Gamma\vdash e_2: t_2 }
  { \Gamma \vdash \evalexp{e_1}{e_2}: t_2 }
  { }
\end{mathpar}

\subsubsection{Algorithmic Type System}

Let-bindings are added to MSC forms as a new atom construction:
\begin{equation*}
    \begin{array}{lrclr}
        \textbf{Atomic expr} &a &::=& \cdots\alt \evalexp{\xx}{\xx}
    \end{array}
\end{equation*}
The intuition is the same as for the declarative type system:
we want to get rid of the aliasing caused by let-bindings,
while still using bindings to \textit{factorize} each subexpression.
Indeed, to produce an atom for the expression
$\letexp x {e_1}{e_2}$ we must replace each subexpression by a
binding variable, which would yield something of the form $\letexp x
{\xx_1}{\xx_2}$. Since the body of the let-expression
is a variable, then the variable $x$ is only an alias for $\xx_1$ and
thus is undesirable.
Consequently, only the other two variables are specified,
which yields $\evalexp{\xx_1}{\xx_2}$ and which explains the
definition of the atom for let expressions.

For instance, the expression 
$\letexp{x}{\lambda y.y}{(x,x)}$ has the following canonical form:
\[
\begin{array}{l}
  \bindexp{\xx_1}{(\lambda y.\bindexp{\yy}{y}\yy)}{}\\
  \bindexp{\xx_2}{(\xx_1,\xx_1)}{}\\
  \bindexp{\xx_3}{(\evalexp{\xx_1}{\xx_2})}{}\\
  {\xx_3}
\end{array}
\]
Note that, as explained above, the variable $x$ is no longer present in the canonical form.

The algorithmic type system can then be extended with the following rule:
\begin{mathpar}
    \Infer[Let\AR]
    { }
    { \Gamma\vdashA \ea {\evalexp{\xx_1}{\xx_2}} \annotvara: \Gamma(\xx_2) }
    { \xx_1\in\dom\Gamma }
\end{mathpar}

It is straightforward to extend the reconstruction with additional rules
in order to support this new construction (c.f. appendix~\ref{sec:reconstruction-appendix}).

\subsection{Type Constraints}\label{sec:tconstr}

A new construction $(e\bbcolon\tau)$ can be added to our source language.
This construction acts as a type constraint: if the expression $e$ does not reduce
to a value of type $\tau$ (and does not diverge), then the reduction will be stuck.
In a sense, it could be seen as a \textit{cast}, but we will not use this terminology
in order to avoid confusions with gradual typing. Actually, we only introduce this construction
because it will be used later to encode more general type-cases.

We add the following construction to our source language:
\begin{equation*}
    \begin{array}{lrclr}
        \textbf{Expressions} &e &::=& \cdots\alt (e\bbcolon\tau)
    \end{array}
\end{equation*}
with the following notion of reduction:
\begin{eqnarray*}
    (v\bbcolon\tau) &\reduces& v \qquad\qquad\qquad\text{if }v\in\tau
\end{eqnarray*}

The declarative type system can trivially be extended by adding this rule:
\begin{mathpar}
    \Infer[Constr]
    {\Gamma\vdash e: \tau\\\Gamma\vdash e: t}
    {\Gamma\vdash (e\bbcolon\tau): t}
    { }
\end{mathpar}

The same construction is added to the atoms of canonical forms:
\begin{equation*}
    \begin{array}{lrclr}
        \textbf{Atomic expr} &a &::=& \cdots\alt \xx\bbcolon\tau
    \end{array}
\end{equation*}

The annotations of the algorithmic type system also need to be extended:
\begin{equation*}
    \begin{array}{lrclr}
      \textbf{Atoms annotations}  &\annota &::=& \cdots \alt \annotannot \Sigma
    \end{array}
  \end{equation*}
and the algorithmic type system is extended with the following rule:
\begin{mathpar}
    \Infer[Constr\AR]
    { }
    { \Gamma\vdashA \ea {\xx\bbcolon\tau} {\annotannot\Sigma}: \Gamma(\xx) }
    { \Gamma(\xx)\Sigma\leq\tau }
\end{mathpar}

It is also straightforward to extend the reconstruction with additional rules
in order to support this new construction (c.f. appendix~\ref{sec:reconstruction-appendix}).

\subsection{Pattern Matching}\label{sec:pat}

Pattern matching is a fundamental feature of functional languages, and even
some dynamic languages such as Python have started to implement it.
In this section, we show how this feature can be added in our source language.
We proceed in two steps: first, a more general typecase construct with arbitrary arity is introduced,
and secondly, this construct is generalized again so that branches can be decorated
with patterns instead of just types.

\subsubsection{Extended Typecases}

We start by adding a generalized version of the typecase, that can have any number
of branches:

\begin{equation*}
  \begin{array}{lrclr}
      \textbf{Expressions} &e &::=& \cdots\alt (\typecases{e}{\tau\rightarrow e\alt\dots\alt\tau\rightarrow e})
  \end{array}
\end{equation*}
with the following notion of reduction:
  \begin{eqnarray*}
    \typecases{v}{\tau_1\rightarrow e_1\alt\dots\alt \tau_n\rightarrow e_n} &\reduces&
    e_k\quad\begin{array}{l}
      \text{if } v:\tau_k\setminus(\tbvee_{i\in 1\isep {k-1}} \tau_i)\\
      \text{for any }k\in 1\isep n
    \end{array}
  \end{eqnarray*}

In terms of typing, however, we choose not to extend the type system with additional rules
in order to preserve its minimality. Instead, we transform expressions with
extended typecases into expressions of the source langauge presented in section~\ref{sec:types},
with the let-binding and type constraints extensions (\ref{sec:let} and \ref{sec:tconstr}).
For that, we use the following transformation:

\begin{align*}
  \transl c & = c\\
  \transl x & = x\\
  \transl {\lambda x.e} & = \lambda x.\transl e\\
  \transl {\pi_i e}&= \pi_i \transl e\\
  \transl {e_1e_2}&= \transl {e_1} \transl {e_2}\\
  \transl {(e_1,e_2)} &= (\transl {e_1}, \transl {e_2})\\
  \transl {\tcase{e}{\tau}{e_1}{e_2}} &= \tcase{\transl e}{\tau}{\transl {e_1}}{\transl {e_2}}\\
  \transl {\letexp{x}{e_1}{e_2}} &= \letexp{x}{\transl {e_1}}{\transl {e_2}}\\
  \transl {e\bbcolon\tau} &= \transl e \bbcolon \tau\\
  \transl{\typecases{e}{\tau_1\rightarrow e_1\alt\dots\alt\tau_n\rightarrow e_n}} & =
  \begin{array}{l}\letexp{x}{(\transl{e}\bbcolon\tbvee_{i\in 1\isep n}\tau_i)}
    {\\\casestr x {\tau_1\rightarrow \transl{e_1}\,;\,\dots\,;\,\tau_n\rightarrow \transl{e_n}}}
  \end{array}\text{ with }x\text{ fresh}\\
  \casestr x {\tau\rightarrow e} & =  e\\
  \casestr x {\tau\rightarrow e \,;\, C} & =
  \tcase{x}{\tau}{e}{\casestr x {C}}
\end{align*}

\subsubsection{Pattern Matching}

Now, we introduce patterns and a pattern matching construct in the source language:
\begin{equation*}
  \begin{array}{lrclr}
    \textbf{Patterns}  &p &::=& \tau\alt x\alt p\patand p\alt p\pator p\alt (p,p)\alt x\pateq c\\
    \textbf{Expressions} &e &::=& \cdots\alt 
    (\matchs{e}{p\rightarrow e\alt\dots\alt p\rightarrow e})
  \end{array}
\end{equation*}

The associated reduction rule can be found in Appendix~\ref{sec:sem-appendix}.

In terms of typing, we proceed as before by transforming
an expression with pattern matching into an expression
without pattern matching (but with extended typecases and let-bindings),
using the following transformation:

\begin{align*}
  \transl c & = c\\
  \transl x & = x\\
  \transl {\lambda x.e} & = \lambda x.\transl e\\
  \transl {\pi_i e}&= \pi_i \transl e\\
  \transl {e_1e_2}&= \transl {e_1} \transl {e_2}\\
  \transl {(e_1,e_2)} &= (\transl {e_1}, \transl {e_2})\\
  \transl {\tcase{e}{\tau}{e_1}{e_2}} &= \tcase{\transl e}{\tau}{\transl {e_1}}{\transl {e_2}}\\
  \transl {\letexp{x}{e_1}{e_2}} &= \letexp{x}{\transl {e_1}}{\transl {e_2}}\\
  \transl {e\bbcolon\tau} &= \transl e \bbcolon \tau\\
  \transl{\typecases{e}{\tau_1\rightarrow e_1\alt\dots\alt\tau_n\rightarrow e_n}} & =
  \typecases{\transl e}{\tau_1\rightarrow \transl{e_1}\alt\dots\alt\tau_n\rightarrow \transl{e_n}}\\
  \transl{\matchs e {p_1\rightarrow e_1\alt\dots\alt p_n\rightarrow e_n}} & =
  \letexp{x}{\transl{e}}{\typecases {x}{
  \begin{array}{l}
  \pattype{p_1}\rightarrow e_1'\\
  \alt\dots\\
  \alt\pattype{p_n}\rightarrow e_n'
  \end{array}}}
\end{align*}
with $x$ fresh,
\begin{align*}
    \pattype{\tau} & = \tau\\
    \pattype{x} & = \Any\\
    \pattype{p_1\patand p_2} & = \pattype{p_1}\land\pattype{p_2}\\
    \pattype{p_1\pator p_2} & = \pattype{p_1}\lor\pattype{p_2}\\
    \pattype{(p_1, p_2)} & = \pattype{p_1}\times\pattype{p_2}\\
    \pattype{x\pateq c} & = \Any
\end{align*}
and where for every $i\in 1\isep m$:\\
$e_i' = \letexp{x_1}{\pattr{x_1}{p_i}{x}}{\ldots\letexp{x_m}{\pattr{x_m}{p_i}{x}}{\transl{e_i}}}$
    for $\{x_1,...,x_m\}=\vars{p_i}$ with
\begin{align*}
  \pattr{x}{x}{e} & = e\\
  \pattr{x}{x\pateq c}{e} & = c\\
  \pattr x {(p_1,p_2)} e  & =  \pattr x {p_i}{\pi_i e} &\text{if } x\in\vars{p_i}\\
  \pattr x {p_1\patand p_2} e   & =  \pattr x {p_i}e &\text{if } x\in\vars{p_i}\\
  \pattr x {p_1\pator p_2} e  & =
    \tcase {e} {\pattype {p_1}}{\pattr x{p_1}e}{\pattr x{p_2}e}\\
  \pattr x p e & = \text{undefined} &\text{otherwise}                                   
\end{align*}

\section{Full semantics with extensions}
\label{sec:sem-appendix}
Expressions of the source language with extensions of Appendix~\ref{sec:extensions-appendix} are defined as follows:
\begin{equation*}
    \begin{array}{lrclr}
      \textbf{Test Types}  &\tau &::=& b\alt \Empty\to \Any\alt \tau\times \tau\alt \tau\vee \tau \alt \neg \tau \alt \Empty\\
      \textbf{Patterns}  &p &::=& \tau\alt x\alt p\patand p\alt p\pator p\alt (p,p)\alt x\pateq c\\
      \textbf{Expressions} &e &::=& c\alt x\alt\lambda x.e\alt e e\alt (e,e)\alt \pi_i e\alt
      \tcase{e}{\tau}{e}{e}\alt \letexp{x}{e}{e}\alt (e\bbcolon\tau)
      \\&&&\hspace*{-8pt}\alt (\typecases{e}{\tau\rightarrow e\alt\dots\alt\tau\rightarrow e})\alt(\matchs{e}{p\rightarrow e\alt\dots\alt p\rightarrow e})\\[.3mm]
      \textbf{Values} & v  &::=& c\alt\lambda x.e\alt (v,v)
    \end{array}
\end{equation*}

The associated reduction rules are:
\begin{eqnarray*}
    (\lambda x.e)v &\reduces& e\subs x{v}\\[-1mm]
    \pi_1 (v_1,v_2) &\reduces& v_1\\[-1mm]
    \pi_2 (v_1,v_2) &\reduces& v_2\\[-1mm]
    \tcase{v}{\tau}{e_1}{e_2} &\reduces&
    e_1\qquad\qquad\qquad\text{if } v\in\tau\\[-1mm]
    \tcase{v}{\tau}{e_1}{e_2} &\reduces& e_2\qquad\qquad\qquad\text{if } v\in\neg\tau\\[-1mm]
    \letexp{x}{v}{e} &\reduces& e\subs x v\\[-1mm]
    (v\bbcolon\tau) &\reduces& v \qquad\qquad\qquad\text{if }v\in\tau\\[-1mm]
    \typecases{v}{\tau_1\rightarrow e_1\alt\dots\alt \tau_n\rightarrow e_n} &\reduces&
    e_k\qquad\qquad\begin{array}{l}
      \text{if } v:\tau_k\setminus(\tbvee_{i\in 1\isep {k-1}} \tau_i)\\
      \text{for any }k\in 1\isep n
    \end{array}\\[-1mm]
    \matchs{v}{p_1\rightarrow e_1\alt\dots\alt p_n\rightarrow e_n} &\reduces&
    e_k(\patsem{v}{p_k})\quad\begin{array}{l}
      \text{if } v:\pattype{p_k}\setminus(\tbvee_{i\in 1\isep {k-1}} \pattype{p_i})\\
      \text{for any }k\in 1\isep n\end{array}
\end{eqnarray*}
together with the context rules that implement a leftmost outermost
reduction strategy, that is, $E[e]\reduces
E[e']$ if $e\reduces e'$ where the evaluation contexts $E[]$ are
defined as
follows:
\begin{align*}
   \textbf{Evaluation Context}\quad E &::=  [\,]\alt v E \alt E e \alt
   (v,E) \alt (E,e) \alt \pi_i E \alt
    \tcase{E}{\tau}{e}{e}\\& \alt \letexp{x}{E}{e} \alt (E\bbcolon\tau)\\
    &\alt (\typecases{E}{\tau\rightarrow e\alt\dots\alt\tau\rightarrow e})\\
    &\alt (\matchs{E}{p\rightarrow e\alt\dots\alt p\rightarrow e})
\end{align*}

Capture-avoiding substitutions are defined as follows (cases for extended typecases and pattern-matchings have been
omitted for concision):
\begin{align*}
    c\subs x{e'} &= c\\
    x\subs x{e'} &= e'\\
    y\subs x{e'} &= y &x\neq y\\
    (\lambda x.e)\subs x{e'} &= \lambda x.e\\
    (\lambda y.e)\subs x{e'} &= \lambda y.(e\subs x{e'}) &x\neq y, y\notin\fv{e'} \\
    (\lambda y.e)\subs x{e'} &= \lambda z.(e\subs y z\subs x{e'}) &x\neq
    y, y\in\fv{e'}, z\text{ fresh} \\
    (e_1 e_2)\subs x{e'} &= (e_1\subs x{e'}) (e_2\subs x{e'})\\
    (e_1, e_2)\subs x{e'} &= (e_1\subs x{e'}, e_2\subs x{e'})\\
    (\pi_i e)\subs x{e'} &= \pi_i (e\subs x{e'})\\
    (\tcase{e_1}{\tau}{e_2}{e_3})\subs x{e'} &= \tcase{e_1\subs x{e'}}{\tau}{e_2\subs x{e'}}{e_3\subs x{e'}}\\
    (\letexp{x}{e_1}{e_2})\subs x{e'} &= \letexp{x}{e_1\subs x{e'}}{e_2}\\
    (\letexp{y}{e_1}{e_2})\subs x{e'} &= \letexp{y}{e_1\subs x{e'}}{e_2\subs x{e'}}&x\neq y, y\notin\fv{e'}\\
    (\letexp{y}{e_1}{e_2})\subs x{e'} &= \letexp{y}{e_1\subs x{e'}}{e_2\subs y z\subs x{e'}}&x\neq y, y\in\fv{e'}, z\text{ fresh}
\end{align*}

The relation $v\in\tau$ that determines whether a value
is of a given type or not and holds true if and only
if $\typof v\leq \tau$, where
\begin{align*}
    \typof {\lambda x.e} &= \Empty\to\Any\\
    \typof c & = \basic{c}\\
    \typof{(v_1,v_2)}&=\typof{v_1}\times\typof{v_2}
\end{align*}

Finally, the operators used in the reduction rule for pattern matching are defined as follows:
\begin{align*}
    \pattype{\tau} & = \tau\\
    \pattype{x} & = \Any\\
    \pattype{p_1\patand p_2} & = \pattype{p_1}\land\pattype{p_2}\\
    \pattype{p_1\pator p_2} & = \pattype{p_1}\lor\pattype{p_2}\\
    \pattype{(p_1, p_2)} & = \pattype{p_1}\times\pattype{p_2}\\
    \pattype{x\pateq c} & = \Any
\end{align*}
and
\begin{align*}
    \patsem{v}{\tau} & = \ids & \text{if }v:\tau\\
    \patsem{v}{x} & = \subs{x}{v} &\\
    \patsem{v}{(p_1\patand p_2)} & = \sigma_1\cup\sigma_2
        &\text{if }\sigma_1=\patsem{v}{p_1}\text{ and }\sigma_2=\patsem{v}{p_2}\\
    \patsem{v}{(p_1\pator p_2)} & = \patsem{v}{p_1}&\text{if }\patsem{v}{p_1}\neq\patfail\\
    \patsem{v}{(p_1\pator p_2)} & = \patsem{v}{p_2}&\text{if }\patsem{v}{p_1}=\patfail\\
    \patsem{v}{(p_1,p_2)} & = \sigma_1\cup\sigma_2
    &\text{if }v=(v_1,v_2)\text{,\,\,\,}\sigma_1=\patsem{v_1}{p_1}\text{ and }\sigma_2=\patsem{v_2}{p_2}\\
    \patsem{v}{(x\pateq c)} & = \subs{x}{c}&\\
    \patsem{v}{p} & = \patfail & \text{otherwise}
\end{align*}

\section{Subtyping Relation}
\label{sec:subtyping-appendix}
\noindent
Subtyping is defined by giving a set-theoretic interpretation of the
types of Definition~\ref{def:types} into a suitable domain
$\Domain$. In case of polymorphic types, the domain at issue must
satisfy the property of \emph{convexity} \cite{CX11}. A simple model
that satisfies convexity was proposed by~\cite{Gesbert2015}. We
succintly present it in this section. The reader may refer
to~\cite[Section 3.3]{Cas22}
for more details. 
\begin{definition}[Interpretation domain~\cite{Gesbert2015}]\label{def:interpretation}
 The \emph{interpretation domain} $ \Domain $ is the set of finite terms $ d $
produced inductively by the following grammar
\begin{align*}
  d & \Coloneqq  c^L \mid (d, d)^L \mid \Set{(d, \domega), \dots, (d, \domega)}^L
    \\
    \domega & \Coloneqq d \mid \Omega
\end{align*}
where $ c $ ranges over the set $ \Constants $ of constants, $L$ ranges over finite sets of type variables,
and where $ \Omega $ is such that $ \Omega \notin \Domain $.
\end{definition}
The elements of $ \Domain $ correspond, intuitively,
to (denotations of) the results of the evaluation of expressions,
labeled by finite sets of type variables. 
In particular, in a higher-order language,
the results of  computations can be functions which, in this model,
are represented by sets of finite relations
of the form $ \Set{(d_1, \domega_1), \dots, (d_n, \domega_n)}^L $,
where $ \Omega $ (which is  not in $ \Domain $)
can appear in second components to signify
that the function fails (i.e., evaluation is stuck) on the
corresponding input. This is implemented by using in the second
projection the meta-variable $\domega$ which ranges over
$ \Domain_\Omega = \Domain \cup \Set{\Omega} $ (we reserve
$d$ to range over $\Domain$, thus excluding $\Omega$).
This constant $\Omega$ is used to ensure that $\Any\to\Any$ is not a
supertype of all function types: if we used $d$ instead of $\domega$,
then every well-typed function could be subsumed to  $\Any\to\Any$
and, therefore, every application could be given the type $\Any$,
independently from its argument as long as this argument is typable (see Section 4.2 of~\cite{FCB08} for details).
The restriction to \emph{finite} relations corresponds to the intuition
that the denotational semantics of a function is given by the set of
its finite approximations, where finiteness is a restriction necessary
(for cardinality reasons) to give the
semantics to higher-order functions. Finally, the sets of type
variables that label the elements of the domain are used to interpret
type variables: we interpret a type variable
$\alpha$ by the set of all elements that are labeled by $\alpha$, that
is $\TypeInter{\alpha}=\Set{ d\mid \alpha\in\tags(d)}$ (where we define
$\tags(c^L)=\tags((d, d' )^L)=\tags(\Set{(d_1, \domega_1), \dots,
  (d_n, \domega_n)}^L)=L$). 

We define the interpretation $ \TypeInter{t} $ of a type $ t $
so that it satisfies the following equalities,
where  $ \Pf $ denotes the restriction of the powerset to finite
subsets and $  \ConstantsInBasicType{}$ denotes the function
that assigns to each basic type the set of constants of that type, so
  that  for every constant $c$ we have $
  c\in \ConstantsInBasicType(\basic {c})$
(we use $\basic{c}$ to denote the
basic type of the constant $c$):
\begin{align*}
  \TypeInter{\Empty} & = \emptyset&
  \TypeInter{\alpha}& =\Set{ d\mid \alpha\in\tags(d)} &
  \TypeInter{t_1 \lor t_2} & = \TypeInter{t_1} \cup \TypeInter{t_2}
  \\
  \TypeInter{b} & = \ConstantsInBasicType(b)  &
  \TypeInter{\lnot t} & = \Domain \setminus \TypeInter{t} &
  \TypeInter{t_1 \times t_2} & = \TypeInter{t_1} \times \TypeInter{t_2} \\
  \TypeInter{t_1 {\to} t_2} & =
    \{R \in \Pf(\Domain{\times}\Domain_\Omega) \mid\forall (d, \domega) \in R. \:d \in \TypeInter{t_1} \implies \domega \in \TypeInter{t_2}\}
   \span\span
   \span\span
\end{align*}
We cannot take the equations above
directly as an inductive definition of $ \TypeInter{} $
because types are not defined inductively but coinductively.
Notice however that the contractivity condition of
Definition~\ref{def:types}  ensures that the binary relation $\vartriangleright
\,\subseteq\!\types{\times}\types$ defined by $t_1 \lor t_2 \vartriangleright
t_i$, $t_1 \land t_2 \vartriangleright
t_i$, $\neg t \vartriangleright t$ is Noetherian.
This gives an induction principle\footnote{In a nutshell, we can do
proofs and give definitions by induction on the structure of unions and negations---and, thus, intersections---but arrows, products, and basic types are the base cases for the induction.}  on $\types$ that we
use combined with  structural induction on $\Domain$ to give the following definition,
which validates these equalities.
\begin{definition}[Set-theoretic interpretation of types]\label{def:interpretation-of-types}
We define a binary predicate $ (d : t) $
(``the element $ d $ belongs to the type $t$''),
where $ d \in \Domain $ and $ t \in \types $,
by induction on the pair $ (d, t) $ ordered lexicographically.
The predicate is defined as follows:
\begin{align*}
  (c : b)  &= c \in \ConstantsInBasicType(b) \\
  (d:\alpha) &= \alpha\in\tags(d) \\
  ((d_1, d_2) : t_1 \times t_2 )  &=
    (d_1 : t_1) \mathrel{\mathsf{and}} (d_2 : t_2) \\
  (\Set{(d_1, \domega_1),..., (d_n, \domega_n)} : t_1 \to t_2)  &=
   \forall i \in [1.. n] . \:
    \mathsf{if} \: (d_i : t_1) \mathrel{\mathsf{then}} (\domega_i : t_2) \\
  (d : t_1 \lor t_2)  &= (d : t_1) \mathrel{\mathsf{or}} (d : t_2) \\
  (d : \lnot t)  &= \mathsf{not} \: (d : t) \\
  (\domega : t)  &= \mathsf{false} & \text{ otherwise}
  \end{align*}
We define the \emph{set-theoretic interpretation}
$ \TypeInter{} : \types \to \Pd(\Domain) $
as $ \TypeInter{t} = \{d \in \Domain \mid (d : t)\} $.
\end{definition}
Finally,
we define the subtyping preorder and its associated equivalence relation
as follows.

\begin{definition}[Subtyping relation]\label{def:subtyping}
  We define the \emph{subtyping} relation $ \leq $
  and the \emph{subtyping equivalence} relation $ \simeq $
  as
  \(
    t_1 \leq t_2 \iffdef \TypeInter{t_1} \subseteq \TypeInter{t_2}\) and   
  \(t_1 \simeq t_2 \iffdef (t_1 \leq t_2) \mathrel{\mathsf{and}} (t_2 \leq t_1)
    \: .
  \)
\end{definition}

\section{Declarative type system with extensions}
\label{sec:decl-appendix}

The declarative type system extended with the extensions of Appendix~\ref{sec:extensions-appendix}
uses expressions produced by the following grammar:
\begin{equation*}
\begin{array}{lrclr}
    \textbf{Expressions} &e &::=& c\alt x\alt\lambda x.e\alt e e\alt (e,e)\alt \pi_i e \alt
    \tcase{e}{\tau}{e}{e}\alt \evalexp{e}{e}\alt (e\bbcolon\tau)
\end{array}
\end{equation*}
Note that extended typecases and pattern matching are absent because they are encoded
using let-bindings and type constraints before typing. Similarly, we use the construction
$\evalexp{e}{e}$ for let-bindings instead of the initial construction $\letexp{x}{e}{e}$
in order to avoid aliasing.
You should refer to Section~\ref{sec:let} for more details on this transformation.

The deduction rules for the declarative type system are:
\begin{mathpar}
    \Infer[Const]
    { }
    {\Gamma\vdash c:\basic{c}}
    { }
    \and
    \Infer[Ax]
  { }
  {\Gamma \vdash x: \Gamma(x)}
  { }
  \\
  \Infer[$\to$I]
    {\Gamma,x:\mt\vdash e:t}
    {\Gamma\vdash\lambda x.e: \arrow{\mt}{t}}
    { }
  \qquad
  \Infer[$\to$E]
  {
    \Gamma \vdash e_1: \arrow {t_1}{t_2}\\
    \Gamma \vdash e_2: t_1
  }
  { \Gamma \vdash {e_1}{e_2}: t_2 }
  { }
  \\
  \Infer[$\times$I]
  {\Gamma \vdash e_1:t_1 \and \Gamma \vdash e_2:t_2}
  {\Gamma \vdash (e_1,e_2):\pair {t_1} {t_2}}
  { }
  \qquad
  \Infer[$\times$E$_1$]
  {\Gamma \vdash e:\pair{t_1}{t_2}}
  {\Gamma \vdash \pi_1 e:t_1}
  { } \quad
  \Infer[$\times$E$_2$]
  {\Gamma \vdash e:\pair{t_1}{t_2}}
  {\Gamma \vdash \pi_2 e:t_2}
  { }
  \\  
\Infer[$\Empty$]
{
  \Gamma \vdash e:\Empty
}
{\Gamma\vdash \tcase {e} \tau {e_1}{e_2}: \Empty}
{ }
\qquad
  \Infer[$\in_1$]
  {
    \Gamma \vdash e:\tau \and \Gamma\vdash e_1: t_1
  }
  {\Gamma\vdash \tcase {e} \tau {e_1}{e_2}: t_1}
  { }
  \quad
  \Infer[$\in_2$]
  {
    \Gamma \vdash e:\neg \tau \and \Gamma\vdash e_2: t_2
  }
  {\Gamma\vdash \tcase {e} \tau {e_1}{e_2}: t_2}
  { }
  \\
  \Infer[$\vee$]
  {
    \Gamma \vdash e': s\\
    \Gamma, x:s\land \mt\vdash e:t\\
    \Gamma, x:s\land \neg\mt\vdash e:t
  }
  {
    \Gamma\vdash e\subs x {e'}  : t
  }
  { }
  \qquad
  \Infer[$\wedge$]
  { \Gamma \vdash e: t_1\\ \Gamma \vdash e: t_2 }
  { \Gamma \vdash e: t_1\land t_2 }
  { }
  \\
  \Infer[Inst]
  { \Gamma \vdash e:t }
  { \Gamma \vdash e: t\sigma }
  { }
  \qquad
  \Infer[$\leq$]
  { \Gamma \vdash e:t }
  { \Gamma \vdash e: t' }
  { t\leq t' }
\end{mathpar}

with these additional rules for the extensions of Appendix~\ref{sec:extensions-appendix} (let-bindings and type constraints):
\begin{mathpar}
    \Infer[Let]
    { \Gamma \vdash e_1: t_1\\ \Gamma\vdash e_2: t_2 }
    { \Gamma \vdash \evalexp{e_1}{e_2}: t_2 }
    { }
    \qquad
    \Infer[Constr]
    {\Gamma\vdash e\bbcolon \tau\\\Gamma\vdash e: t}
    {\Gamma\vdash (e\bbcolon\tau): t}
    { }
\end{mathpar}

\section{Computation of MSC-forms}
\label{sec:msc-appendix}

\subsection{From Canonical Forms to Source Language Expressions}\label{app:unwind}

We recall the grammar for canonical forms, with the extensions presented in Appendix~\ref{sec:extensions-appendix}:
\begin{equation*}
  \begin{array}{lrclr}
    \textbf{Atomic expressions} &a &::=& c\alt x\alt \lambda x.\kappa\alt (\xx,\xx)\alt \xx \xx\alt \pi_i \xx\alt \tcase{\xx}{\tau}{\xx}{\xx}
    \alt\evalexp{\xx}{\xx}\alt\xx\bbcolon\tau\\
    \textbf{Canonical Forms} & \kappa &::=& \xx\alt \bindexp{\xx}{a}{\kappa} \\
  \end{array}
\end{equation*}

Any canonical form can be transformed into an expression of the source language
using the unwiding operator $\eras{.}$ defined as follows:

\begin{align*}
  \eras c & = c\\
  \eras x & = x\\
  \eras {\lambda x.\kappa} & = \lambda x.\eras \kappa\\
  \eras {\xx_1 \xx_2} &= \xx_1 \xx_2\\
  \eras {(\xx_1,\xx_2)} &=  (\xx_1,\xx_2)\\
  \eras {\pi_i \xx}&= \pi_i \xx & i=1,2\\
  \eras {\tcase{\xx}{\tau}{\xx_1}{\xx_2}}&= \tcase{\xx}{\tau}{\xx_1}{\xx_2}\\
  \eras {\evalexp{\xx_1}{\xx_2}} &= \letexp{x}{\xx_1}{\xx_2\subs{\xx_1}{x}}\text{ with }x\text{ fresh}\\
  \eras {\xx\bbcolon\tau} &= \xx\bbcolon\tau\\
  \eras {\bindexp{\xx}{a}{\kappa}} &= \eras{\kappa}\subs{\xx}{\eras{a}}\\
  \eras {\xx} &= \xx
\end{align*}

\subsection{From Source Language Expressions to Canonical Forms}\label{app:sl2cf}

The inverse direction, that is, producing from a source language
expression a canonical form that unwinds to it, is also straightforward.

Let $\benv$ denote a binding context, that is, an ordered list of mappings 
from binding variables to atoms. Each mapping is written as a pair $\be \xx e$.
We note these lists extensionally by separating elements by a semicolon, that is,
$\be{\xx_1}{a_1};\dots;\be{\xx_n}{a_n}$ and use $\bempty$ to denote the empty list.

We define an operation $\term(\benv,\kappa)$ which takes a binding context $\benv$ 
and a canonical form $\kappa$ and constructs the canonical form containing the bindings
listed in $\benv$ and ending with $\kappa$, that is:
\begin{align*}
  \term(\varepsilon, \kappa) &\eqdef \kappa\\
  \term((\be\xx a;\benv),\kappa) &\eqdef\bindexp{\xx}{a}{\term(\benv,\kappa)}
\end{align*}

We can now define the function $\ucanon e$ that transforms an expression
$e$ into a pair $(\benv,\xx)$ formed
by a binding context $\benv$ and a binding variable $\xx$ that will be bound to
the atom representing $e$. The definition is as follows, where $\xx_\circ$ is a fresh binding variable.
\begin{align*}
  \ucanon c & = (\be{\xx_\circ}c,\xx_\circ)\\
  \ucanon x & = (\be{\xx_\circ}x,\xx_\circ)\\
  \ucanon {\lambda x.e} & = (\be{\xx_\circ}{\lambda x.\term\ucanon e},\xx_\circ)\\
  \ucanon {\pi_i e}&= ((\benv;\be{\xx_\circ}{\pi_i \xx}),\xx_\circ) &\hspace*{-12mm}\text{
    where } (\benv,\xx) = \ucanon{e}
  \\
  \ucanon {e_1e_2}&=
  ((\benv_1;\benv_2;\be{\xx_\circ}{\xx_1\xx_2}),\xx_\circ)
  &\hspace*{-12mm}\text{where } (\benv_1,\xx_1) = \ucanon{e_1},\; (\benv_2,\xx_2) = \ucanon{e_2}
  \\
  \ucanon {(e_1,e_2)} &=
  ((\benv_1;\benv_2;\be{\xx_\circ}{(\xx_1,\xx_2)}),\xx_\circ)
  &\hspace*{-2mm}\text{where } (\benv_1,\xx_1) = \ucanon{e_1},\; (\benv_2,\xx_2) = \ucanon{e_2}
  \\
  \ucanon {\tcase{e}{\tau}{e_1}{e_2}}
     &= ((\benv;\benv_1;\benv_2;\be{\xx_\circ}{\tcase{\xx}{\tau}{\xx_1}{\xx_2}}),\xx_\circ)\hspace*{-17mm}\\[-1mm]
  \span\span\text{where }  (\benv,\xx) = \ucanon{e},\;
  (\benv_1,\xx_1) = \ucanon{e_1},\;  (\benv_2,\xx_2) = \ucanon{e_2}\\
  \ucanon \xx & = (\bempty,\xx)\\
  \ucanon{\letexp{x}{e_1}{e_2}} & =
  ((\benv_1; \benv_2; \be{\xx_\circ}{\evalexp{\xx_1}{\xx_2}}), \xx_\circ)
  \\[-1mm]
  \span\span
  \text{where }
  (\benv_1, \xx_1) = \ucanon{e_1},
  (\benv_2, \xx_2) = \ucanon{e_2\subs{x}{\xx_1}}\\
  \ucanon{e\bbcolon\tau} & = ((\benv;  \be{\xx_\circ}{ \xx\bbcolon\tau}), \xx_\circ)
  &\hspace*{-12mm}\text{where }(\benv, \xx) = \ucanon{e}
\end{align*}

It is easy to prove that, for any term of the source language $e$,
$\eras{\term(\ucanon e)}= e$.

\subsection{From Canonical Forms to a MSC Form}\label{app:cf2msc}

It is easy to transform a canonical form into a MSC-form that has the same unwinding.
This can be done by applying the rewriting rules below, that are confluent and normalizing.

\begin{align}
  \begin{array}{l}
    \bindexp{\xx_1}{a_1}{}\\
    \bindexp{\xx_2}{a_2}{\kappa} 
  \end{array}
  &\msred\mkern10mu
  \bindexp{\xx_1}{a_1}{\kappa\subs{\xx_2}{\xx_1}}\hspace*{-33mm}
  &a_1\eqcan a_2
  \label{sharing}
\\
\bindexp{\xx}a{\kappa}
  \mkern10mu&\msred\mkern10mu
  \kappa
  &\xx\notin\fv\kappa
  \label{useless}
\\
\begin{array}{l}
\bindexp{\xx}{\lambda y.(\\\mkern10mu\bindexp{\zz}{a}{\kappa_\circ})\\\hspace*{-.5mm}}{\kappa}
\end{array}
   \mkern10mu&\msred\mkern10mu
\begin{array}{l}
\bindexp{\zz}{a}{\\\bindexp{\xx}{\lambda y.\kappa_\circ}{\kappa}}
\end{array}
   &y\notin\fv a,\zz\not\in\fv\kappa\label{extrude}
\end{align}

\begin{align}
\kappa_1\mkern10mu&\msred\mkern10mu\kappa_2&
\exists\kappa_1'.\ \kappa_1\eqcan\kappa_1'\msred\kappa_2\label{modeqcan}
\end{align}
Rule \eqref{sharing} implements the maximal sharing: if two variables
bind atoms with the same unwinding (modulo $\alpha$-conversion), then the variables are unified. Rule
\eqref{useless} removes useless bindings. Rule \eqref{extrude}
extrudes bindings from abstractions of variables that do not occur in
the argument of the binding. Rule \eqref{modeqcan} applies the
previous rule modulo the canonical equivalence: in practice it applies
the swap of binding defined in Definition~\ref{def:order} as many times as it is
needed to apply one of the other rules. As
customary, these rules can be applied under any context.

The transformation above transforms every canonical form into an MSC-form
that has the same unwinding. It thus allows to compute $\MSCF e$ for any expression
$e$ of the source language.

\section{Type operators}
\label{sec:typeop-appendix}
The algorithmic type system presented in this work use the following type-operators:
\begin{eqnarray*}
\dom t & = & \max \{ u \alt t\leq u\to \Any\}
\\
\apply t s & = &\,\min \{ u \alt t\leq s\to u\}
\\
\bpl t & = & \min \{ u \alt t\leq \pair u\Any\}
\\
\bpr t & = & \min \{ u \alt t\leq \pair \Any u\}
\end{eqnarray*}
In words, $\apply t s$ is the best (i.e., smallest wrt $\leq$) type we can deduce for the
application of a function of type $t$ to an argument of type $s$.
Projection and domain are standard.
All these operators can be effectively computed as shown below
(see~\citet{FCB08,CLLN20} for details and proofs).

\textcolor{modification}{If any of the types at issue is empty, then the computation is straightforward:
$\dom\Empty = \Any$ and $\apply\Empty s = \apply t \Empty =
\bpl\Empty=\bpr\Empty=\Empty$. Otherwise the operators are computed as follows. }

For $t\eqdnf \bigvee_{i\in I}\left(\bigwedge_{p'\in P'_i}\alpha_{p'}\land\bigwedge_{n'\in N'_i}\neg\alpha_{n'}'\wedge\bigwedge_{p\in P_i}(s_p\to t_p) \wedge \bigwedge_{n\in N_i}\neg(s_n'\to t_n')\right)$,\\
\textcolor{modification}{where each summand of the outer union is not empty}, the first two operators are computed by:
\begin{eqnarray*}
  \dom{t}    & = & \bigwedge_{i\in I}\bigvee_{p\in P_i}s_p\\[4mm]
  \apply t s & = & \bigvee_{i\in I}\left(\bigvee_{\{Q\subsetneq P_i\alt s\not\leq\bigvee_{q\in Q}s_q\}}\left(\bigwedge_{p\in P_i\setminus Q}t_p\right)\right)\hspace*{1cm}\makebox[0cm][l]{(for $s\leq\dom{t}$)}\\[4mm]
\end{eqnarray*}
For $t\eqdnf \bigvee_{i\in I}\left(\bigwedge_{p'\in P'_i}\alpha_{p'}\land\bigwedge_{n'\in N'_i}\neg\alpha_{n'}'\wedge\bigwedge_{p\in P_i}(s_p, t_p) \wedge \bigwedge_{n\in N_i}\neg(s_n', t_n')\right)$,\\
\textcolor{modification}{where each summand of the outer union is not empty}, the last two operators are computed by
\begin{eqnarray*}
  \bpl t & = & \bigvee_{i \in I} \bigvee_{N' \subseteq N_i} \left(\bigwedge_{p \in P_i} s_p \wedge \bigwedge_{n \in N'} \neg s_n' \right)\\[4mm]
  \bpr t & = & \bigvee_{i \in I} \bigvee_{N' \subseteq N_i} \left(\bigwedge_{p \in P_i} t_p \wedge \bigwedge_{n \in N'} \neg t_n' \right)\\[4mm]
\end{eqnarray*}

\section{Algorithmic type system}\label{app:algosys}
\label{sec:algo-appendix}

\begin{equation*}
\hspace*{-2mm}  \begin{array}{lrclr}
    \textbf{Atom annots}  &\annota &::=& \annotconst 
    \alt \annotlambda{\mt}{\annot}\alt\annotpair\renaming\renaming\alt\annotapp\Sigma\Sigma\alt\annotproj\Sigma
    \alt\annotempty\Sigma\alt\annotthen\Sigma\alt\annotelse\Sigma\alt\annotinter{\{\annota,...\,,\annota\}}\\
    \textbf{Form annots}  &\annot &::=& \annotvar\renaming\alt
    \annotbind{\annota}{\{(\mt,\annot), \dots, (\mt,\annot)\}}\alt\annotskip{\annot}\alt\annotinter{\{\annot,\dots,\annot\}}
    \end{array}
\end{equation*}

The algorithmic type system is defined by the following deduction rules:

\begin{mathpar}
    \Infer[Const\AR]
    { }
    {\Gamma\vdashA \ea c \annotconst:\basic{c}}
    { }
    \quad
    \Infer[Ax\AR]
    { }
    { \Gamma\vdashA \ea x \annotvara: \Gamma(x) }
    { }
    \\
    \Infer[$\to$I\AR]
      {\Gamma,x:\mt\vdashA \ea\kappa{\annot}:t}
      {\Gamma\vdashA\ea{\lambda x.\kappa}{\annotlambda{\mt}{\annot}}:\arrow{\mt}{t}}
      { }
    \\
    \Infer[$\to$E\AR]
    { }
    {\Gamma \vdashA \ea{{\xx_1}{\xx_2}}{\annotapp{\Sigma_1}{\Sigma_2}}: t_1\circ t_2 }
    { \stackedguards{t_1=\Gamma(\xx_1)\Sigma_1,\,\,\,t_2=\Gamma(\xx_2)\Sigma_2}
      {t_1\leq\Empty\to\Any,\,\,\,t_2\leq\dom{t_1}} }
    \\
    \Infer[$\times$I\AR]
    { }
    {\Gamma \vdashA \ea{(\xx_1,\xx_2)}{\annotpair{\renaming_1}{\renaming_2}}:\pair {t_1} {t_2}}
    { t_1=\Gamma(\xx_1)\renaming_1,\,\,\,t_2=\Gamma(\xx_2)\renaming_2 }
    \\
    \Infer[$\times$E$_1$\AR]
    { }
    {\Gamma \vdashA \ea{\pi_1 \xx}{\annotproj\Sigma}:\bpi_1(t)}
    { \stackedguards{t=\Gamma(\xx)\Sigma}
      {t\leq(\Any\times\Any)} }
    \\
    \Infer[$\times$E$_2$\AR]
    { }
    {\Gamma \vdashA \ea{\pi_2 \xx}{\annotproj\Sigma}:\bpi_2(t)}
    { \stackedguards{t=\Gamma(\xx)\Sigma}
      {t\leq(\Any\times\Any)} }
    \\
    \Infer[$\Empty$\AR]
    { }
    {\Gamma\vdashA \ea{\tcase {\xx} \tau {\xx_1}{\xx_2}}{\annotempty\Sigma}: \Empty}
    { \Gamma(\xx)\Sigma\simeq\Empty }
    \\
    \Infer[$\in_1$\AR]
    { }
    {\Gamma\vdashA \ea{\tcase {\xx} \tau {\xx_1}{\xx_2}}{\annotthen\Sigma}: \Gamma(\xx_1)}
    { \Gamma(\xx)\Sigma\leq \tau }
    \\
    \Infer[$\in_2$\AR]
    { }
    {\Gamma\vdashA \ea{\tcase {\xx} \tau {\xx_1}{\xx_2}}{\annotelse\Sigma}: \Gamma(\xx_2)}
    { \Gamma(\xx)\Sigma\leq \neg \tau }
    \\
    \Infer[$\land$\AR]
    {
    {\small(\forall i\in I)}\quad\Gamma\vdashA \ea a{\annota_i}:t_i
    }
    {
      \Gamma\vdashA\ea{a}{\annotinter{\{\annota_i\}_{i\in I}}} : \textstyle\bigwedge_{i\in I}t_i
    }
    { I\neq\emptyset }
  \end{mathpar}
  \begin{mathpar}
    \Infer[Var\AR]
    { }
    { \Gamma\vdashA \ea\xx {\annotvar\renaming}: \Gamma(\xx)\renaming }
    { }
    \quad
    \Infer[Bind$_1$\AR]
    {\Gamma\vdashA \ea{\kappa}{\annot}:t}
    {
    \Gamma\vdashA\ea{\bindexp {\xx} {a} {\kappa}}{\annotskip \annot}:t
    }
    { \xx\not\in\dom\Gamma }
    \\
    \Infer[Bind$_2$\AR]
    {\Gamma\vdashA \ea a{\annota}:s\\
    {\small(\forall i\in I)}\quad\Gamma,\xx:s\wedge\mt_i\vdashA \ea\kappa{\annot_i}:t_i
    }
    {
      \Gamma\vdashA\ea{\bindexp {\xx} {a} {\kappa}}{\annotbind {\annota} {\{(\mt_i,\annot_i)\}_{i\in I}}} : \textstyle\bigvee_{i\in I}t_i
    }
    { \textstyle\bigvee_{i\in I}\mt_i\simeq\Any }
    \\
    \Infer[$\land$\AR]
    {
    {\small(\forall i\in I)}\quad\Gamma\vdashA \ea\kappa{\annot_i}:t_i
    }
    {
      \Gamma\vdashA\ea{\kappa}{\annotinter{\{\annot_i\}_{i\in I}}} : \textstyle\bigwedge_{i\in I}t_i
    }
    { I\neq\emptyset }
  \end{mathpar}

To extend the system to type the extensions presented in
Appendix~\ref{sec:extensions-appendix} the following rules must be added:

\begin{mathpar}
    \Infer[Let\AR]
    { }
    { \Gamma\vdashA \ea {\evalexp{\xx_1}{\xx_2}} \annotvara: \Gamma(\xx_2) }
    { \xx_1\in\dom\Gamma }
\\
    \Infer[Constr\AR]
    { }
    { \Gamma\vdashA \ea {\xx\bbcolon\tau} {\annotannot\Sigma}: \Gamma(\xx) }
    { \Gamma(\xx)\Sigma\leq\tau }
\end{mathpar}

\section{Full reconstruction system}
\label{sec:reconstruction-appendix}

\subsection{Main Reconstruction System}\label{sec:mainreconstruction-appendix}

\begin{equation*}
  \begin{array}{lrclr}
  \textbf{Split annotations}  &\sanns &::=& \{(\mt,\banns),\dots,(\mt,\banns)\}\\
  \textbf{Atoms intermediate annot.}  &\lanns &::=& \annotiinfer
  \alt\annotiuntyp\alt\annotityp
  \alt\annotiinter{\{\lanns,\dots,\lanns\}}{\{\lanns,\dots,\lanns\}}\\&&&
  \alt\annotithen\alt\annotielse\alt\annotilambda{\mt}{\banns}\\
  \textbf{Forms intermediate annot.}  &\banns &::=& \annotiinfer\alt\annotiuntyp\alt\annotityp
  \alt\annotiinter{\{\banns,\dots,\banns\}}{\{\banns,\dots,\banns\}}\\&&&
  \alt\annotitryskip\banns\alt\annotitrybind{\lanns}{\banns}{\banns}\\&&&
  \alt\annotiprop{\lanns}{\Gammas}{\sanns}{\sanns}\\&&&
  \alt\annotiskip{\banns}\alt\annotibind{\lanns}{\sanns}{\sanns}
  \end{array}
\end{equation*}


  \begin{mathpar}
    \Infer[Ok]
    { }
    {\Gamma\aavdash \epa \kappaa\annotityp\refines\resok{\annotityp}}
    { }
    \qquad
    \Infer[Fail]
    { }
    {\Gamma\aavdash \epa \kappaa\annotiuntyp\refines\resfail}
    { }
  \end{mathpar}
  
  \begin{mathpar}
    \Infer[Const]
    { }
    {\Gamma\aavdash \epa c{\annotiinfer}\refines\resok{\annotityp}}
    { }
    \\
    \Infer[AxOk]
    { x\in\dom\Gamma }
    {\Gamma\aavdash \epa x{\annotiinfer}\refines\resok{\annotityp}}
    { }
    \qquad
    \Infer[AxFail]
    { }
    {\Gamma\aavdash \epa x{\annotiinfer}\refines\resfail}
    { }
  \end{mathpar}
  
  \begin{mathpar}
    \Infer[PairVar$_i$]
    { \xx_i\not\in\dom\Gamma }
    {\Gamma\aavdash \epa{(\xx_1,\xx_2)}{\annotiinfer}\refines
    {\resvar{\xx_i}{\annotiinfer}{\annotiuntyp}}}
    { }
    \\
    \Infer[PairOk]
    { }
    {\Gamma\aavdash \epa{(\xx_1,\xx_2)}{\annotiinfer}\refines
    {\resok{\annotityp}}}
    { }
    \\
    \Infer[ProjVar]
    { \xx\not\in\dom\Gamma }
    {\Gamma\aavdash \epa{\pi_i \xx}{\annotiinfer}\refines
    {\resvar{\xx}{\annotiinfer}{\annotiuntyp}}}
    { }
    \\
    \Infer[ProjInfer]
    {
    \tallyinfer{\Msubst}{\Gamma(\xx)}{\polyvar\times\polyvarb}
    }
    {\Gamma\aavdash \epa{\pi_i \xx}{\annotiinfer}\refines
    {\ressubst{\Msubst}{\annotityp}{\annotiuntyp}}}
    { \polyvar, \polyvarb\in\polyvars\text{ fresh} }
  \end{mathpar}
  
  \begin{mathpar}
    \Infer[AppVar$_i$]
    { \xx_i\not\in\dom\Gamma }
    {\Gamma\aavdash \epa{\xx_1 \xx_2}{\annotiinfer}\refines
    {\resvar{\xx_i}{\annotiinfer}{\annotiuntyp}}}
    { }
    \\
    \Infer[AppInfer]
    {
        \tallyinfer{\Msubst}{\Gamma(\xx_1)}{\Gamma(\xx_2)\to \polyvar}
    }
    {\Gamma\aavdash \epa{{\xx_1}{\xx_2}}{\annotiinfer}\refines
    {\ressubst{\Msubst}{\annotityp}{\annotiuntyp}}}
    { \polyvar\in\polyvars\text{ fresh} }
  \end{mathpar}
  
  \begin{mathpar}
    \Infer[CaseVar]
    { \xx\not\in\dom\Gamma }
    {\Gamma\aavdash \epa{\tcase {\xx} \tau {\xx_1}{\xx_2}}{\annotiinfer}\refines
    {\resvar{\xx}{\annotiinfer}{\annotiuntyp}}}
    { }
    \\
    \Infer[CaseSplit]
    { \Gamma(\xx)\not\leq\tau\\\Gamma(\xx)\not\leq\neg\tau }
    {\Gamma\aavdash \epa{\tcase {\xx} \tau {\xx_1}{\xx_2}}{\annotiinfer}\refines
    {\respart{\envsingl \xx \tau}{\annotiinfer}{\annotiinfer}}}
    { }
  \end{mathpar}
  \begin{mathpar}
    \Infer[CaseEmpty]
    {
      \Gamma(\xx)\simeq\Empty
    }
    {\Gamma\aavdash \epa{\tcase {\xx} \tau {\xx_1}{\xx_2}}{\annotiinfer}\refines\resok{\annotityp}}
    { }
    \\
    \Infer[CaseThen]
    {
      \Gamma(\xx)\leq\tau\\
      \tallyinfer{\Msubst}{\Gamma(\xx)}{\Empty}
    }
    {\Gamma\aavdash \epa{\tcase {\xx} \tau {\xx_1}{\xx_2}}{\annotiinfer}\refines\ressubst{\Msubst}{\annotityp}{\annotithen}}
    { }
    \\
    \Infer[CaseElse]
    {
      \Gamma(\xx)\leq\neg\tau\\
      \tallyinfer{\Msubst}{\Gamma(\xx)}{\Empty}
    }
    {\Gamma\aavdash \epa{\tcase {\xx} \tau {\xx_1}{\xx_2}}{\annotiinfer}\refines\ressubst{\Msubst}{\annotityp}{\annotielse}}
    { }
  \end{mathpar}
  \begin{mathpar}
    \Infer[CaseVar$_i$]
    {
      \xx_i\not\in\dom\Gamma
    }
    {\Gamma\aavdash \epa{\tcase {\xx} \tau {\xx_1}{\xx_2}}{\annotithenelse}\refines
    {\resvar{\xx_i}{\annotityp}{\annotiuntyp}}}
    { }
    \\
    \Infer[CaseOk$_i$]
    { }
    {\Gamma\aavdash \epa{\tcase {\xx} \tau {\xx_1}{\xx_2}}{\annotithenelse}\refines
    {\resok{\annotityp}}}
    { }
  \end{mathpar}

  \begin{mathpar}
    \Infer[LambdaInfer]
    {
        \Gamma\aavdash \epa{\lambda x. \kappa}{\annotilambda{\monovar}{\annotiinfer}} \refines\results
    }
    {
        \Gamma\aavdash \epa{\lambda x. \kappa}{\annotiinfer} \refines\results
    }
    { \monovar\in\monovars\text{ fresh} }
    \\
    \Infer[LambdaEmpty]
    { }
    {
        \Gamma\aavdash \epa{\lambda x. \kappa}{\annotilambda{\Empty}{\banns}} \refines\resfail
    }
    { }
    \\
    \Infer[Lambda]
    {
        \Gamma,x:\mt\raavdash \epa{\kappa}{\banns} \refines\results
    }
    {
        \Gamma\aavdash \epa{\lambda x. \kappa}{\annotilambda{\mt}{\banns}} \refines
        {\mapres\results{X}{\annotilambda{\mt}{X}}}
    }
    { }
  \end{mathpar}
  with $\mapres{\results}{X}{f(X)}$ an auxiliary function that applies $f$ to each intermediate annotation in $\results$:
  \begin{align*}
      \mapres{\resok{\aanns}}{X}{f(X)} & \eqdef \resok{f(\aanns)}\\
      \mapres{\resfail{}}{X}{f(X)} & \eqdef \resfail{}\\
      \mapres{\respart{\Gamma}{\aanns_1}{\aanns_2}}{X}{f(X)} & \eqdef \respart{\Gamma}{f(\aanns_1)}{f(\aanns_2)}\\
      \mapres{\ressubst{\Msubst}{\aanns_1}{\aanns_2}}{X}{f(X)} & \eqdef \ressubst{\Msubst}{f(\aanns_1)}{f(\aanns_2)}\\
      \mapres{\resvar{\xx}{\aanns_1}{\aanns_2}}{X}{f(X)} & \eqdef \resvar{\xx}{f(\aanns_1)}{f(\aanns_2)}
  \end{align*}

  \begin{mathpar}
    \Infer[FormVar]
    { \xx\not\in\dom\Gamma }
    { \Gamma\aavdash \epa\xx\annotiinfer\refines{\resvar{\xx}{\annotiinfer}{\annotiuntyp}}}
    { }
    \\
    \Infer[FormOk]
    { }
    { \Gamma\aavdash \epa\xx\annotiinfer\refines{\resok{\annotityp}}}
    { }
  \end{mathpar}
  
  \begin{mathpar}
    \Infer[BindInfer]
    {
        \Gamma\aavdash\epa{\bindexp \xx a \kappa}{\annotitryskip{\annotiinfer}}\refines\results
    }
    {
        \Gamma\aavdash\epa{\bindexp \xx a \kappa}{\annotiinfer}\refines\results
    }
    { }
  \end{mathpar}
  
  \begin{mathpar}
    \Infer[BindTrySkip$_1$]
    {
        \Gamma\raavdash \epa\kappa\banns\refines\resvar{\xx}{\banns_1}{\banns_2}\\
        \Gamma\aavdash \epa{\bindexp \xx a \kappa}{\annotitrybind{\annotiinfer}{\banns_1}{\banns_2}}\refines\results
    }
    {
        \Gamma\aavdash\epa{\bindexp \xx a \kappa}{\annotitryskip{\banns}}\refines\results
    }
    { }
    \\
    \Infer[BindTrySkip$_2$]
    {
        \Gamma\raavdash \epa\kappa\banns\refines\resok{\banns'}
    }
    {
        \Gamma\aavdash\epa{\bindexp \xx a \kappa}{\annotitryskip{\banns}}\refines\resok{\annotiskip{\banns'}}
    }
    { }
    \\
    \Infer[BindTrySkip$_3$]
    {
        \Gamma\raavdash \epa\kappa\banns\refines\results
    }
    {
        \Gamma\aavdash\epa{\bindexp \xx a \kappa}{\annotitryskip{\banns}}\refines\mapres{\results}{X}{\annotitryskip{X}}
    }
    { }
  \end{mathpar}

  \begin{mathpar}
    \Infer[BindSkip$_1$]
    {
        \Gamma\raavdash \epa\kappa\banns\refines\resvar{\xx}{\banns_1}{\banns_2}\\
        \Gamma\aavdash \epa{\bindexp \xx a \kappa}{\annotiskip{\banns_2}}\refines\results
    }
    {
        \Gamma\aavdash\epa{\bindexp \xx a \kappa}{\annotiskip{\banns}}\refines\results
    }
    { }
    \\
    \Infer[BindSkip$_2$]
    {
        \Gamma\raavdash \epa\kappa\banns\refines\results
    }
    {
        \Gamma\aavdash\epa{\bindexp \xx a \kappa}{\annotiskip{\banns}}\refines\mapres{\results}{X}{\annotiskip{X}}
    }
    { }
  \end{mathpar}

  \begin{mathpar}
    \Infer[BindTryKeep$_1$]
    {
        \Gamma\raavdash \epa a\lanns\refines\resok{\lanns'}\\
        \Gamma\aavdash\epa{\bindexp \xx a \kappa}{\annotibind{\lanns'}{\{(\Any,\banns_1)\}}{\emptyset}}\refines\results
    }
    {
        \Gamma\aavdash\epa{\bindexp \xx a \kappa}{\annotitrybind{\lanns}{\banns_1}{\banns_2}}\refines\results
    }
    { }
    \\
    \Infer[BindTryKeep$_2$]
    {
        \Gamma\raavdash \epa a\lanns\refines\resfail\\
        \Gamma\aavdash\epa{\bindexp \xx a \kappa}{\annotiskip{\banns_2}}\refines\results
    }
    {
        \Gamma\aavdash\epa{\bindexp \xx a \kappa}{\annotitrybind{\lanns}{\banns_1}{\banns_2}}\refines\results
    }
    { }
    \\
    \Infer[BindTryKeep$_3$]
    {
        \Gamma\raavdash \epa a\lanns\refines\results
    }
    {
        \Gamma\aavdash\epa{\bindexp \xx a \kappa}{\annotitrybind{\lanns}{\banns_1}{\banns_2}}\refines
        \results'
    }
    { }
  \end{mathpar}
  where $\results'=\mapres{\results}{X}{\annotitrybind{X}{\banns_1}{\banns_2}}$.

  \begin{mathpar}
    \Infer[BindOk]
    { }
    {
      \Gamma\aavdash\epa{\bindexp \xx a \kappa}{\annotibind{\lanns}{\emptyset}{\sanns}}
      \refines\resok{\annotibind{\lanns}{\emptyset}{\sanns}}
    }
    { }
  \end{mathpar}

  \begin{mathpar}
    \Infer[BindKeep$_1$]
    {
      \Gamma\paavdash \epa a{\lanns}\refines\annota\\\Gamma\vdashA \ea a{\annota}:s\\
      \Gamma, \xx:s\land\mt\raavdash \epa\kappa\banns\refines\resok{\banns'}\\
      \Gamma\aavdash\epa{\bindexp \xx a \kappa}{\annotibind{\lanns}{\sanns}{\{(\mt,\banns')\}\cup\sanns'}}\refines\results
    }
    {
      \Gamma\aavdash\epa{\bindexp \xx a \kappa}{\annotibind{\lanns}{\{(\mt,\banns)\}\cup \sanns}{\sanns'}}
      \refines{\results}
    }
    { }
    \\
    \Infer[BindKeep$_2$]
    {
      \Gamma\paavdash \epa a{\lanns}\refines\annota\\\Gamma\vdashA \ea a{\annota}:s\\
      \Gamma, \xx:s\land\mt\raavdash\epa{\kappa}{\banns}\refines\respart{\Gamma'}{\banns_1}{\banns_2}\\
      \xx\in\dom{\Gamma'}\\
      \Gamma\eaavdash\ety{a}{\neg (\mt\land\Gamma'(\xx))}\refines\Gammas_1\\
      \Gamma\eaavdash\ety{a}{\neg (\mt\setminus\Gamma'(\xx))}\refines\Gammas_2
    }
    {
      \Gamma\aavdash\epa{\bindexp \xx a \kappa}{\annotibind{\lanns}{\{(\mt,\banns)\}\cup \sanns}{\sanns'}}
      \refines\respart{\Gamma'\setminus\xx}{\banns_1'}{\banns_2'}
    }
    { }
  \end{mathpar}
  where:
  \begin{itemize}
    \item $\banns_1'=\annotiprop{\lanns}{\Gammas_1\cup\Gammas_2}{\{(\mt\land\Gamma'(\xx),\banns_1),\, (\mt\setminus\Gamma'(\xx),\banns_2)\}\cup\sanns}{\sanns'}$
    \item $\banns_2'=\annotibind{\lanns}{\{(\mt,\banns_2)\}\cup \sanns}{\sanns'}$
  \end{itemize}
  \begin{mathpar}
    \Infer[BindKeep$_3$]
    {
      \Gamma\paavdash \epa a{\lanns}\refines\annota\\\Gamma\vdashA \ea a{\annota}:s\\
      \Gamma, \xx:s\land\mt\raavdash\epa\kappa\banns\refines\results
    }
    {
      \Gamma\aavdash\epa{\bindexp \xx a \kappa}{\annotibind{\lanns}{\{(\mt,\banns)\}\cup \sanns}{\sanns'}}
      \refines\results'
    }
    { }
  \end{mathpar}
  where $\results'={\mapres{\results}{X}{\annotibind{\lanns}{\{(\mt,X)\}\cup \sanns}{\sanns'}}}$.
 
  \begin{mathpar}
    \Infer[BindProp$_1$]
    {
      \Gamma'\in\Gammas\\\compatible{\Gamma}{\Gamma'}
    }
    {
      \Gamma\aavdash\epa{\bindexp \xx a \kappa}{\annotiprop{\lanns}{\Gammas}{\sanns}{\sanns'}}
      \refines\respart{\Gamma''}{\banns_1}{\banns_2}
    }
    { }
  \end{mathpar}
  where:
  \small\begin{itemize}
    \item $\compatible{\Gamma}{\Gamma'} \Leftrightarrow (\dom{\Gamma'}\subseteq\dom\Gamma) \text{ and } (\forall \xx\in\dom{\Gamma'}.\ (\Gamma(\xx)\land\Gamma'(\xx)\not\simeq\Empty) \text{ or }(\Gamma(\xx)\simeq\Empty))$
    \item $\Gamma''=\{(\xx:\mt)\in\Gamma'\,\alt\,\Gamma(\xx)\not\leq\mt\}$
    \item $\banns_1=\annotibind{\lanns}{\sanns}{\sanns'}$
    \item $\banns_2=\annotiprop{\lanns}{\Gammas\setminus\{\Gamma'\}}{\sanns}{\sanns'}$
  \end{itemize}
  \normalsize\begin{mathpar}
    \Infer[BindProp$_2$]
    {
      \Gamma\aavdash\epa{\bindexp \xx a \kappa}{\annotibind{\lanns}{\sanns}{\sanns'}}
      \refines\results
    }
    {
      \Gamma\aavdash\epa{\bindexp \xx a \kappa}{\annotiprop{\lanns}{\Gammas}{\sanns}{\sanns'}}
      \refines\results
    }
    { }
  \end{mathpar}

  \begin{mathpar}
    \Infer[InterEmpty]
    { }
    {
        \Gamma\aavdash \epa{\kappaa}{\annotiinter{\emptyset}{\emptyset}} \refines \resfail
    }
    { }
    \\
    \Infer[InterOk]
    { }
    {
        \Gamma\aavdash \epa{\kappaa}{\annotiinter{\emptyset}{S}} \refines
        \resok{\annotiinter{\emptyset}{S}}
    }
    { }
    \\
    \Infer[Inter$_1$]
    {
        \Gamma\raavdash \epa\kappaa{\aanns}\refines\resok{\aanns'}\\
        \Gamma\aavdash \epa{\kappaa}{\annotiinter{S}{\{\aanns'\}\cup S'}}\refines\results
    }
    {
        \Gamma\aavdash \epa{\kappaa}{\annotiinter{\{\aanns\}\cup S}{S'}} \refines\results
    }
    { }
    \\
    \Infer[Inter$_2$]
    {
        \Gamma\raavdash \epa\kappaa{\aanns}\refines\resfail\\
        \Gamma\aavdash \epa{\kappaa}{\annotiinter{S}{S'}}\refines\results
    }
    {
        \Gamma\aavdash \epa{\kappaa}{\annotiinter{\{\aanns\}\cup S}{S'}} \refines\results
    }
    { }
    \\
    \Infer[Inter$_3$]
    {
        \Gamma\raavdash \epa\kappaa\aanns\refines\results
    }
    {
        \Gamma\aavdash \epa{\kappaa}{\annotiinter{\{\aanns\}\cup S}{S'}} \refines
        \mapres{\results}{X}{{(\annotiinter{\{X\}\cup S}{S'})}}
    }
    { }
  \end{mathpar}
  

  \begin{mathpar}
    \Infer[Iterate$_1$]
    { \Gamma\aavdash\epa{\kappaa}{\aanns}\refines\respart{\Gamma'}{\aanns_1}{\aanns_2}\\
      \Gamma\raavdash\epa{\kappaa}{\aanns_1}\refines\results'
    }
    { \Gamma\raavdash\epa{\kappaa}{\aanns}\refines\results' }
    {\Gamma'=\emptyenv}
    \\
    \Infer[Iterate$_2$]
    { 
      \Gamma\aavdash \epa\kappaa\aanns\refines\ressubst{\{\msubst_i\}_{i\in I}}{\aanns_1}{\aanns_2}\\
      \Gamma\raavdash\epa{\kappaa}{\annotiinter{\{{\aanns_1\msubst_i}\}_{i\in I}\cup\{\aanns_2\}}{\emptyset}}\refines\results'
    }
    { \Gamma\raavdash\epa{\kappaa}{\aanns}\refines\results' }
    {\forall i\in I.\ \msubst_i\disjoint\Gamma }
    \\
    \Infer[Stop]
    { \Gamma\aavdash\epa{\kappaa}{\aanns}\refines\results }
    { \Gamma\raavdash\epa{\kappaa}{\aanns}\refines\results }
    { }
  \end{mathpar}
  with $\aanns\msubst$ denoting the intermediate annotation $\aanns$ in which the
  substitution $\msubst$ has been applied recursively to every type (in $\lambda$-abstraction annotations and binding annotations).

  The following rules can be added to support the extensions presented in Appendix~\ref{sec:extensions-appendix}:

\begin{mathpar}
  \Infer[LetVar$_i$]
  { \xx_i\not\in\dom\Gamma }
  { \Gamma \aavdash \epa{\evalexp{\xx_1}{\xx_2}}{\annotiinfer} \refines \resvar{\xx_i}{\annotiinfer}{\annotiuntyp}}
  { }
  \\
  \Infer[LetOk]
  { }
  { \Gamma \aavdash \epa{\evalexp{\xx_1}{\xx_2}}{\annotiinfer} \refines \resok{\annotityp}}
  { }
\end{mathpar}

\begin{mathpar}
  \Infer[ConstrVar]
  { \xx\not\in\dom\Gamma }
  {\Gamma\aavdash \epa{\xx\bbcolon\tau}{\annotiinfer}\refines\resvar{\xx}{\annotiinfer}{\annotiuntyp}}
  { }
  \\
  \Infer[ConstrInfer]
  {
  \tallyinfer{\Msubst}{\Gamma(\xx)}{\tau}
  }
  {\Gamma\aavdash \epa{\xx\bbcolon\tau}{\annotiinfer}\refines
  \ressubst{\Msubst}{\annotityp}{\annotiuntyp}
  }
  { }
\end{mathpar}

\subsection{Auxiliary Reconstruction System}\label{sec:auxreconstruction-appendix}
  
In the following, $\arenaming{t}$ denotes a renaming from $\vars t\cap\polyvars$ to fresh polymorphic variables.

  \begin{mathpar}
    \Infer[Const]
    { }
    {\Gamma\paavdash \epa c{\annotityp}\refines\annotconst}
    { }
    \qquad
    \Infer[Ax]
    { }
    {\Gamma\paavdash \epa x{\annotityp}\refines\annotvara}
    { x\in\dom\Gamma }
  \end{mathpar}
  \begin{mathpar}
    \Infer[Pair]
    { \renaming_1=\arenaming{\Gamma(\xx_1)}\\\renaming_2=\arenaming{\Gamma(\xx_2)} }
    {\Gamma\paavdash \epa{(\xx_1,\xx_2)}{\annotityp}\refines\annotpair{\renaming_1}{\renaming_2}}
    { }
    \\
    \Infer[Proj]
    {
        \tally{\Sigma}{\Gamma(\xx)}{\polyvar\times\polyvarb}
    }
    {\Gamma\paavdash \epa{\pi_i \xx}{\annotityp}\refines\annotproj{\Sigma}}
    {\begin{array}{l}\Sigma\neq\emptyset\\[-1mm] \polyvar, \polyvarb\in\polyvars\text{ fresh}\end{array}}
    \\
    \Infer[App]
    {
        t_1=\Gamma(\xx_1)\\t_2=\Gamma(\xx_2)\\
        \renaming_1=\arenaming{t_1}\\\renaming_2=\arenaming{t_2}\\
        \tally{\Sigma}{t_1\renaming_1}{t_2\renaming_2\to \polyvar}
    }
    {\Gamma\paavdash \epa{{\xx_1}{\xx_2}}{\annotityp}\refines
    \annotapp{\{\sigma\circ\renaming_1\,\alt\,\sigma\in\Sigma\}}{\{\sigma\circ\renaming_2\,\alt\,\sigma\in\Sigma\}}}
    {\begin{array}{l}\Sigma\neq\emptyset\\[-1mm] \polyvar\in\polyvars\text{ fresh}\end{array} }  
  \end{mathpar}
  
  \begin{mathpar}
    \Infer[Case$_\Empty$]
    { \sigma\in\tallyf{\Gamma(\xx)}{\Empty} }
    {\Gamma\paavdash \epa{\tcase {\xx} \tau {\xx_1}{\xx_2}}{\annotityp}\refines\annotempty{\{\sigma\}}}
    { }
    \\
    \Infer[Case$_1$]
    { \sigma\in\tallyf{\Gamma(\xx)}{\tau} }
    {\Gamma\paavdash \epa{\tcase {\xx} \tau {\xx_1}{\xx_2}}{\annotityp}\refines\annotthen{\{\sigma\}}}
    { \xx_1\in\dom\Gamma }
    \\
    \Infer[Case$_2$]
    { \sigma\in\tallyf{\Gamma(\xx)}{\neg\tau} }
    {\Gamma\paavdash \epa{\tcase {\xx} \tau {\xx_1}{\xx_2}}{\annotityp}\refines\annotelse{\{\sigma\}}}
    { \xx_2\in\dom\Gamma }
  \end{mathpar}
  
  \begin{mathpar}
    \Infer[Lambda]
    {\Gamma,x:\mt\paavdash \epa\kappa{\banns}\refines\annot}
    {\Gamma\paavdash\epa{\lambda x.\kappa}{\annotilambda{\mt}{\banns}}
      \refines \annotlambda{\mt}{\annot}}
    { }
  \end{mathpar}
  
  \begin{mathpar}
    \Infer[Var]
    { \renaming=\arenaming{\Gamma(\xx)} }
    { \Gamma\paavdash \epa\xx \annotityp: \annotvar\renaming }
    { }
    \quad
    \Infer[BindSkip]
    {\Gamma\paavdash \epa\kappa\banns\refines \annot}
    {
    \Gamma\paavdash\epa{\bindexp {\xx} {a} {\kappa}} {\annotiskip \banns}\refines\annotskip\annot
    }
    { \xx\not\in\dom\Gamma }
    \\
    \Infer[BindKeep]
    {
    \Gamma\paavdash \epa a \lanns\refines\annota\qquad
    \Gamma\vdashA \ea a \annota:s\qquad 
    {\small(\forall i\in I)}\ \ \Gamma,\xx:s\land\mt_i\paavdash \epa\kappa{\banns_i}\refines \annot_i\\
    }
    {
    \Gamma\paavdash\epa{\bindexp {\xx} {a} {\kappa}} {\annotibind \lanns{\emptyset}{\{(\mt_i,\banns_i)\}_{i\in I}}}
    \refines \annotbind \annota {\{(\mt_i,\annot_i)\}_{i\in I}}
    }{(*)}  
  \end{mathpar}
  where $(*)$ is $\tbvee_{i\in I}\mt_i\simeq\Any$.
    
  \begin{mathpar}
    \Infer[Inter]
    {
    {\small(\forall i\in I)}\quad\Gamma\paavdash \epa\kappaa{\aanns_i}\refines\aannot_i
    }
    {
      \Gamma\paavdash\epa{\kappaa}{\annotiinter{\emptyset}{\{\aanns_i\}_{i\in I}}}
      \refines \annotinter {\{\aannot_i\}_{i\in I}}
    }
    { I \neq \emptyset }
  \end{mathpar}

  The following rules can be added to support the extensions presented in Appendix~\ref{sec:extensions-appendix}:

\begin{mathpar}
    \Infer[Let]
    { }
    {\Gamma\paavdash \epa{\evalexp{\xx_1}{\xx_2}}{\annotityp}\refines\annotvara}
    { }
\qquad
  \Infer[Constr]
  {
      \sigma\in\tallyf{\Gamma(\xx)}{\tau}
  }
  {\Gamma\paavdash \epa{\xx\bbcolon\tau}{\annotityp}\refines\annotannot{\{\sigma\}}}
  { }
\end{mathpar}

\subsection{Split Propagation System}\label{sec:envrefinement-appendix}

The split propagation system defined in this section tries to deal with the following problem:
\textit{given an environment $\Gamma$, an atom $a$ and a type $t$,
what additional assumptions can be made on $\Gamma$ in order to ensure that $a$ has type $t$}?
It is used by the main reconstruction system in order to propagate splits made by bindings.

  \begin{mathpar}
    \Infer[Const$_1$]
    { \basic{c}\leq\mt }
    { \Gamma\eaavdash \ety{c}\mt \refines \{\emptyenv\} }
    { }
    \qquad
    \Infer[Const$_2$]
    { }
    { \Gamma\eaavdash \ety{c}\mt \refines \{ \} }
    { }
    \\
    \Infer[Ax$_1$]
    { \Gamma(x)\leq\mt }
    { \Gamma\eaavdash \ety{x}\mt \refines \{\emptyenv\} }
    { }
    \qquad
    \Infer[Ax$_2$]
    { }
    { \Gamma\eaavdash \ety{x}\mt \refines \{ \} }
    { }
    \\
    \Infer[Proj$_1$]
    { }
    { \Gamma\eaavdash \ety{\pi_1 \xx}\mt \refines \{\{\xx:\mt\times\Any\}\} }
    { }
    \quad
    \Infer[Proj$_2$]
    { }
    { \Gamma\eaavdash \ety{\pi_2 \xx}\mt \refines \{\{\xx:\Any\times \mt\}\} }
    { }
    \\
    \Infer[Pair]
    { \mt\eqdnf (\textstyle\bigvee_{i\in I}(\pair{\mt_i}{\ms_i})) \vee \dots }
    { \Gamma\eaavdash \ety{(\xx_1,\xx_2)}\mt \refines \{\{\xx_1:\mt_i\}\land\{\xx_2:\ms_i\}\,\alt\,i\in I\} }
    { }  
    \\
    \Infer[Case]
    { }
    { \Gamma\eaavdash \ety{\tcase {\xx} \tau {\xx_1}{\xx_2}}\mt \refines \{\{\xx:\tau,\ \xx_1:\mt\}, \{\xx:\neg\tau,\ \xx_2:\mt\}\} }
    { }
  \end{mathpar}
  \begin{mathpar}
    \Infer[App]
    {
      \Gamma(\xx_1)\eqdnf \tbvee_{i\in I}t_i\\
      \forall i\in I.\ \{\psubst_j\}_{j\in J_i}=\tallyf{t_i}{\polyvar\to\mt}
    }
    { \Gamma\eaavdash \ety{{\xx_1}{\xx_2}}\mt \refines
      \textstyle\bigcup_{i\in I}\Gammas_i
    }
    { \polyvar\in\Polyvars\text{ fresh} }
  \end{mathpar}
  where, for every $i\in I$, $\Gammas_i=\{\{\xx_1:(t_i\psubst_j)\psubst_j',\ \xx_2:(\polyvar\psubst_j)\psubst_j'\}\,\alt\,j\in J_i\}$
  with $\psubst_j'$ a type substitution mapping each polymorphic type variable $\polyvarb$ appearing in
  $t_i\psubst_j$ or $\polyvar\psubst_j$ to either:
  \begin{itemize}
    \item $\Any$ if $\polyvarb$ only appears in covariant positions in $\polyvar\psubst_j$,
    \item $\Empty$ if $\polyvarb$ only appears in contravariant positions in $\polyvar\psubst_j$,
    \item a fresh monomorphic type variable otherwise.
  \end{itemize}
  \begin{mathpar}
    \Infer[Lambda]
    { }
    { \Gamma\eaavdash \ety{\lambda x.\ \kappa}\mt \refines \{ \} }
    { }
  \end{mathpar}
  
The following rules can be added to support the extensions presented in Appendix~\ref{sec:extensions-appendix}:

\begin{mathpar}
    \Infer[Let]
    { }
    { \Gamma\eaavdash \ety{\evalexp{\xx_1}{\xx_2}}\mt \refines \{\{\xx_2:\mt\}\} }
    { }
\qquad
  \Infer[Constr]
  { }
  { \Gamma\eaavdash \ety{\xx\bbcolon\tau}\mt \refines \{\{\xx:\mt\}\} }
  { }
\end{mathpar}

\section{Proofs}
\label{sec:proofs-appendix}

The proofs are for the source language presented in section~\ref{sec:types} without extension:
\begin{equation*}
    \begin{array}{lrclr}
        \textbf{Expressions} &e &::=& c\alt x\alt\lambda x.e\alt e e\alt (e,e)\alt \pi_i e \alt
        \tcase{e}{\tau}{e}{e}
    \end{array}
\end{equation*}

We fix some notations relative to substitutions:
\begin{itemize}
    \item $\subst$ ranges over substitutions from type variables ($\Polyvars\disjcup\Monovars$) to types
    \item $\renaming$ ranges over renamings of polymorphic variables, that is, injective substitutions from $\polyvars$ to $\polyvars$
    \item $\sigma$ ranges over substitutions from polymorphic type variables $\polyvars$ to types
    \item $\Sigma$ ranges over sets of substitutions from polymorphic type variables $\polyvars$ to types
    \item $\msubst$ ranges over substitutions from monomorphic type variables $\monovars$ to monomorphic types
    \item $\Msubst$ ranges over sets of substitutions from monomorphic type variables $\monovars$ to monomorphic types
\end{itemize}

\subsection{A Canonical Form for the Derivations}\label{sec:normalisation}

Derivations for the declarative type system can have many shapes
(see Appendix~\ref{sec:decl-appendix} for the full declarative system,
without the rules for extensions).
In particular, the union elimination rule \Rule{$\vee$E} can be used anywhere in
the derivation and changes the expression to type by performing a substitution
on it. Other non-structural rules such as \Rule{$\wedge$}, \Rule{$\leq$} and \Rule{Inst}
can also be applied anywhere in the derivation.
In this section, we will define canonical derivations that restrict the use of those rules.
This will then be used, in Section~\ref{sec:typesafety}, to establish a type safety theorem.

\subsubsection{Alternative Form of the Declarative Type System}

In order to be able to express our normalisations lemmas, we first need to slightly modify some rules of the declarative
type system. In order to avoid confusions, the modified declarative type system will use this turnstile symbol: $\vdasha$.

First, we modify the \Rule{Ax} rule so that it can perform a renaming of the polymorphic type variables
in $\Gamma(x)$:

\begin{mathpar}
    \Infer[Ax]
    { }
    {\Gamma \vdasha x: \Gamma(x)\renaming}
    { }
\end{mathpar}

This new \Rule{Ax} rule is derivable in the initial delcarative type system by composing
a \Rule{Ax} rule and a \Rule{Inst} rule. Still, allowing the \Rule{Ax} rule to perform a renaming
of polymorphic type variables is useful, as it allows to uncorrelate types
without resorting to the \Rule{Inst} rule.
For instance, consider the pair $(x,x)$ with $x$ having the type $\polyvar\to\polyvar$.
While this pair could be typed $\pair{(\polyvar\to\polyvar)}{(\polyvar\to\polyvar)}$,
this type does not allow instantiating the left-hand side and right-hand side of the product independently.
A better type would be $\pair{(\polyvar\to\polyvar)}{(\polyvarb\to\polyvarb)}$,
and with this new \Rule{Ax} rule, it can be derived without having to use a \Rule{Inst} rule.
This way, the \Rule{Inst} rule can be reserved to cases that require non trivial
instantiations (i.e. not just renamings). Note that the necessity of performing this renaming
comes from the fact that we do not use type schemes $\forall \vec\alpha.\ t$, where renaming
of the type variables in $\vec\alpha$ can be performed implicitely anywhere.

Secondly, we use a \Rule{$\wedge$} rule of multiple arity instead of a binary one:
\begin{mathpar}
    \Infer[$\wedge$]
    { \small(\forall i\in I)\quad\Gamma \vdasha e:t_i }
    {\Gamma \vdasha e: \tbwedge_{i\in I} t_i }
    { I\neq\emptyset }
\end{mathpar}

This allows to combine successive \Rule{$\wedge$} rule applications
into one \Rule{$\wedge$} rule, making the normalisation lemmas easier to express.
This new \Rule{$\wedge$} rule is admissible in the $\vdash$ system: it can be replaced by
several consecutive \Rule{$\wedge$} nodes. 

Similarly, we will use a \Rule{$\vee$} rule of multiple arity:
\begin{mathpar}
    \Infer[$\vee$]
    {  
      \Gamma \vdasha e': s\\
      \small(\forall i\in I)\quad\Gamma, x:s\land \mt_i\vdasha e:t
    }
    {
    \Gamma\vdasha e\subs x {e'}  : t
    }
    { \{\mt_i\}_{i\in I}\in\Partitions{\Any} }
\end{mathpar} with $\Partitions{t}$ denoting the set of partitions of the type $t$,
that is, the set of all sets $\{t_i\}_{i\in I}$ such that:
$(i)$ $\tbvee_{i\in I}t_i\simeq t$, $(ii)$ $\forall i\in I.\ t_i\not\simeq\Empty$, and
$(iii)$ $\forall i,j\in I.\ i\neq j\Rightarrow t_i\land t_j\simeq\Empty$. The guard condition
is most of time omitted, for concision.

This allows to combine successive \Rule{$\vee$} rule applications substituting the same sub-expression
into one \Rule{$\vee$} rule, making the normalisation lemmas easier to express.
Again, this new \Rule{$\vee$} rule is admissible. For instance, the following derivation:
\small\begin{mathpar}
  \Infer[$\vee$]
  {  
    \Infer{ A }{\Gamma \vdasha e': s}{ }
    \Infer{ B }{\Gamma, y:s\land \mt_1\vdasha e:t}{ }
    \Infer{ C }{\Gamma, y:s\land \mt_2\vdasha e:t}{ }
    \Infer{ D }{\Gamma, y:s\land \mt_3\vdasha e:t}{ }
  }  
  {
    \Gamma\vdasha e \subs x {e'}  : t
  }
  { }
\end{mathpar}\normalsize
can be transformed to use only two binary \Rule{$\vee$} rules:
\small\begin{mathpar}
    \Infer[$\vee$]
    {  
      \Infer{ A }{\Gamma \vdasha e': s}{ }
      \Infer{ B }{\Gamma, x:s\land \mt_1\vdasha e:t}{ }
      \Infer{ X }{\Gamma, x:s\land \neg\mt_1\vdasha e:t}{ }
    }  
    {
      \Gamma\vdasha e \subs x {e'}  : t
    }
    { }
    \\\text{with }X\text{ being the following derivation:}\\
    \Infer[$\vee$]
    {
    [\text{Ax}]\quad
    \Infer{ C\subs x y }{\Gamma, x:s\land \neg\mt_1, y:s\land \mt_2\vdasha e\subs x y:t}{ }
    \Infer{ D\subs x y }{\Gamma, x:s\land \neg\mt_1, y:s\land \mt_3\vdasha e\subs x y:t}{ }
    }  
    {
      \Gamma, x:s\land \neg\mt_1\vdasha (e\subs x y)\subs y x:t
    }
    { }
\end{mathpar}\normalsize
This construction can be generalized for a partition of $\Any$ of any cardinality.

\newcommand{\Axl}[0]{Ax$_\lambda$}
\newcommand{\Axv}[0]{Ax$_\vee$}

Lastly, we distinguish variables that are introduced by
a \Rule{$\to$I} node from variables introduced by a \Rule{$\vee$} node.
The formers are called \textit{lambda variables}, the set of all lambda variables is denoted by $\LVars$
and ranged over by $x$, $y$, and $z$.
The latters are called \textit{binding variables}, the set of all binding variables is denoted by $\BVars$
and ranged over by $\xx$, $\yy$, and $\zz$.
$\LVars$ and $\BVars$ form a partition of the set of variables $\LBVars$ (formally, $\LBVars = \LVars\disjcup\BVars$).
The syntax of expressions and the rules of the type system are changed accordingly as follows:
\begin{equation*}
    \begin{array}{lrclr}
      \textbf{Expressions} &e &::=& c\alt x\alt\xx\alt\lambda x.e\alt e e\alt (e,e)\alt \pi_i e \alt\tcase{e}{\tau}{e}{e}\\
    \end{array}
\end{equation*}
\begin{mathpar}
    \Infer[\Axl]
    { }
    {\Gamma \vdasha x: \Gamma(x)\renaming}
    { }
    \qquad
    \Infer[\Axv]
    { }
    {\Gamma \vdasha \xx: \Gamma(\xx)\renaming}
    { }
    \\
    \Infer[$\vee$]
    {  
    \Gamma \vdasha e': s\\
    \small(\forall i\in I)\quad\Gamma, \xx:s\land \mt_i\vdasha e:t
    }
    {
    \Gamma\vdasha e\subs \xx {e'}  : t
    }
    { \{\mt_i\}_{i\in I}\in\Partitions{\Any} }
\end{mathpar}

When needed, we will use the notation $\xxx$ to range over both binding variables and lambda variables.
For instance, we could write the proposition $\forall \xxx\in\dom\Gamma.\ \Gamma(\xxx)\not\simeq\Empty$.
We say that an expression $e$ is a ground expression if $e$ does not contain any binding variable,
and that a derivation $D$ is a ground derivation if it derives a judgement for a ground expression.
For what concerns programs, they use lambda variables for top-level definitions,
and are only composed of ground expressions.

This new system is equivalent to the initial type system:
the combination of both \Rule{\Axl} and \Rule{\Axv}
gives the previous \Rule{Ax} rule.

A full declarative type system with these modifications is presented in Figure~\ref{fig:alt_declarative}.

\begin{figure}[t]
  \input{proofs/fig-declarative-alt}
  \caption{Alternative Declarative Type System\label{fig:alt_declarative}}
\end{figure}

The rules \Rule{Const}, \Rule{\Axl}, \Rule{$\to$I},
\Rule{$\to$E}, \Rule{$\times$I}, \Rule{$\times$E$_1$}, \Rule{$\times$E$_2$}, \Rule{$\Empty$},
\Rule{$\in_1$} and \Rule{$\in_2$} will be called \textit{structural rules} as their use is guided
by the structure of the expression to type,
each of them allowing to type a specific syntactic construction.
In particular, note that the rule \Rule{\Axv} is not considered structural as binding variables
$\xx$ are not supposed to appear in the initial expression (they are only introduced in the derivation
when using a \Rule{$\vee$} rule).

Also, the first premise of a \Rule{$\vee$} rule will be called its \textit{definition premise},
and its others premises will be called \textit{body premises}.

All the proofs in the next sections and chapters will use the $\vdasha$ declarative type system,
which is equivalent to the $\vdash$ type system.

\begin{proposition}
  For any ground expression $e$, type environment $\Gamma$ and type $t$:
  \begin{align*}
    \Gamma\vdash e:t &\Leftrightarrow \Gamma\vdasha e:t
  \end{align*}
\end{proposition}
\begin{proof}
  The $\Rightarrow$ direction is trivial. The $\Leftarrow$ direction is obtained
  by using \Rule{Inst} nodes to rename polymorphic type variables of axioms whenever needed, and by locally
  transforming $n$-ary \Rule{$\wedge$} nodes into $n-1$ binary \Rule{$\wedge$} nodes,
  and $n$-ary \Rule{$\vee$} nodes into $n-1$ binary \Rule{$\vee$} nodes, as detailled above.
\end{proof}

Now, we introduce a new order $\polyleq$ on types. Intuitively, it adds to
the subtyping order $\leq$ the possibility to instantiate polymorphic type variables.
We then use it in the statement of a monotonicity lemma that will be used extensively
in the next sections.

\begin{definition}[Polymorphic subtyping order]\label{def:polyleq}
  We define the order relation $\polyleq$ over types as follows:
  \[ \forall t_1,t_2.\ t_1\polyleq t_2 \Leftrightarrow \exists\Psubst.\ t_1\Psubst\leq t_2 \]
\end{definition}

Note that, while this order will be extensively used in the proofs,
it will not be used in algorithms as we have no way to compute it.
Indeed, while deciding this order might seem equivalent to
solving a tallying problem (defined in Section~\ref{sec:tallying}),
it is actually not the case as deciding this order requires to find a set of substitutions $\Psubst$,
and not a single substitution $\psubst$.

\begin{definition}
  For any order relation $\leq$ over types,
  we define the order relation $\leq$ over environments as follows:
  \[ \forall \Gamma_1,\Gamma_2.\ \Gamma_1\leq \Gamma_2 \Leftrightarrow \forall \xxx\in\dom{\Gamma_2}.\ \xxx\in\dom{\Gamma_1}
  \text{ and } \Gamma_1(x)\leq \Gamma_2(x) \]
\end{definition}

\newcommand{\InstLeq}[0]{Inst$\wedge$$\leq$}
We introduce for convenience a new notation that takes the form of a new rule \Rule{\InstLeq},
but is actually just a shortand for
a specific combination of \Rule{Inst}, \Rule{$\land$}, and \Rule{$\leq$} rules:
\small\begin{mathpar}
    \Infer[\InstLeq]{
        \Infer{A}{\Gamma\vdasha e: t'}{}\quad(\exists\Psubst.\ \tbwedge_{\psubst\in\Psubst}t'\psubst \leq t)
    }{
        \Gamma\vdasha e: t
    }{ }
    \quad\raisebox{7pt}{$\leftrightarrow$}\quad
    \Infer[$\leq$]{
        \Infer[$\wedge$]{
            \Infer[Inst]{
                \Infer{A}{\Gamma\vdasha e: t'}{}
            }{
                \Gamma\vdasha e: t'\psubst
            }{ \forall \psubst\in\Psubst }
        }
        { \Gamma\vdasha e: \tbwedge_{\psubst\in\Psubst}t'\psubst }
        { }
    }
    { \Gamma\vdasha e: t }
    { }
\end{mathpar}\normalsize

\begin{lemma}[Monotonicity]\label{monotonicity_lemma}
  For derivation $D$ of $\Gamma\vdasha e:t$ and environment $\Gamma'$
  such that $\Gamma'\polyleq\Gamma$, $D$ can be transformed into a derivation of $\Gamma'\vdasha e:t$
  just by adding \Rule{$\leq$}, \Rule{Inst} and \Rule{$\wedge$} nodes.
\end{lemma}
\begin{proof}
  Straightforward induction on the derivation $\Gamma\vdasha e:t$,
  where each \Rule{\Axv} and \Rule{\Axl} node is replaced by a \Rule{\InstLeq} pattern
  of that node.
\end{proof}

\subsubsection{Normalisation Lemmas}\label{norm_lemmas}

Derivations for the declarative type system of Figure~\ref{fig:alt_declarative} can still take
many different shapes. In this section, we define several normalisation lemmas, each restricting
the use of a non-structural rule. They are then combined into a normalisation theorem.

\newcommand{\Vee}[0]{$\vee$}
\newcommand{\an}[1]{\textcolor{red}{#1}}
\textbf{Normalisation of \Rule{\Vee} nodes}

\begin{lemma}[Introduction of an arbitrary \Rule{\Vee} node]\label{vee_generation_lemma}
    Let $\Gamma$ a type environment, $e$ and $e_\xx$ two expressions,
    and $\{\mt_i\}_{i\in I} $ a partition of $\Any$.
    Let $D$ be a derivation for the judgement $\Gamma\vdasha e\subs{\xx}{e_\xx}:t$ such that
    $D$ does not contain any \Rule{\Vee} node performing a substitution $\subs{\yy}{e_\yy}$
    with $e_\yy$ a strict sub-expression of $e_\xx$.
    If $e_\xx$ is typable under the context $\Gamma$, then there exists a type $s$ such that
    $D$ can be transformed into a derivation whose root is a \Rule{\Vee} node of the following form:
    \begin{mathpar}
        \Infer[\Vee]
        { 
            \Infer{\dots}{\Gamma\vdasha e_\xx:s}{}\\
            \Infer{\dots}{\Gamma, \xx:s\land \mt_i\vdasha e:t}{\forall i\in I}
        }
        { \Gamma\vdasha e\subs{\xx}{e_\xx}:t }
        { }
    \end{mathpar}
\end{lemma}

\begin{proof}
    Let $C$ a derivation for $\Gamma\vdasha e_\xx:\Any$.
    We collect in $D$ the set $\{C_k\}_{k\in K}$ of all the subderivations for the expression $e_\xx$.
    As substitutions are capture-avoiding, no variable in $\fv{e_\xx}$ could have been introduced
    in the environment by a \Rule{$\to$I} or \Rule{\Vee} node in $D$.
    Thus, we know that the derivations $\{C_k\}_{k\in K}$ are still valid under the initial environment $\Gamma$.

    Thus, we can build the following derivation:
    \small\begin{mathpar}
        \Infer[\Vee]
        { 
            \Infer[$\wedge$]{
                \Infer{C}{\Gamma\vdasha e_\xx:\Any}{}\quad
                \Infer{C_k}{\Gamma\vdasha e_\xx:t_k}{\forall k\in K}
            }{\Gamma\vdasha e_\xx:\tbwedge_{k\in K} t_k}{}\\
            \Infer{D'_i}{\Gamma, \xx: (\tbwedge_{k\in K} t_k)\land\mt_i\vdasha e:t}{\forall i\in I}
        }
        { \Gamma\vdasha e\subs{\xx}{e_\xx}:t }
        { }
    \end{mathpar}\normalsize with each $D'_i$ being a derivation easily derived from $D$ by
    substituting $e_\xx$ by $\xx$ when relevant, using a \Rule{\Axv} rule on $\xx$
    instead of a subderivation for $e_\xx$ when necessary, and by using monotonicity (Lemma~\ref{monotonicity_lemma}).
    The hypothesis on the derivation $D$ ensures that it does not contain any conflicting \Rule{\Vee} node
    that would become inapplicable due to the fact that $e_\xx$ has been substituted by $\xx$.
\end{proof}

\begin{lemma}[Elimination of aliasing]\label{remove_aliasing_lemma}
    Let $D$ be a ground derivation, and $N$ be a \Rule{\Vee} node in $D$ applying a substitution
    $\subs{\yy}{\xx}$. Then, $N$ can be removed from $D$, without adding any new \Rule{\Vee} node nor structural node in $D$.
\end{lemma}

\begin{proof}
    The following transformation can be performed to the subderivation introducing $\xx$
    (as $D$ is a ground derivation, there must be a \Rule{\Vee} node that introduces $\xx$):

    \small\begin{mathpar}
        \Infer[\Vee]
        { 
            \Infer{ A }{ \Gamma\vdasha e_\xx : s }{ }
            \Infer{ B }{ \Gamma, \xx:s\land \mt_k\vdasha e: t }{ }
            \Infer{ E_i }{ \Gamma, \xx:s\land \mt_i\vdasha e:t }
            { \forall i\in I\setminus\{k\} 
            }
        }
        { \Gamma\vdasha e\subs{\xx}{e_\xx}:t }
        { }
        \\
        \text{where }B\text{ is a derivation that contains this subderivation $S$}\\
        \text{(whose root is the node $N$ to eliminate):}\\
        \Infer[\Vee]{
          \Infer{ C }{ \Gamma'\vdasha \xx : s' }{ }
          \Infer{ D_j }{ \Gamma', \yy:s'\land \mt_j'\vdasha e' : t' }
          { \forall j\in J
          }
        }
        {
            \Gamma'\vdasha e'\subs{\yy}{\xx}:t'
        }
        { }
        \\
        \text{with }\Gamma' = (\Gamma, \xx:s\land \mt_k) \disjcup \Gamma'' \text{ for some }\Gamma''
        \\\raisebox{1pt}{\hspace*{3cm}$\downarrow$ (transformed into)}\\
        \Infer[\Vee]
        { 
            \Infer{ A }{ \Gamma\vdasha e_\xx : s }{ }
            \Infer{ B_j' }{ \Gamma, \xx:s\land \mt_k\land \mt_j'\vdasha e: t }{ \forall j\in J }
            \Infer{ E_i }{ \Gamma, \xx:s\land \mt_i\vdasha e:t }
            { \forall i\in I\setminus\{k\}
            }
        }
        { \Gamma\vdasha e\subs{\xx}{e_\xx}:t }
        { }
        \\
        \text{where }B_j'\text{ is constructed from }B\text{ by monotonicity (Lemma~\ref{monotonicity_lemma}),}\\
        \text{and by replacing the subderivation $S$ by this one:}\\
        \Infer{ D_j' }
        {
            (\Gamma, \xx:s\land \mt_k\land \mt_j') \disjcup \Gamma''\vdasha e'\subs{\yy}{\xx}:t'
        }
        { }\\
        \text{where $D_j'$ is constructed from $D_j\subs{\yy}{\xx}$ by monotonicity (Lemma~\ref{monotonicity_lemma}) $(\star)$}
    \end{mathpar}
    \normalsize $(\star)$ Note that we have $s\land \mt_k\polyleq s'$
    as the only structural rule that the derivation $C$ can use is \Rule{\Axv} on $\xx$,
    and thus $s\land \mt_k\land \mt_j'\leq s'\land\mt_j'$.
\end{proof}

An interesting thing to note is that the proof above would not work in the presence of a generalization rule such as
the one used at top-level:
\begin{mathpar}
  \Infer[Gen]{ \Gamma\vdasha e : t }{ \Gamma\vdasha e : t\subst }{ \subst\disjoint\Gamma }
\end{mathpar}

Indeed, the guard condition $\subst\disjoint\Gamma$ may prevent monotonicity.
More precisely, in the transformation above, if the partition $\{\mt_i\}_{i\in I}$
introduces a new monomorphic type variable, its introduction earlier in the environment
might prevent a potential application of a \Rule{Gen} rule in the derivation $B$.
This impossibility to eliminate aliasing would be an issue for the normalisation lemma (Lemma~\ref{vee_norm_lemma})
detailled below, as it states, in particular,
that the union elimination rule only needs to be applied once for a given sub-expression.
This is the reason why, in our declarative type system, generalisation only occurs at top-level.

\begin{definition}[Acceptable \Rule{\Vee} node]\label{def:acceptable}
  In any derivation, a \Rule{\Vee} node $N$ doing the substitution $e\subs{\xx}{e'}$
  is said acceptable if it satisfies the following constraints:
  \begin{itemize}
      \item $e$ contains $\xx$ (no useless definition), and
      \item $e$ does not contain $e'$ (maximal sharing), and
      \item $e'$ is not a binding variable (no aliasing)
  \end{itemize}
\end{definition}

\begin{definition}[Binding context]\label{def:benv}
  A binding context $\benv$ is an ordered list of mappings 
  from binding variables to atoms. Each mapping is written as a pair $\be \xx e$.
  We note these lists extensionally by separating elements by a semicolon, that is,
  $\be{\xx_1}{a_1};\dots;\be{\xx_n}{a_n}$ and use $\bempty$ to denote the empty list.

  We note $\bsubs e \benv$ the expression $e\subs{\xx_n}{e_n}\dots\subs{\xx_1}{e_1}$
  where $\be{\xx_1}{e_1};\dots;\be{\xx_n}{e_n}$ are the pairs in $\benv$.
\end{definition}

Given a derivation $D$ and a node $N$ of $D$,
we call \textit{binding context of $N$ in $D$} the binding context $\be{\xx_1}{e_1};\dots;\be{\xx_n}{e_n}$,
where $\subs{\xx_1}{e_1};\dots;\subs{\xx_n}{e_n}$ are the substitutions made by all the \Rule{\Vee} nodes
crossed by one of their body premises when going from the root of $D$ to $N$.

In the following definition, we use an order over expressions:
\begin{definition}[Expression order]\label{def:leqexpr}
  A (possibly partial) order $\leqexpr$ is called expression order if it is an order over expressions
  modulo $\alpha$-renaming, and if it is an extension of the sub-expression order:
  if $e_1$ is a sub-expression of $e_2$ modulo $\alpha$-renaming, then we should have $e_1\leqexpr e_2$.
  We write $e_1 \ltexpr e_2$ when $e_1 \leqexpr e_2$ and $e_2 \not\leqexpr e_1$.
\end{definition}

\begin{definition}[Well-positionned \Rule{\Vee} node]\label{def:well-positionned}
    A \Rule{\Vee} node $N$ of a derivation $D$, of binding context $\benv$ in $D$
    and doing the substitution $e\subs{\xx}{e_1}$,
    is said well-positionned in $D$ relatively to the expression order $\leqexpr$
    (or more succintly, well-positionned in $(D,\leqexpr)$)
    if there is no node $N'$ on the path from the root to $N$ such that:
    \begin{itemize}
        \item Every variable in $\fv{\bsubs{e_1}{\benv}}$ is in the type environment of $N'$, and
        \item $N'$ is not a \Rule{\Vee} rule, or $N'$ is a \Rule{\Vee} rule
        of binding context $\benv'$ and performing a substitution
        $\subs{\yy}{e_2}$ such that $\bsubs{e_1}\benv\ltexpr\bsubs{e_2}{\benv'}$.
    \end{itemize}
  \end{definition}

  \begin{definition}[\Rule{\Vee}-canonical derivation]
    For a given expression order $\leqexpr$,
    a derivation $D$ is \Rule{\Vee}-canonical for the order $\leqexpr$
    if every \Rule{\Vee} node it contains is acceptable and well-positionned in $(D,\leqexpr)$.
  \end{definition}

\begin{definition}[Form derivations, atomic derivations]
  A derivation $D$ is a form derivation if every segment of branch containing the root
  and stopping at the first \Rule{\Vee} node (if any, or stopping at a leaf otherwise):
  \begin{itemize}
    \item does not contain any structural node, and
    \item if it ends with a \Rule{\Vee} node $N$, the definition premise of $N$ is an atomic derivation,
    and the body premises of $N$ are form derivations.
  \end{itemize}

  A derivation $D$ is an atomic derivation if every segment of branch containing the root
  and stopping at the first \Rule{$\to$I} node (if any, or stopping at a leaf otherwise):
  \begin{itemize}
    \item does not contain any \Rule{\Vee} node, and
    \item contains exactly one structural node, and
    \item if it ends with a \Rule{$\to$I} node $N$, the premise of $N$ is a form derivation.
  \end{itemize}
\end{definition}

A derivation that is both a \Rule{\Vee}-canonical derivation and a form derivation
is called \Rule{\Vee}-canonical form derivation.

\begin{lemma}[Normalisation of \Rule{\Vee}]\label{vee_norm_lemma}
    Given an expression order $\leqexpr$, any ground derivation $D$ of $\Gamma\vdasha e:t$
    can be transformed into a \Rule{\Vee}-canonical form derivation of $\Gamma\vdasha e:t'$
    for the order $\leqexpr$ and with $t'\polyleq t$.
\end{lemma}

\begin{proof}
    First, we can remove any aliasing (i.e. \Rule{\Vee} nodes doing a substitution
    $\subs{\xx}{\yy}$) by applying Lemma~\ref{remove_aliasing_lemma} as needed.
    We can also trivially remove useless \Rule{\Vee} nodes
    (i.e. those doing a substitution $e\subs{\xx}{e_\xx}$ where $e$ does not contain $\xx$).

    Then, let's consider, in the whole derivation, all the nodes that satisfy one of those conditions:
    \begin{itemize}
        \item It is a structural node $N$ such that, when going towards the root,
        it crosses another structural node before crossing a \Rule{$\to$I} node.
        \item It is a structural node $N$ such that, when going towards the root,
        it crosses a \Rule{\Vee} node by one of its body premises
        before crossing a \Rule{\Vee} node by its definition premise.
        \item It is a \Rule{\Vee} node that is not acceptable or not well-positionned in $(D,\leqexpr)$.
    \end{itemize}
    If there is no such node, then the properties of Lemma~\ref{vee_norm_lemma} are satisfied.
    Otherwise, we associate to each of these faulty nodes an expression and a binding context:
    \begin{itemize}
        \item For a structural node applied on an expression $e$ in a binding context $\benv$, we associate $(e,\benv)$,
        \item For a \Rule{\Vee} node doing the substitution $\subs{\xx}{e_\xx}$ in a binding context $\benv$,
        we associate $(e_\xx, \benv)$.
    \end{itemize}
    Now, we select among those nodes the one whose associated pair is minimal with respects to the following order:
    $(e_1,\benv_1)$ is smaller than $(e_2,\benv_2)$ if and only if $\bsubs{e_1}{\benv_1}\leqexpr\bsubs{e_2}{\benv_2}$.
    Let's call this node $N$, and its associated expression and binding context $e$ and $\benv$ respectively.
    
    Now, let's locate, in the segment from the root to $N$, the farthest location $L$ from the root such that
    $N$ would be well-positionned in $(D,\leqexpr)$ at this location.

    Let's note $\Gamma_L\vdasha e_L:t_L$ the judgement at this location.
    We consider the subderivation $D_L$ that derives this judgement in $D$.
    Note that, in $D_L$, there is no \Rule{\Vee} node that makes a substitution $\subs{\yy}{e_\yy}$ with $e_\yy$ a strict sub-expression of $e$,
    because $\bsubs{e_\yy}{\benv_\yy}$ (with $\benv_\yy$ the binding context of the corresponding \Rule{\Vee} node)
    would be smaller than $\bsubs{e}{\benv}$ for $\leqexpr$, thus making it not well-positionned and contradicting the minimality of $(e,\benv)$.
    Also note that $D_L$ contains the node $N$ (otherwise, the \Rule{\Vee} node $N$ would not be well-positionned at location $L$).

    We apply Lemma~\ref{vee_generation_lemma} on the root of $D_L$ so that it performs the substitution
    $e_L'\subs{\zz}{e}$ using the decomposition $\{\mt_i\}_{i\in I}= \{\Any\}$,
    with $\zz$ fresh and $e_L'$ an expression that does not contain $e$ and such that
    $e_L'\subs{\zz}{e}\equiv e_L\subs{\zz}{e}$ (basically, $e_L'$ is $e_L$ where occurrences of $e$ have been replaced by $\zz$).
    This lemma can be applied as:
    \begin{itemize}
        \item There cannot be in our subderivation any \Rule{\Vee} node substituting a strict sub-expression of $e$,
        \item We know that $\Gamma_L\vdasha e:\Any$ holds: we can derive it from the definition premise of $N$ if $N$ is a \Rule{\Vee} node,
        or from $N$ itself if $N$ is a structural node.
    \end{itemize}
    This gives us a new derivation $D_L'$.

    If $N$ is a \Rule{\Vee} node, $N$ can be removed in $D_L'$ by using Lemma~\ref{remove_aliasing_lemma},
    as well as other aliasing that would have been introduced in other branches. Not that, if $N$ is a structural node,
    it has already been eliminated by the application of Lemma~\ref{vee_generation_lemma}.
    Finally, we replace the subderivation $D_L$ by $D_L'$ in $D$.

    We conclude by repeating this whole process until all the nodes satisfy the conditions.
    We are guaranteed that it terminates by the fact that $\bsubs{e}{\benv}$ strictly increases
    for $\leqexpr$ at each iteration (with $(e,\benv)$ the pair associated to the choosen node $N$).
\end{proof}

\textbf{Normalisation of \Rule{Inst} nodes}

\begin{lemma}\label{vartype_alpharenaming_lemma}
  A derivation $\Gamma\vdasha e:t$ can be transformed into a derivation $\Gamma\vdasha e:t\renaming$
  (for any renaming $\renaming$) without changing the structure of the derivation.
\end{lemma}

\begin{proof}
  Any polymorphic type variable in $t$ must be introduced either by a \Rule{\Axl}, \Rule{\Axv}, or \Rule{Inst} rule.
  Thus, we can derive $\Gamma\vdasha e:t\renaming$ by induction on $\Gamma\vdasha e:t$,
  where:
  \begin{itemize}
      \item Every renaming $\renaming'$ of a \Rule{\Axl} or \Rule{\Axv} node is replaced by the renaming $\renaming\circ\renaming'$, and
      \item Every substitution $\psubst$ of a \Rule{Inst} node is replaced by the substitution
      $\renaming\circ\psubst\circ\renaming^{-1}$.
  \end{itemize}
\end{proof}

\begin{proposition}\label{inter_polyleq_prop}
  If $\forall i\in I.\ t_i'\polyleq t_i$,
  then $\tbwedge_{i\in I}t_i' \polyleq \tbwedge_{i\in I}t_i$.
\end{proposition}
\begin{proof}
  For each $i\in I$, let $\Psubst_i$ a set of substitutions such that $t_i'\Psubst_i\leq t_i$.
  We consider the set of substitutions $\Psubst = \textstyle\bigcup_{i\in I}\Psubst_i$,
  and we show that $(\tbwedge_{i\in I}t_i')\Psubst \leq \tbwedge_{i\in I}t_i$.

  We have:
  \begin{align*}\tbwedge_{\psubst\in\Psubst}(\tbwedge_{i\in I}t_i')\psubst\
    &\simeq \tbwedge_{i\in I}\tbwedge_{\psubst\in\Psubst_i}(\tbwedge_{j\in I}t_j')\psubst\\
    &\leq \tbwedge_{i\in I}\tbwedge_{\psubst\in\Psubst_i}t_i'\psubst\\
    &\leq \tbwedge_{i\in I}t_i
  \end{align*}
\end{proof}

\begin{proposition}\label{vee_polyleq_prop}
  If $\forall i\in I.\ t_i'\polyleq t_i$ and all $\{t_i\}_{i\in I}$ have disjoint polymorphic type variables,
  then $\tbvee_{i\in I}t_i' \polyleq \tbvee_{i\in I}t_i$.
\end{proposition}
\begin{proof}
  For each $i\in I$, let $\Psubst_i$ a set of substitutions such that $t_i'\Psubst_i\leq t_i$.
  We consider the set of substitutions
  $\Psubst=\{\psubst_1\disjcup\dots\disjcup\psubst_n\,\alt\,\psubst_1\in\Psubst_1,\dots,\psubst_n\in\Psubst_n\}$
  for $I=\{1,\dots,n\}$, where $\disjcup$ denotes the composition of disjoint substitutions (their disjointness
  is guaranteed by the fact that all $\{t_i'\}_{i\in I}$ have disjoint polymorphic type variables),
  and we show that $(\tbvee_{i\in I}t_i')\Psubst \leq \tbvee_{i\in I}t_i$.

  We have:
  \begin{align*}
      &\tbwedge_{\psubst\in\Psubst} (\tbvee_{i\in 1\isep n}t_i')\psubst\\
      \simeq& \tbwedge_{(\psubst_1,\dots,\psubst_n)\in \Psubst_1\times\dots\times\Psubst_n}(\tbvee_{i\in 1\isep n}t_i')(\psubst_1\disjcup\dots\disjcup\psubst_n)\\
      \simeq& \tbwedge_{(\psubst_1,\dots,\psubst_n)\in \Psubst_1\times\dots\times\Psubst_n} \tbvee_{i\in 1\isep n}t_i'(\psubst_1\disjcup\dots\disjcup\psubst_n)\\
      \simeq& \tbwedge_{(\psubst_1,\dots,\psubst_n)\in \Psubst_1\times\dots\times\Psubst_n} \tbvee_{i\in 1\isep n}t_i'\psubst_i\\
      \simeq& \tbvee_{i\in 1\isep n}\tbwedge_{\psubst_i\in\Psubst_i}t_i'\psubst_i&\text{(distributivity of $\lor$ over $\land$)}\\
      \leq& \tbvee_{i\in 1\isep n} t_i
  \end{align*}
\end{proof}

\begin{definition}[\Rule{Inst}-canonical derivation]
  A derivation $D$ is \Rule{Inst}-canonical
  if every \Rule{Inst} node it contains is part of a \Rule{\InstLeq} pattern that is either:
  \begin{itemize}
    \item The first premise of a \Rule{$\Empty$}, \Rule{$\in_1$} or \Rule{$\in_2$} node, or
    \item The premise of a \Rule{$\times$E$_1$} or \Rule{$\times$E$_2$} node, or
    \item One of the premises of a \Rule{$\to$E} node
  \end{itemize}
\end{definition}

\begin{lemma}[Normalisation of \Rule{Inst}]\label{inst_norm_lemma}
  Given an expression order $\leqexpr$, any \Rule{\Vee}-canonical form derivation $D$
  of $\Gamma\vdasha e:t$ can be transformed into a \Rule{\Vee}\Rule{Inst}-canonical form derivation
  of $\Gamma\vdasha e:t'$ for the order $\leqexpr$ and with $t'\polyleq t$.
\end{lemma}

\begin{proof}
    We proceed by induction on $(n_\vee,n)$ (using the lexicographic order),
    where $n_\vee$ denotes the number of \Rule{\Vee} nodes in the derivation, and $n$ denotes the total
    number of nodes in the derivation.
    
    If the root is a \Rule{Inst} or \Rule{$\leq$}, we can remove the root (its premise will be the new root)
    and proceed inductively on the result.
    
    If the root is a \Rule{$\wedge$}, we proceed inductively on all its premises and update the intersection type $t$
    derived by the root into a new intersection type $t'$ (according to the new premises).
    We know that the new derived type $t'$ satisfies $t'\polyleq t$ according to Proposition~\ref{inter_polyleq_prop}.

    If the root is a \Rule{\Vee} of the following form:
    \begin{mathpar}
      \Infer[\Vee]
      {
        \Infer{A}{\Gamma\vdasha e':s}{}
        \Infer{B_i}{\Gamma,\xx:s\land\mt_i\vdasha e:t}{\forall i\in I}
      }
      {\Gamma \vdasha e\subs{\xx}{e'}: t}
      { }
  \end{mathpar}
    \begin{enumerate}
        \item We first proceed inductively on $A$, which gives a derivation $A'$.
        We consider following derivation, where $s'\polyleq s$ (and thus $s'\wedge \mt_i\polyleq s\wedge \mt_i$):
        \begin{mathpar}
            \Infer[\Vee]{
                \Infer{A'}{\Gamma \vdasha e' : s'}{ }\quad
                \Infer
                { B_i' }
                {\Gamma, \xx:s'\wedge \mt_i\vdasha e: t}
                { \forall i\in I }
            }
            {\Gamma \vdasha e\subs{\xx}{e'}: t}{ }
        \end{mathpar} with $B_i'$ a derivation easily derived from $B_i$ by monotonicity (Lemma~\ref{monotonicity_lemma}).
        Note that the application of the monotonicity lemma might insert unwanted \Rule{\InstLeq} patterns
        after axioms, but they will be eliminated with the next step.

        \item The next step is to proceed inductively on the $\{B_i'\}_{i\in I}$ premises,
        yielding some derivations $\{B_i''\}_{i\in I}$ that derive some types $\{t_i\}_{i\in I}$ (with $\forall i\in I.\ t_i\polyleq t$).
        We can suppose that all the $\{t_i\}_{i\in I}$ have disjoint polymorphic type variables:
        if it is not the case, it can be ensured by applying Lemma~\ref{vartype_alpharenaming_lemma} to these premises.
        Then, we consider the following derivation:
    \small\begin{mathpar}
        \Infer[\Vee]
        {
        \Infer{A'}{\Gamma \vdasha e' : s'}{ }\quad
        \Infer[$\leq$]{\Infer{B_i''}{\Gamma, \xx:s'\land \mt_i \vdasha e : t_i}{ }
        }
        { \Gamma, \xx:s'\land \mt_i \vdasha e : \tbvee_{i\in I}t_i }{ \forall i\in I}
        }
        {\Gamma \vdasha e\subs{\xx}{e'}: \tbvee_{i\in I}t_i}{ }
    \end{mathpar}\normalsize
    The result $\tbvee_{i\in I}t_i$ satisfies $\tbvee_{i\in I}t_i\polyleq t$ according to Proposition~\ref{vee_polyleq_prop}.    
    
    \item The new \Rule{$\leq$} nodes that appear as premise of the \Rule{\Vee} root could
    break the properties of Lemma~\ref{vee_norm_lemma} if the corresponding $B_i''$ ends with a \Rule{\Vee} node.
    In this case, we move up the faulty \Rule{$\leq$} nodes as needed using this transformation:
    \small\begin{mathpar}
        \Infer[$\leq$]{ 
            \Infer[\Vee]
            { 
                \Infer{ A }{ \Gamma\vdasha e' : s }{ }\quad
                \Infer{ B_i }{ \Gamma, \xx:s\land \mt_i\vdasha e:t' }{ \forall i\in I }
            }
            { \Gamma\vdasha e\subs{\xx}{e'}:t' }
            { }
        }{
            \Gamma\vdasha e\subs{\xx}{e'}:t
        } { }
        \\\raisebox{1pt}{$\downarrow$}\\
        \Infer[\Vee]
        { 
            \Infer{ A }{ \Gamma\vdasha e' : s }{ }\quad
            \Infer[$\leq$]{
                \Infer{B_i}{ \Gamma, \xx:s\land \mt_i\vdasha e:t'}{ }\quad
            }{ \Gamma, \xx:s\land \mt_i\vdasha e:t }{ \forall i\in I }
        }
        { \Gamma\vdasha e\subs{\xx}{e'}:t }
        { }
    \end{mathpar}\normalsize
  \end{enumerate}

  The other cases are straightforward.
\end{proof}

\textbf{Normalisation of \Rule{$\leq$} nodes}

\begin{definition}[\Rule{$\leq$}-canonical derivation]
  A derivation $D$ is \Rule{$\leq$}-canonical
  if every \Rule{$\leq$} node it contains is either:
  \begin{itemize}
      \item The first premise of a \Rule{$\in_1$} or \Rule{$\in_2$} node, or
      \item One of the body premises of a \Rule{\Vee} node, or
      \item The premise of a \Rule{$\times$E$_1$} or \Rule{$\times$E$_2$} node, or
      \item The first premise of a \Rule{$\to$E} node
  \end{itemize}
\end{definition}

\begin{lemma}[Normalisation of \Rule{$\leq$}]\label{leq_norm_lemma}
  Given an expression order $\leqexpr$, any \Rule{\Vee}\Rule{Inst}-canonical form derivation $D$
  of $\Gamma\vdasha e:t$
  can be transformed into a \Rule{\Vee}\Rule{Inst}\Rule{$\leq$}-canonical form derivation
  of $\Gamma\vdasha e:t'$ for the order $\leqexpr$ and with $t'\polyleq t$.
\end{lemma}

\begin{proof}
    We proceed by induction on $(n_\vee,n)$ (using the lexicographic order),
    with $n_\vee$ the number of \Rule{\Vee} nodes in the derivation, and $n$ the total
    number of nodes in the derivation.

    If the root is a \Rule{$\leq$} or \Rule{Inst}, we can remove the root (its premise will be the new root)
    and proceed by induction on its premise.

    If the root is a \Rule{$\wedge$}, we proceed inductively on all its premises and update the intersection type $t$
    derived by the root into a new intersection type $t'$ (according to the new premises).
    We trivially have $t'\polyleq t$.

    If the root is a \Rule{$\to$I} pattern, we proceed inductively on its premise and update the arrow type $t$
    derived by the root into a new arrow type $t'$ (according to the new premise).
    We trivially have $t'\polyleq t$.

    If the root is a \Rule{\Vee} of the following form:
    \begin{mathpar}
      \Infer[\Vee]{
          \Infer{A}{\Gamma \vdasha e' : s}{ }\quad
          \Infer
          { B_i }
          {\Gamma, \xx:s\wedge \mt_i\vdasha e: t}
          { \forall i\in I }
      }
      {\Gamma \vdasha e\subs{\xx}{e'}: t}{ }
    \end{mathpar}
    \begin{enumerate}
        \item We first proceed inductively on $A$, yielding a derivation $A'$.
        We then consider the following derivation, with $s'\polyleq s$:
        \begin{mathpar}
            \Infer[\Vee]{
                \Infer{A'}{\Gamma \vdasha e' : s'}{ }\quad
                \Infer
                { B_i' }
                {\Gamma, \xx:s'\wedge \mt_i\vdasha e: t}
                { \forall i\in I }
            }
            {\Gamma \vdasha e\subs{\xx}{e'}: t}{ }
        \end{mathpar} where each $B_i'$ is derived from $B_i$ by applying Lemma~\ref{monotonicity_lemma} (monotonicity),
        and then Lemma~\ref{inst_norm_lemma} in order to normalize \Rule{Inst} nodes that might have been introduced
        by the monotonicity lemma.
        Note that this might add unwanted \Rule{$\leq$} nodes, but they will be eliminated with the next step.        
        
        \item We proceed inductively on the $\{B_i'\}_{i\in I}$ premises, yielding some derivations
        $\{B_i''\}_{i\in I}$ that derive some types $\{t_i\}_{i\in I}$ (with $\forall i\in I.\ t_i\polyleq t$).
        Then, we consider the following derivation:
        \small\begin{mathpar}
          \Infer[\Vee]
          {
          \Infer{A'}{\Gamma \vdasha e' : s'}{ }\quad
          \Infer[$\leq$]{\Infer{B_i''}{\Gamma, \xx:s'\land \mt_i \vdasha e : t_i}{ }
          }
          { \Gamma, \xx:s'\land \mt_i \vdasha e : \tbvee_{i\in I}t_i }{ \forall i\in I}
          }
          {\Gamma \vdasha e\subs{\xx}{e'}: \tbvee_{i\in I}t_i}{ }
      \end{mathpar}\normalsize
      The result $\tbvee_{i\in I}t_i$ trivially satisfies $\tbvee_{i\in I}t_i\polyleq t$.    
      
      \item The new \Rule{$\leq$} nodes that appear as premise of the \Rule{\Vee} root could
      break the properties of Lemma~\ref{vee_norm_lemma} if the corresponding $B_i''$ ends with a \Rule{\Vee} node.
      In this case, we move up the faulty \Rule{$\leq$} nodes as needed using this transformation:
      \small\begin{mathpar}
          \Infer[$\leq$]{ 
              \Infer[\Vee]
              { 
                  \Infer{ A }{ \Gamma\vdasha e' : s }{ }\quad
                  \Infer{ B_i }{ \Gamma, \xx:s\land \mt_i\vdasha e:t' }{ \forall i\in I }
              }
              { \Gamma\vdasha e\subs{\xx}{e'}:t' }
              { }
          }{
              \Gamma\vdasha e\subs{\xx}{e'}:t
          } { }
          \\\raisebox{1pt}{$\downarrow$}\\
          \Infer[\Vee]
          { 
              \Infer{ A }{ \Gamma\vdasha e' : s }{ }\quad
              \Infer[$\leq$]{
                  \Infer{B_i}{ \Gamma, \xx:s\land \mt_i\vdasha e:t'}{ }\quad
              }{ \Gamma, \xx:s\land \mt_i\vdasha e:t }{ \forall i\in I }
          }
          { \Gamma\vdasha e\subs{\xx}{e'}:t }
          { }
      \end{mathpar}\normalsize  
    \end{enumerate}

    The other cases are straightforward.
\end{proof}

\begin{definition}[Canonical derivation]
  Given an expression order $\leqexpr$,
  a canonical derivation for the expression order $\leqexpr$
  is a \Rule{\Vee}\Rule{Inst}\Rule{$\leq$}-canonical derivation.
  We say it is a canonical form derivation if it is also a form derivation,
  and a canonical atomic derivation if it is an atomic derivation.
\end{definition}
When qualifying a derivation $D$ of canonical, the order $\leqexpr$ may be omitted:
in this case, we consider that there exists an expression order $\leqexpr$ such that
$D$ is canonical for $\leqexpr$.

\begin{theorem}[Normalisation of derivations]\label{norm_thm}
  Given an expression order $\leqexpr$,
  any ground derivation $D$ of $\Gamma\vdasha e:t$
  can be transformed into a canonical form derivation of $\Gamma\vdasha e:t'$,
  for the order $\leqexpr$ and with $t'\polyleq t$.
\end{theorem}
\begin{proof}
  By successive application of Lemma~\ref{vee_norm_lemma}, Lemma~\ref{inst_norm_lemma},
  and Lemma~\ref{leq_norm_lemma}.
\end{proof}

\subsection{Type Safety}\label{sec:typesafety}

\subsubsection{The Parallel Semantics}\label{sec:parallelsem}

One technical difficuulty is that the subject reduction does not hold for the semantics
presented in Figure~\ref{fig:syntax}:
performing a reduction step on an expression $e$ might break the use of a \Rule{$\vee$} rule.
Indeed, if in the original typing derivation a rule \Rule{$\vee$} substitutes
multiple occurrences of the sub-expression $e$ by a variable $x$, reducing one occurrence of $e$
but not the others can make the application of this \Rule{$\vee$} rule impossible:
the correlation between the reduced $e$ and the other occurrences of $e$ is thus lost.

To circumvent this issue, we introduce a notion of parallel
reduction which forces to reduce all occurrences of a
sub-expression at the same time. 
The idea is to first define reduction rules that only apply at top-level,
and then define a context rule (rule \Rule{$\kappa$} below) that allows
reducing under an evaluation context, but that will apply this reduction
everywhere in the term. With this alternative semantics, the subject reduction becomes
true, allowing to prove type safety. The type safety for the initial semantics
(Figure~\ref{fig:syntax}) is then deduced from this result.

\begin{figure}[t]
  \input{proofs/fig-parallel-sem}
  \caption{Parallel Semantics\label{fig:parallel_sem}}
\end{figure}

The parallel semantics is formalized in Figure~\ref{fig:parallel_sem}.
A step of reduction happening at top-level is noted $\treduces$,
and a step of reduction of the parallel semantics under any evaluation context is noted $\preduces$.
Notice that the rule \Rule{$\kappa$} applies on an expression $e''$ a substitution from
an expression $e'$ to an expression $e$, noted $e''\esubs {e'}{e}$,
which is defined inductively on $e''$ as follows:
\begin{itemize}
  \item If  $e'\aequiv e''$,  then $e''\esubs{e'}e = e$.\vspace{1mm}
  \item If  $e'\not\aequiv e''$,  then $e''\esubs{e'}e$ is inductively defined as \vspace{-1.5mm} 
\end{itemize}
  \begin{align*}
    c\esubs{e'}e   & = c\\
    x\esubs{e'}{e} & = x\\
    \xx\esubs{e'}{e} & = \xx\\
    (e_1e_2)\esubs{e'}{e} & =  (e_1\esubs{e'}{e})(e_2\esubs{e'}{e})\\
    (\lambda x.e_\circ)\esubs{e'}{e} & = \lambda x.e_\circ &x\in\fv{e'}\\
    (\lambda x.e_\circ)\esubs{e'}{e} & = \lambda x.(e_\circ\esubs{e'}{e}) &x\not\in\fv e\cup\fv{e'}\\
    (\lambda x.e_\circ)\esubs{e'}{e} &= \lambda y.((e_\circ\subs x y)\esubs{e'}{e})
    &\hspace*{-1cm}x\not\in\fv e,x\in\fv{e'},y\text{ fresh}\\
    (\pi_i e_\circ)\esubs{e'}{e}   & = \pi_i (e_\circ \esubs{e'}{e})\\
    (e_1,e_2)\esubs{e'}{e} & =  (e_1\esubs{e'}{e},e_2\esubs{e'}{e})\\    
    (\tcase{e_1}{t}{e_2}{e_3})\esubs{e'}{e}  & = \tcase{e_1\esubs{e'}{e}}{t}{e_2\esubs{e'}{e}}{e_3\esubs{e'}{e}}
  \end{align*}
In particular, notice that expression substitutions are up to $\alpha$-renaming and perform only one pass.

Here is an example of reduction step using the parallel semantics:
\begin{mathpar}
  \Infer[$\kappa$]
  {
    (\lambda x. x)(\lambda x. x) \treduces \lambda x. x
  }
  {((\lambda x. x)(\lambda x. x),(\lambda x. x)(\lambda x. x)) \preduces (\lambda x. x, \lambda x. x)}
  {}
\end{mathpar}

\subsubsection{Subject Reduction}

\begin{proposition}\label{test_types_prop}
    If $\Gamma\vdasha v:\tau$, then $v\in\tau$ (see Figure~\ref{fig:syntax} for the definition of $\in$).
\end{proposition}

\begin{proof}
    Straightforward, by induction on the derivation of the judgement $\Gamma\vdasha v:\tau$.
    Note that the case of $\lambda$-abstractions is trivial as arrows in $\tau$ can only be $\Empty\to\Any$.
\end{proof}

\begin{lemma}[Substitution lemma]\label{substitution_lemma}
  If $\Gamma, x:s\vdasha e:t$ and $\Gamma\vdasha e':s$,
  then $\Gamma\vdasha e\subs{x}{e'}:t$.
\end{lemma}
\begin{proof}
  Straightforward induction on the derivation $\Gamma, x:s\vdasha e:t$, where
  each \Rule{\Axl} node is replaced by the derivation $\Gamma\vdasha e':s$.
\end{proof}

\begin{definition}[Atomic type]
  A type $t$ is said atomic if it cannot be decomposed into a non-trivial union.
  Formally, $t$ is atomic if and only if,
  for any set of types $\{t_i\}_{i\in I}$,
  $t \leq \tbvee_{i\in I} t_i \Rightarrow \exists i\in I.\ t\leq t_i$.
  Equivalently, $t$ is atomic if and only if, for any type $s$, either $t\leq s$ or $t\leq\neg s$.
\end{definition}

\begin{lemma}[Atomicity of value types]\label{atomic_value_lemma}
    If there exists a canonical atomic derivation $D$ of $\Gamma \vdasha v:s$ (with $v$ a value),
    then $s$ is atomic.
\end{lemma}
\begin{proof}
    As $v$ is a value, we know that the derivation does not contain any destructor nor axiom node,
    except possibly in the premise of a \Rule{$\to$I} node.
    Moreover, as $D$ is a canonical atomic derivation,
    it does not contain any \Rule{\Vee} nor \Rule{$\leq$} node neither.
    In particular, this implies that $s$ cannot contain any type variable nor union,
    except under an arrow.
    More precisely, we can easily prove that $s$ can be constructed with the following syntax:

    \[
        \begin{array}{lrcl}
        \textbf{Value Type} & \bar t & ::= & b\alt t\to t\alt \bar t\times \bar t \alt \bar t \wedge \bar t
        \end{array}
    \]
    
    Any type $t_1\to t_2$ is atomic, so is any base type $b$.
    Atomicity is preserved by intersection and product,
    thus we can deduce that a type constructed with the syntax above is atomic.
\end{proof}

\begin{definition}[Unavoidable \Rule{\Vee} node]
  In any derivation, a \Rule{\Vee} node $N$ doing the substitution $e\subs{\xx}{e'}$
  is said unavoidable if $N$ is acceptable (Definition~\ref{def:acceptable}) and if $e'$ is not a value.
  A \Rule{\Vee} node that is not unavoidable is said avoidable.
\end{definition}

\begin{definition}[Minimal derivation]
  A derivation is said minimal if every \Rule{\Vee} node it contains is unavoidable.
\end{definition}

\begin{lemma}[Elimination of value substitutions]\label{value_elim_lemma}
  Any canonical form derivation or canonical atomic derivation $D$ of $\Gamma\vdasha e:t$
  can be transformed into a minimal derivation of $\Gamma\vdasha e:t$.
\end{lemma}
\begin{proof}
  We proceed by structural induction on $D$.

  If the root is not a \Rule{\Vee} node, or if it is an unavoidable \Rule{\Vee} node,
  then we proceed inductively on all its premises.

  Otherwise, if the root is an avoidable \Rule{\Vee} node, we know that it is doing
  a substitution $e\subs{\xx}{v}$ with $v$ a value (as $D$ is canonical, the \Rule{\Vee} node
  must be acceptable). We can thus apply Lemma~\ref{atomic_value_lemma} on its definition premise
  $\Gamma\vdasha v:s$: we deduce that $s$ is atomic. Consequently, among the types $\{\mt_i\}_{i\in I}$
  composing the partition of $\Any$ used by this node, there is a type $\mt_i$, for some $i\in I$,
  such that $s\leq \mt_i$.
  
  We consider the corresponding body premise $\Gamma,\xx:s\vdasha e:t$,
  and we proceed by induction on it in order to get a minimal derivation $D'$ for $\Gamma,\xx:s\vdasha e:t$.
  Similarly, we proceed by induction on the definition premise $\Gamma\vdasha v:s$ in order to get
  a minimal derivation $A$ for $\Gamma\vdasha v:s$.

  We can then derive $\Gamma\vdasha e\subs{\xx}{v}:t$ from $D'$ by replacing \Rule{\Axv} nodes
  applying on $\xx$ by the derivation $A$ (in a similar way as done in Lemma~\ref{substitution_lemma},
  but for a binding variable). This new derivation is minimal.
\end{proof}

\newcommand{\srsubst}[0]{\esubs{e_\circ}{e_\circ'}}
\begin{lemma}[Subject reduction]\label{sr_lemma}
    If $D$ is a minimal derivation for $\Gamma\vdasha e:t$, and if $e_\circ\treduces e_\circ'$,
    then $\Gamma\vdasha e\srsubst:t$.
\end{lemma}

\begin{proof}
  We proceed by structural induction on $D$.

  If $e$ contains no occurrence of $e_\circ$ (modulo $\alpha$-renaming), this result is trivial.
  Thus, we will suppose in the following that $e$ contains at least one occurrence of $e_\circ$.

    We denote by $e'$ the expression $e\srsubst$, and
    consider the root of the derivation:
    \begin{description}
        \item[\Rule{Const}] Impossible case ($e$ cannot contain any reducible expression).
        \item[\Rule{\Axl}] Impossible case ($e$ cannot contain any reducible expression).
        \item[\Rule{\Axv}] Impossible case ($e$ cannot contain any reducible expression).
         
        \item[\Rule{$\leq$}] By induction on the premise $\Gamma\vdash e:t'$ (with $t'\leq t$),
        we get a derivation for $\Gamma\vdash e':t'$,
        thus we can derive $\Gamma\vdash e':t$ by using a \Rule{$\leq$} node.
        \item[\Rule{Inst}] Similar to the previous case.
        \item[\Rule{$\wedge$}] We proceed inductively on each premise, and intersect the resulting
        derivations with a \Rule{$\wedge$} node.

        \item[\Rule{$\to$I}] We have $e'\equiv \lambda x.\ (e_\lambda\srsubst)$ for some expression $e_\lambda$.
        We can derive $\Gamma,x:\mt\vdasha e_\lambda\srsubst:t'$ by induction
        on the premise $\Gamma,x:\mt\vdasha e_\lambda:t'$ (with $t\simeq \mt\to t'$) and conclude by using a \Rule{$\to$I} node.
        \item[\Rule{$\times$I}] We have $e'\equiv (e_1\srsubst, e_2\srsubst)$ for some $e_1$ and $e_2$. We proceed inductively
        on the premises, as in the previous case.
    
        \item[\Rule{$\to$E}] We have $e\equiv e_1 e_2$ for some expressions $e_1$ and $e_2$.
        If $e_\circ$ is a sub-expression of $e_1$ and/or $e_2$, we conclude by proceeding
        inductively on the premises, as in the previous case.
    
        Otherwise, $e_\circ\equiv e_1 e_2$ and thus the reduction $e_\circ\treduces e_\circ'$ uses the rule \ref{redapp}.
        Consequently, we know that $e_\circ \equiv e \equiv (\lambda x.\ e_\lambda) v$ for some expression $e_\lambda$ and value $v$,
        and $e_\circ'\equiv e'\equiv e_\lambda\subs{x}{v}$.
    
        We have the following premises:
        \begin{enumerate}
            \item $\Gamma\vdasha\lambda x.\ e_\lambda:t_1\to t_2$ (with $t_2\simeq t$)
            \item $\Gamma\vdasha v:t_1$
        \end{enumerate}
        As $D$ is minimal, we know that the premise (1) has no \Rule{\Vee} node,
        except possibly in the premise of a \Rule{$\to$I} node.
        Thus, we can extract from (1) a collection of derivations of the judgements
        $\Gamma, x:\mt_i\vdasha e_\lambda:s_i$ for $i\in I$,
        such that there exists some instantiations $\Psubst_i$ such that $\tbwedge_{i\in I} (\mt_i\to s_i)\Psubst_i \leq t_1\to t_2$.
        We have $\tbwedge_{i\in I} \mt_i\to s_i\Psubst_i \simeq \tbwedge_{i\in I} (\mt_i\to s_i)\Psubst_i \leq t_1\to t_2$.
        Moreover, using \Rule{\InstLeq} patterns, we can derive for each $i\in I$
        the judgement $\Gamma, x:\mt_i\vdasha e_\lambda:s_i\Psubst_i$.

        Now, let's consider a partition $\{\ms_j\}_{j\in J}$ of $\tbvee_{i\in I}\mt_i$ of minimal cardinality
        such that: $\forall j\in J.\ \forall i\in I.\ \ms_j\leq \mt_i$ or $\ms_j\land \mt_i\simeq\Empty$.
        We know that $J$ is not empty: the case $t_1\leq\tbvee_{i\in I}\mt_i\simeq\Empty$ is impossible since
        the value $v$ would have type $\Empty$, contradicting Proposition~\ref{test_types_prop}.
        For every $j\in J$, we define $I_j = \{i\in I\,\alt\, \ms_i\land \mt_i\not\simeq\Empty\}$
        ($I_j$ cannot be empty as the partition $\{\ms_j\}_{j\in J}$ has minimal cardinality).
        Note that for any $j\in J$ and $i\in I_j$, we have $\ms_j\leq \mt_i$.

        Using the monotonicity lemma (Lemma~\ref{monotonicity_lemma}),
        we can derive for every $j\in J$ and $i\in I_j$ the judgement
        $\Gamma, x:\ms_j\vdasha e_\lambda:s_i\Psubst_i$.
        Thus, using a \Rule{$\wedge$} node, for every $j\in J$
        we can derive $\Gamma, x:\ms_j\vdasha e_\lambda:\tbwedge_{i\in I_j}s_i\Psubst_i$.
        Moreover, as $\tbwedge_{i\in I} \mt_i\to s_i\Psubst_i \leq t_1\to t_2$, we have, for every $j\in J$,
        $\tbwedge_{i\in I_j} s_i\Psubst_i \leq t_2$.
        Consequently, for every $j\in J$, we can derive the judgement
        $\Gamma, x:\ms_j\vdasha e_\lambda:t_2$ using a \Rule{$\leq$} node.
        Combining these judgements with a \Rule{\Vee} node allows us to derive
        $\Gamma, x:\tbvee_{j\in J}\ms_j\vdasha e_\lambda:t_2$.

        As $\tbvee_{j\in J}\ms_j$ covers $t_1$, we can construct from (2)
        a derivation $\Gamma\vdasha v:\tbvee_{j\in J}\ms_j$ using a \Rule{$\leq$} node.
        Finally, applying the substitution lemma (Lemma~\ref{substitution_lemma})
        gives a derivation for $\Gamma\vdasha e_\lambda\subs{x}{v}:t_2$.
        
        \item[\Rule{$\times$E$_1$}] We have $e\equiv \pi_1 e_1$ for some expression $e_1$.
        If $e_\circ$ is a sub-expression of $e_1$, we conclude by proceeding inductively on the premise.
        
        Otherwise, $e_\circ\equiv \pi_1 e_1$ and thus the reduction $e_\circ\treduces e_\circ'$ uses the rule \ref{redproj1}.
        Consequently, we know that $e_\circ \equiv e \equiv \pi_1 (v_1, v_2)$ for some values $v_1$ and $v_2$,
        and $e_\circ'\equiv e'\equiv v_1$.
        
        Similarly to the previous case, as $D$ is minimal, we can extract
        from the premise $\Gamma\vdasha (v_1,v_2):t_1\times t_2$
        a collection of derivations of the judgements $\Gamma\vdasha v_1:s_i$ and $\Gamma\vdasha v_2:s_i'$
        for $i\in I$, such that there exists some instantiations $\Psubst_i$ such that
        $\tbwedge_{i\in I} (s_i\times s_i')\Psubst_i \leq t_1\times t_2$.
        We thus have $\tbwedge_{i\in I} s_i\Psubst_i \leq t_1$, as none of the $\{s_i'\}_{i\in I}$ can be $\Empty$
        (this would contradict Proposition~\ref{test_types_prop} as $v_2$ is a value).
    
        Therefore, we can conclude this case by using a \Rule{\InstLeq} pattern with the premises
        $\{\Gamma\vdasha v_1:s_i\}_{i\in I}$ in order to derive $\Gamma\vdasha v_1:t_1$.
    
        \item[\Rule{$\times$E$_2$}] Similar to the previous case.
        
        \item[\Rule{$\vee$}]
        We have $e\equiv e_1\subs{\xx}{e_2}$ for some expressions $e_1$ and $e_2$, and thus
        $e'\equiv (e_1\subs{\xx}{e_2})\esubs{e_\circ}{e_\circ'}$.

        We have the following premises:
        \begin{enumerate}
            \item Definition premise: $\Gamma\vdasha e_2:s$
            \item Body premises: $\forall i\in I.\ \Gamma,\xx:s\land\mt_i\vdasha e_1:t$
        \end{enumerate}
        As $D$ is minimal, we know that $e_2$ cannot be a value.
        Also, we know that $e_\circ$ does not contain $\xx$
        as a free variable (otherwise there would be no occurrence of $e_\circ$ in $e_1\subs{\xx}{e_2}$).
        Consequently, $e_\circ'$ does not contain $\xx$ neither, because a reduction step cannot introduce a new
        free variable.
    
        There are several cases:
        \begin{itemize}
            \item $e_\circ$ does not contain $e_2$ and $e_2$ does not contain $e_\circ$.
            In this case, we have:\\
            $e'\equiv(e_1\subs{\xx}{e_2})\esubs{e_\circ}{e_\circ'} \equiv (e_1\esubs{e_\circ}{e_\circ'})\subs{\xx}{e_2}$.
            Thus, we can easily conclude by keeping the definition premise of the \Rule{$\vee$} node
            and applying the induction hypothesis on the body premises.

            \item $e_2$ contains $e_\circ$. In this case, we pose $e_2'=e_2\esubs{e_\circ}{e_\circ'}$.\\
            We have $e'\equiv(e_1\subs{\xx}{e_2})\esubs{e_\circ}{e_\circ'} \equiv (e_1\esubs{e_\circ}{e_\circ'})\subs{\xx}{e_2'}$.
            We can derive $\Gamma\vdasha e_2':s$ by induction on the definition premise,
            and $\Gamma, \xx:s\land \mt_i\vdasha e_1\esubs{e_\circ}{e_\circ'}:t$ for all $i\in I$ by induction on the body premises.
            Thus, we can derive $\Gamma\vdasha (e_1\esubs{e_\circ}{e_\circ'})\subs{\xx}{e_2'} : t$ using a \Rule{$\vee$} node.

            \item $e_\circ$ contains $e_2$ as a strict sub-expression.
            In this case, we pose $e_\bullet=e_\circ\esubs{e_2}{\xx}$ and $e_\bullet'=e_\circ'\esubs{e_2}{\xx}$.
            We know that $e_1$ does not contain any occurrence of $e_2$ (because $D$ is minimal),
            and thus no occurrence of $e_\circ$ neither.
            Consequently, we have $e'\equiv(e_1\subs{\xx}{e_2})\esubs{e_\circ}{e_\circ'} \equiv (e_1\esubs{e_\bullet}{e_\bullet'})\subs{\xx}{e_2}$.
            
            As $e_2$ is not a value, it can only appear in $e_\circ$ inside of a $\lambda$-abstraction, and/or inside of a branch of a typecase:
            otherwise, $e_\circ$ would not be reducible.
            
            Thus, we can deduce that $e_\bullet = e_\circ\esubs{e_2}{\xx} \treduces e_\circ'\esubs{e_2}{\xx}=e_\bullet'$.
            Consequently, we can easily conclude by keeping the definition premise of the \Rule{$\vee$} node
            and applying the induction hypothesis on the body premises.
        \end{itemize}
        \item[\Rule{$\Empty$}] We have $e\equiv\tcase{e_1}{\tau}{e_2}{e_3}$ for some $e_1$, $e_2$ and $e_3$.
        As values cannot have the type $\Empty$ (Proposition~\ref{test_types_prop}), we know that $e_1$ is not a value.
        Thus, $e'\equiv\tcase{e_1\srsubst}{\tau}{e_2\srsubst}{e_3\srsubst}$. We can derive $\Gamma\vdasha e_1\srsubst:\Empty$
        by proceeding inductively on the premise, and then we can derive $\Gamma\vdasha e':\Empty$ by using \Rule{$\Empty$}.
        \item[\Rule{$\in_1$}] We have $e\equiv\tcase{e_1}{\tau}{e_2}{e_3}$ for some $e_1$, $e_2$ and $e_3$.
        There are three cases:
        \begin{description}
            \item[$e'\equiv\tcase{e_1\srsubst}{\tau}{e_2\srsubst}{e_3\srsubst}$] We can easily conclude by
            proceeding inductively on the premises.
            \item[$e'\equiv e_2$] The second premise, unchanged, allows us to conclude.
            \item[$e'\equiv e_3$] This case is impossible. Indeed, it implies that $e_1$ is a value, and
            as $\Gamma\vdasha e_1:\tau$ (first premise), we can deduce using Proposition~\ref{test_types_prop}
            that $e_1\in\tau$, which contradicts $e\treduces e_3$.
        \end{description}
        \item[\Rule{$\in_2$}] Similar to the previous case.  
    \end{description}    
\end{proof}

\begin{theorem}[Subject reduction]\label{subj_red}
    If $\Gamma\vdasha e:t$ with $e$ a ground expression and if $e\preduces e'$, then $\Gamma\vdasha e':t$.
\end{theorem}

\begin{proof}
  Using the normalisation theorem (Theorem~\ref{norm_thm}) and the Lemma~\ref{value_elim_lemma},
  we can build a minimal derivation for $\Gamma\vdasha e:t$.

  Moreover, the root of the reduction step $e\preduces e'$ is a \Rule{$\kappa$} node,
  with its premise being of the form $e_\circ\treduces e_\circ'$, and with
  $e'\equiv e\esubs{e_\circ}{e_\circ'}$.

  Thus, by using Lemma~\ref{sr_lemma}, we obtain $\Gamma\vdasha e':t$.
\end{proof}

\subsubsection{Progress}

\begin{lemma}[Progress]\label{progress_lemma}
  If $D$ is a minimal derivation for $\Gamma\vdasha e:t$,
  and if there is no evaluation context $E$ and
  variable $\xxx$ such that $e \equiv E[\xxx]$,
  then either $e$ is a value or $\exists e'.\ e\preduces e'$.
\end{lemma}

\begin{proof}
  We proceed by structural induction on $D$.

  We consider the root of the derivation:
  \begin{description}
      \item[\Rule{Const}] Trivial ($e$ is a value).
      \item[\Rule{\Axl}] Impossible case (contradict the hypotheses).
      \item[\Rule{\Axv}] Impossible case (contradict the hypotheses).
      
      \item[\Rule{$\leq$}] Trivial (by induction on the premise).
      \item[\Rule{Inst}] Trivial (by induction on the premise).
      \item[\Rule{$\wedge$}] Trivial (by induction on one of the premises).
      \item[\Rule{$\to$I}] Trivial ($e$ is a value).
      \item[\Rule{$\times$I}] We have $e\equiv(e_1, e_2)$ for some expressions $e_1$ and $e_2$.
      \begin{itemize}
          \item If $e_1$ is not a value, and as we have $\forall E, \xxx.\ e_1\not\equiv E[\xxx]$,
          we know by applying the induction hypothesis that $e_1$ can be reduced.
          Thus, $e$ can also be reduced under the evaluation context $([\,], e_2)$.
          \item If $e_1$ is a value,
          then we can apply the induction hypothesis on the second premise
          (as $e_1$ is a value, we know that $\forall E, \xxx.\ e_2\not\equiv E[\xxx]$).
          It gives that either $e_2$ is a value or it can be reduced. We can easily conclude in both cases:
          if $e_2$ is a value, then $e$ is also a value, otherwise, $e$ can be reduced under the evaluation context $(e_1, [\,])$.
      \end{itemize}

      \item[\Rule{$\to$E}] We have $e\equiv e_1 e_2$ for some expressions $e_1$ and $e_2$, with $\Gamma\vdasha e_1: s\to t$
      and $\Gamma\vdasha e_2:s$.
      \begin{itemize}
          \item If $e_1$ is not a value, and as we have $\forall E, \xxx.\ e_1\not\equiv E[\xxx]$,
          we know by applying the induction hypothesis that $e_1$ can be reduced.
          Thus, $e$ can also be reduced under the evaluation context $[\,] e_2$.
          \item If $e_1$ is a value, we can apply Proposition~\ref{test_types_prop} on it:
          as $\Gamma \vdasha e_1 : \Empty\to\Any$, it implies that $e_1\in \Empty\to\Any$ and thus
          $e_1\equiv \lambda x.\ e_\lambda$ for some $e_\lambda$.
          Moreover, we can apply the induction hypothesis on the second premise
          (as $e_1$ is a value, we know that $\forall E, \xxx.\ e_2\not\equiv E[\xxx]$).
          It gives that either $e_2$ is a value or it can be reduced. We can easily conclude in both cases:
          if $e_2$ is a value, then $e$ can be reduced using the rule \ref{redapp}, otherwise,
          $e$ can be reduced under the evaluation context $e_1 [\,]$.
      \end{itemize}

      \item[\Rule{$\times$E$_1$}] We have $e\equiv\pi_1 e_1$ for some $e_1$, with $\Gamma\vdasha e_1: t\times s$.
      By applying the induction hypothesis on the premise, we know that $e_1$ is either a value or it can be reduced.
      If $e_1$ can be reduced, then $e$ can also be reduced under the evaluation context $\pi_1 [\,]$.
      Otherwise, as $\Gamma\vdasha e_1: \Any\times\Any$, we can apply Proposition~\ref{test_types_prop} on it,
      yielding $e_1\in\Any\times\Any$. Thus, $e_1\equiv (v_1, v_2)$ for some values $v_1$ and $v_2$,
      and consequently $e$ can be reduced using the rule \ref{redproj1}.
      \item[\Rule{$\times$E$_2$}] Similar to the previous case.
      
      \item[\Rule{$\vee$}] We have $e\equiv e_1\subs{\xx}{e_2}$ for some $e_1$ and $e_2$, with $\Gamma\vdasha e_2:s$
      and $\forall i\in I.\ \Gamma, \xx:s\land \mt_i\vdasha e_1:t$. There are two cases:
      \begin{itemize}
          \item There exists an evaluation context $E$ such that $e_1\equiv E[\xx]$.
          In this case, we know that $\forall E', \yyy.\ e_2\not\equiv E'[\yyy]$,
          otherwise we would have $e\equiv E[E'[\yyy]]$. Thus, we can apply the induction
          hypothesis on the definition premise. It gives that either $e_2$ is a value or it can be reduced.
          As $D$ is minimal, $e_2$ cannot be a value and thus $e_2$ can be reduced.
          Consequently, $e$ can also be reduced under the evaluation context $E$.

          \item There is no evaluation context $E$ such that $e_1\equiv E[\xx]$.
          In this case, we apply the induction hypothesis on one of the body premises.
          It gives that either $e_1$ is a value or it can be reduced.
          We can easily conclude in both cases: if $e_1$ is a value, then $e$ is also a value,
          otherwise, $e$ can be reduced.
      \end{itemize} 

      \item[\Rule{$\Empty$}] We have $e\equiv\tcase{e_1}{\tau}{e_2}{e_3}$ for some $e_1$, $e_2$ and $e_3$,
      with $\Gamma\vdasha e_1: \Empty$.
      As values cannot have the type $\Empty$ (Proposition~\ref{test_types_prop}),
      we know that $e_1$ is not a value. Thus, by applying the induction hypothesis on the premise,
      we know that $e_1$ can be reduced. Consequently, $e$ can be reduced under the evaluation context $\tcase{[\,]}{\tau}{e_2}{e_3}$.
      \item[\Rule{$\in_1$}] We have $e\equiv\tcase{e_1}{\tau}{e_2}{e_3}$ for some $e_1$, $e_2$ and $e_3$,
      with $\Gamma\vdasha e_1: \tau$. We thus have $e_1\in\tau$ (Proposition~\ref{test_types_prop}).
      By applying the induction hypothesis on the first premise, we know that $e_1$ is either a value or it can be reduced.
      If $e_1$ is a value, then $e$ can be reduced using the rule \ref{redthen}.
      Otherwise, $e_1$ can be reduced and thus $e$ can also be reduced under the evaluation context $\tcase{[\,]}{\tau}{e_2}{e_3}$.
      \item[\Rule{$\in_2$}] Similar to the previous case.
  \end{description}
\end{proof}

\begin{theorem}[Progress]\label{progress}
  If $\emptyenv\vdasha e:t$ with $e$ a ground expression, then either $e$ is a value or $\exists e'.\ e\preduces e'$.
\end{theorem}

\begin{proof}
  Using the normalisation theorem (Theorem~\ref{norm_thm}) and the Lemma~\ref{value_elim_lemma},
  we can build a minimal derivation for $\emptyenv\vdasha e:t$.

  Moreover, we can deduce from $\emptyenv\vdasha e:t$ that there is no evaluation context $E$ and lambda variable $\xxx$
  such that $e \equiv E[\xxx]$. Thus, we can conclude using Lemma~\ref{progress_lemma}.
\end{proof}

\subsubsection{Type Safety for Programs}

Now that we expressed subject reduction and progress for the parallel semantics on expressions,
we extend it to programs, still using the parallel semantics.
Then, we will deduce from it a type safety theorem on programs for the parallel semantics.

\begin{lemma}[Weak monotonicity]\label{weak_monotonicity_lemma}
  For derivation $D$ of $\Gamma\vdasha e:t$ and environment $\Gamma'$
  such that $\Gamma'\leq\Gamma$, $D$ can be transformed into a derivation of $\Gamma'\vdasha e:t$
  just by adding \Rule{$\leq$} nodes.
\end{lemma}
\begin{proof}
  Straightforward induction on the derivation $\Gamma\vdasha e:t$,
  where each \Rule{\Axv} and \Rule{\Axl} node is replaced by a \Rule{$\leq$} node
  with the \Rule{\Axv} or \Rule{\Axl} node as premise.
\end{proof}

\begin{lemma}[Elimination of \Rule{Inst} nodes]\label{inst_elimination_lemma}
  For any canonical form derivation $D$ of $\Gamma\vdasha e:t$ and set of substitutions $\Psubst$,
  $D$ can be transformed into a derivation $D'$ of $\Gamma\Psubst'\vdasha e:t\Psubst$, for some $\Psubst'$,
  and such that $D'$ does not contain any \Rule{Inst} node.
\end{lemma}
\begin{proof}
  We procced by structural induction on the derivation $D$ (with $D$ being a canonical form derivation).

  If the root of $D$ is a \Rule{\Axv} node, we can directly conclude by choosing $\Psubst' = \Psubst$.

  Otherwise, if the root $D$ is a \Rule{\Vee} node, doing a substitution $e\subs{\xx}{e'}$,
  we first proceed inductively on its body premises (which all are canonical form derivations).
  It yields some derivations $\Gamma\Psubst_i',\xx:s\Psubst_i'\land\mt_i\vdasha e:t\Psubst$ for $i\in I$.
  We consider the set of substitutions $\Psubst' = \textstyle\bigcup_{i\in I} \Psubst_i'$.

  Now, let's try to derive $\Gamma\Psubst''\vdasha e':s'$, for some $\Psubst''$ and $s'$,
  with $s'\leq s\Psubst'$ ($\leq s\Psubst_i'$ for every $i\in I$).
  For that, we consider the definition premise of $D$, deriving $\Gamma\vdasha e':s$,
  and apply the following process on it:
  \begin{itemize}
    \item If its root is a \Rule{$\land$} node, we apply inductively this process to each premise
    $\Gamma\vdasha e':s_j$ (with $\tbvee_{j\in J} s_j\simeq s$) in order to derive some
    $\Gamma\Psubst_j''\vdasha e':s_j'$, with $s_j'\leq s\Psubst'$. We then consider
    $\Psubst''=\textstyle\bigcup_{j\in J} \Psubst_j''$, and derive for each $j\in J$
    $\Gamma\Psubst''\vdasha e':s_j'$ using Lemma~\ref{weak_monotonicity_lemma}.
    Finally, we intersect the those derivations with a \Rule{$\land$} node.
    \item If its root is a \Rule{$\to$E} node, we apply the following transformation:
    \small\begin{mathpar}
      \Infer[$\to$E]
      {
        \Infer[\InstLeq]{\Infer[\Axv]{ }{\Gamma\vdasha \xx_1:\Gamma(\xx_1)}{}}{\Gamma\vdasha \xx_1: t\to s}{ }\\
        \Infer[\InstLeq]{\Infer[\Axv]{ }{\Gamma\vdasha \xx_2:\Gamma(\xx_2)}{}}{\Gamma\vdasha \xx_2: t}{ }
      }
      { \Gamma\vdasha \xx_1 \xx_2:s }
      { }
      \\
      \text{with $t\to s\leq\Gamma(\xx_1)\Psubst_1''$ and $t\leq\Gamma(\xx_2)\Psubst_2''$}
      \\\raisebox{1pt}{$\downarrow$}\\
      \Infer[$\to$E]
      {
        \Infer[$\leq$]{\Infer[\Axv]{ }{\Gamma\Psubst''\vdasha \xx_1:\Gamma(\xx_1)\Psubst''}{}}{\Gamma\Psubst''\vdasha \xx_1: t\Psubst'\to s\Psubst'}{ }\\
        \Infer[$\leq$]{\Infer[\Axv]{ }{\Gamma\Psubst''\vdasha \xx_2:\Gamma(\xx_2)\Psubst''}{}}{\Gamma\Psubst''\vdasha \xx_2: t\Psubst'}{ }
      }
      { \Gamma\Psubst''\vdasha \xx_1 \xx_2:s\Psubst' }
      { }
      \\
      \text{with $\Psubst''=\{\psubst'\circ\psubst''\,\alt\,\psubst'\in\Psubst', \psubst''\in\Psubst_1''\cup\Psubst_2''\}$}
    \end{mathpar}\normalsize
    \item If its root is a \Rule{$\to$I} node, we apply the following transformation:
    \small\begin{mathpar}
      \Infer[$\to$I]
      {
        \Infer{ A }{ \Gamma, x:\mt \vdasha e_\circ : t }{ }
      }
      { \Gamma\vdasha \lambda x. e_\circ : \mt\to t }
      { }
      \quad\raisebox{7pt}{$\leftrightarrow$}\quad
      \Infer[$\to$I]
      {
        \Infer{ A' }{ \Gamma\Psubst'', x:\mt \vdasha e_\circ : (t\Psubst') }{ }
      }
      { \Gamma\Psubst''\vdasha \lambda x. e_\circ : \mt\to (t\Psubst') }
      { }
      \\
      \text{with $A'$ obtained by applying the induction hypothesis on $A$ and $\Psubst'$.}
    \end{mathpar}\normalsize

    \item If its root is a \Rule{\Axl} node, we apply the following transformation:
    \small\begin{mathpar}
      \Infer[\Axl]
      { }
      { \Gamma\vdasha x : \Gamma(x) }
      { }
      \quad\raisebox{7pt}{$\leftrightarrow$}\quad
      \Infer[\Axl]
      { }
      { \Gamma\Psubst'\vdasha x : \Gamma(x)\Psubst' }
      { }
    \end{mathpar}\normalsize

    \item The other cases are similar.
  \end{itemize}

  We thus have $\Psubst''$ and $s'$ such that
  $\Gamma\Psubst''\vdasha e':s'$ and $\forall i\in I.\ s' \leq s\Psubst_i'$.
  We transform, for each $i\in I$, the derivation $\Gamma\Psubst_i',\xx:s\Psubst_i'\land\mt_i\vdasha e:t\Psubst$
  into a derivation $\Gamma\Psubst'',\xx:s'\land\mt_i\vdasha e:t\Psubst$ using Lemma~\ref{weak_monotonicity_lemma}.
  By combining all those derivations into a \Rule{\Vee} node, we can derive
  $\Gamma\Psubst''\vdasha e\subs{\xx}{e'}: t\Psubst$.
\end{proof}

\begin{lemma}\label{instantiation_lemma}
  If $\Gamma\vdasha e:t$, then for any substitution $\msubst$,
  we can derive $\Gamma\msubst\vdasha e:t\msubst$.
\end{lemma}

\begin{proof}
  Straightforward induction on the derivation of $\Gamma\vdasha e:t$.
  The type substitution $\msubst$ can be applied to every type (and type environment) in the derivation,
  and as $\dom\msubst\subseteq\Monovars$ it will not conflict with any \Rule{Inst} node in the derivation.
  Subtyping relations are preserved as $t_1\leq t_2 \Rightarrow t_1\msubst\leq t_2\msubst$.
\end{proof}

\begin{lemma}[Generalisation lemma]\label{generalisation_lemma}
  If $\Gamma, x:s\subst \vdasha e:t$ with $e$ a ground expression, $\subst\disjoint\Gamma$, and $\Gamma\vdasha e':s$,
  then $\Gamma\vdasha e\subs{x}{e'}:t'$ for some $t'$ such that there exists a substitution $\subst'$
  such that $t'\subst'\simeq t$ and $\subst'\disjoint\Gamma$.
\end{lemma}
\begin{proof}
  We can suppose, without loss of generality, that $\subst$ is a full generalisation,
  that is, an injective substitution mapping every type variable in $(\vars{x}\land\Monovars)\setminus\vars{\Gamma}$
  to a fresh polymorphic type variable (and being the identity for any other type variable).
  Indeed, if for any $\subst$ the conditions of Lemma~\ref{generalisation_lemma} are satisfied, then they are also satisfied
  for any such generalisation according to the monotonicity lemma (Lemma~\ref{monotonicity_lemma}).

  Now, we transform the derivation of $\Gamma, x:s\subst \vdasha e:t$ into a canonical form derivation $D$
  of $\Gamma, x:s\subst \vdasha e:t'$ (with $t'\polyleq t$) using Theorem~\ref{norm_thm}.
  Let $\Psubst$ such that $t'\Psubst \leq t$. Then, using Lemma~\ref{inst_elimination_lemma},
  we transform $D$ into a derivation $D'$ of $\Gamma\Psubst', x:(s\subst)\Psubst' \vdasha e:t'\Psubst$,
  for some $\Psubst'$, and such that $D'$ does not contain any \Rule{Inst} node.
  Adding a \Rule{$\leq$} at the root gives $\Gamma\Psubst', x:(s\subst)\Psubst' \vdasha e:t$.

  Let $\psubst''$ be an injective substitution mapping every polymorphic type variable appearing in the image of
  at least one of the substitutions $\psubst'\in\Psubst'$ and $\subst$ to a fresh monomorphic type variable
  (and being the identity for any other type variable).
  From $\Gamma\Psubst', x:(s\subst)\Psubst' \vdasha e:t$, we can derive a derivation
  of $(\Gamma\Psubst')\psubst'', x:((s\subst)\Psubst')\psubst'' \vdasha e:t\psubst''$ simply by renaming the type variables
  everywhere according to $\psubst''$. This is only possible because the derivation does not contain any \Rule{Inst} node
  (substituting a polymorphic type variable by a monomorphic one could invalidate \Rule{Inst} nodes).
  By monotonicity (Lemma~\ref{monotonicity_lemma}), we get $\Gamma, x:((s\subst)\Psubst')\psubst'' \vdasha e:t\psubst''$.

  Now, we consider the set of substitutions $\{\subst_i\}_{i\in I}\eqdef\{ \psubst''\circ\psubst'\circ\subst\,\alt\, \psubst'\in\Psubst' \}$.
  For any $i\in I$, we can decompose $\subst_i$ into two substitutions,
  $\msubst_i \eqdef \restr{\psubst_i}{\Monovars}$ and $\psubst_i\eqdef \restr{\psubst_i}{\Polyvars}$.
  Note that $\forall i\in I.\ \msubst_i\disjoint\Gamma$ (it follows from $\subst\disjoint\Gamma$).
  Rewriting the previous judgement with these new notations gives
  $\Gamma, x:\tbwedge_{i\in I}(s\msubst_i)\psubst_i \vdasha e:t\psubst''$.

  Using Lemma~\ref{instantiation_lemma} on the derivation of $\Gamma\vdasha e':s$,
  we can derive, for any $i\in I$, $\Gamma\vdasha e':s\msubst_i$ (we recall that $\msubst_i\disjoint\Gamma$,
  and thus $\Gamma\msubst_i\simeq\Gamma$). Using a \Rule{Inst} node,
  we can then derive $\Gamma\vdasha e':(s\msubst_i)\psubst_i$ for any $i\in I$.
  Regrouping those derivations with a \Rule{$\wedge$} node gives $\Gamma\vdasha e':\tbwedge_{i\in I}(s\msubst_i)\psubst_i$.

  Finally, using the substitution lemma (Lemma~\ref{substitution_lemma}) on
  $\Gamma, x:\tbwedge_{i\in I}(s\msubst_i)\psubst_i \vdasha e:t\psubst''$ and $\Gamma\vdasha e':\tbwedge_{i\in I}(s\msubst_i)\psubst_i$,
  we get a derivation for $\Gamma \vdasha e\subs{x}{e'}:t\psubst''$.
  We recall that the substitution $\psubst''$ is injective and thus inversible,
  with its inverse $\subst'$ being such that $\subst'\disjoint\Gamma$ (as the image of $\psubst''$ is only composed of fresh monomorphic type variables).
  Thus, we have $(t\psubst'')\subst'\simeq t$, which concludes the proof.
\end{proof}

We now have all the results necessary to prove the type safety of programs for the parallel semantics.
The parallel semantics is extended to programs in a straightforward way:

\begin{minipage}[b]{0.45\linewidth}
  \begin{displaymath}
  \begin{array}{l}
  \textbf{Reduction rule}\,\,\,\,\,
  \begin{array}{rcll}
  \tletexp{x}{v}{p} & \preducesprog& p\subs x v
  \end{array}\\
  \textbf{Evaluation Context}\,\,\,\,\,
  \begin{array}[b]{r@{~}r@{~}l}
    P & ::= &  [\,]\alt\tletexp{x}{[]}{p}
  \end{array}
  \end{array}
  \end{displaymath}
  \end{minipage}\hfill\begin{minipage}{0.40\linewidth}
  \begin{displaymath}
    \Infer {e \preduces e'}{P[e]\preducesprog P[e']}{}
  \end{displaymath}
\end{minipage}\vspace*{0.5cm}

\begin{theorem}[Subject reduction for programs]\label{subj_red_prog}
  If $\Gamma\vdashp p:t$ and $p\preducesprog p'$, then $\Gamma\vdashp p':t$.
\end{theorem}
\begin{proof}
  Direct consequence of Theorem~\ref{subj_red} and Lemma~\ref{generalisation_lemma}.
\end{proof}

\begin{theorem}[Progress for programs]\label{progress_prog}
  If $\emptyenv\vdashp p:t$, then either $p$ is a value or $\exists p'.\ p\preducesprog p'$.
\end{theorem}
\begin{proof}
  Direct consequence of Theorem~\ref{progress}.
\end{proof}

\begin{theorem}[Type safety for the parallel semantics]\label{type_safety_parallel_prog}
  For any program $p$, if $\emptyenv\vdashp p:t$, then either $p\preducesprog^* v$ with
  $\emptyenv\vdashp v:t$ or $e\preducesprog^\infty$.
\end{theorem}
\begin{proof}
  Straightforward consequence of Theorem~\ref{subj_red_prog} and Theorem~\ref{progress_prog}.
\end{proof}

\subsubsection{Type Safety for the Source Semantics}

The final step is to deduce a type safety theorem for the source semantics (Figure~\ref{fig:syntax})
from the type safety theorem for the parallel semantics.
In order to help comparing the two reduction semantics,
we introduce again another reduction semantics $\creduces$ on expressions
that can perform a reduction $\treduces$
under any context $\ct$ (not just an evaluation context):

\begin{minipage}[b]{0.45\linewidth}
\begin{displaymath}
\begin{array}[b]{r@{~}r@{~}l}
  \ct &::=&  [\,]\alt \ct e \alt e \ct \alt
  (\ct,e) \alt (e,\ct) \alt \pi_i \ct \\
  &&\alt\tcase{\ct}{\tau}{e}{e}\alt\tcase{e}{\tau}{\ct}{e}\alt\tcase{e}{\tau}{e}{\ct}\\
\end{array}
\end{displaymath}
\end{minipage}\hfill\begin{minipage}[b]{0.45\linewidth}
\begin{displaymath}
  \Infer {e \treduces e'}{\ct[e]\creduces \ct[e']}{}
\end{displaymath}
\end{minipage}

\begin{definition}
  We say that a context $\ct_1$ is a subcontext of $\ct_2$, noted $\ct_1\ctleq\ct_2$,
  if and only if there exists a context $\ct_1'$ such that $\ct_2\equiv \ct_1[\ct_1']$.
\end{definition}

\begin{lemma}\label{creduces_lemma}
  For any expressions $e_1$ and $e_3$, if we have a chain $e_1\creduces^* e_3$
  such that at least one of the reduction steps
  happen under an evaluation context, then there exists an expression $e_2$ such that
  $e_1\reduces e_2$ and $e_2\creduces^* e_3$.
\end{lemma}
\begin{proof}
  Let $e_1$ and $e_4$ two expressions such that $e_1\creduces^* e_4$, where at least one of the reduction steps
  happen under an evaluation context.

  Let's consider the first reduction step $e_2\creduces e_3$ happening under an evaluation context.
  We have $e_1\creduces^* e_2$ (with no reduction step happening under an evaluation context)
  and $e_3\creduces^* e_4$.
  Moreover, we have $e_2\equiv E[e_2']$ and $e_3\equiv E[e_3']$ for some evaluation context $E$
  and expressions $e_2'$ and $e_3'$ such that $e_2'\treduces e_3'$.

  No reduction step in $e_1\creduces^* e_2$ happen under a context $\ct$ such that $\ct\ctleq E$:
  otherwise, it would also be an evaluation context, contradicting the fact that $e_2\creduces e_3$ is the first
  reduction step happening under an evaluation context.
  Consequently, we can "reverse" in $E$ and $e_2'$ the reduction
  steps happening in $e_1\creduces^* e_2$ (one after the other, in reverse order):
  \begin{itemize}
    \item If a reduction step happens under a context $\ct$ such that $E\ctleq \ct$, it only involves $e_2'$,
    we can thus apply the reverse rewriting to $e_2'$. After that, the expression we get
    is still reducible at top-level, as the reduction step that has been reversed
    cannot happen under an evaluation context (no reduction step in $e_1\creduces^* e_2$ can happen under an evaluation context).
    \item Otherwise, if a reduction step happens under a context $\ct$ such that
    $\ct\not\ctleq E$ and $E\not\ctleq \ct$, it only involves $E$,
    we can thus apply the reverse rewriting to $E$. After that, the new context we get
    is still an evaluation context, as the reduction step that has been reversed
    cannot happen under an evaluation context (no reduction step in $e_1\creduces^* e_2$ can happen under an evaluation context).
  \end{itemize}
  After this reversing process, we get a new evaluation context $E'$ and expression $e_1'$ such that $e_1\equiv E'[e_1']$
  and $e_1'\treduces e_2''$ for some $e_2''$ such that $E'[e_2'']\creduces^* E[e_2']$.
  
  Consequently, we have $e_1\equiv E'[e_1']\reduces E'[e_2'']$, and
  $E'[e_2'']\creduces^* E[e_2'] \equiv e_3 \creduces^* e_4$, which concludes the proof.
\end{proof}

\begin{lemma}\label{preduces_lemma}
  For any expressions $e_1$, $e_2$ and $e_3$, if $e_1\creduces^* e_2\preduces e_3$,
  then there exists an expression $e_1'$ such that $e_1\reduces e_1'\creduces^* e_3$.
\end{lemma}
\begin{proof}
  The step $e_2\preduces e_3$ can be decomposed into several $\creduces$ steps
  with at least one happening under an evaluation context.
  Thus, this lemma is an immediate consequence of Lemma~\ref{creduces_lemma}
  applied on the chain $e_1\creduces^* e_2\creduces^* e_3$.
\end{proof}

\begin{lemma}\label{creduces_value_lemma}
  For any expression $e$ and value $v$, if $e\creduces^* v$
  then either $e\reduces^\infty$ or there exists a value $v'$ such that $e\reduces^* v'\creduces^* v$.
\end{lemma}
\begin{proof}
  If $e$ is not already a value, then there must be at least one step in $e\creduces^* v$
  that happen under an evaluation context (it is not possible for $v$ to be a value otherwise).
  We can thus apply Lemma~\ref{creduces_lemma} successively, starting on $e\creduces^* v$
  and continuing until the remaining $e'\creduces^* v$ chain is such that $e'$ is a value.
  If this process terminates, it builds a chain $e\reduces^* v'$ with $v'\creduces^* v$,
  otherwise it builds $e\reduces^\infty$.
\end{proof}

\begin{lemma}\label{preduces_value_lemma}
  For any expression $e_1$, $e_2$, and value $v$,
  if $e_1\creduces^*e_2\preduces^* v$ then there exists $v'$ such that
  $e_1\reduces^* v'\creduces^* v$ or $e_1\reduces^\infty$.
\end{lemma}
\begin{proof}
  By induction on the number of steps in $e_2\preduces^* v$, we prove using Lemma~\ref{preduces_lemma}
  that $e_1\reduces^* e'\creduces^* v$ for some $e'$.
  Then, we conclude by applying Lemma~\ref{creduces_value_lemma} on $e'\creduces^* v$.
\end{proof}

\begin{lemma}\label{preduces_diverge_lemma}
  For any expression $e_1$ and $e_2$,
  if $e_1\creduces^*e_2\preduces^\infty$ then $e_1\reduces^\infty$.
\end{lemma}
\begin{proof}
  We can arbitrarily extend a chain $e_1\reduces\dots$ using Lemma~\ref{preduces_lemma}.
\end{proof}

The semantics $\creduces$ is extended into a semantics $\creducesprog$ for programs,
with an extended context allowing to perform a reduction under any top-level definition of the program:

\begin{minipage}[b]{0.45\linewidth}
\begin{displaymath}
\begin{array}[b]{r@{~}r@{~}l}
  \ctprog &::=& [] \alt \tletexp{x}{[]}{p} \alt \tletexp{x}{e}{\ctprog}
\end{array}
\end{displaymath}
\end{minipage}\hfill\begin{minipage}[b]{0.45\linewidth}
\begin{displaymath}
  \Infer {e \treduces e'}{\ctprog[\ct[e]]\creducesprog \ctprog[\ct[e']]}{}
\end{displaymath}
\end{minipage}

\begin{lemma}\label{preduces_prog_value_lemma}
  For any program $p_1$, $p_2$, and value $v$,
  if $p_1\creducesprog^*p_2\preducesprog^* v$ then there exists $v'$ such that
  $p_1\reducesprog^* v'\creducesprog^* v$ or $p_1\reducesprog^\infty$.
\end{lemma}
\begin{proof}
  Straightforward induction on the number of top-level definitions in $p_2$, using Lemma~\ref{preduces_value_lemma}
  (note that $p_1$ and $p_2$ must have the same number of top-level definitions as $p_1\creducesprog^*p_2$).
\end{proof}

\begin{lemma}\label{preduces_prog_diverge_lemma}
  For any program $p_1$ and $p_2$,
  if $p_1\creducesprog^*p_2\preducesprog^\infty$ then $p_1\reducesprog^\infty$.
\end{lemma}
\begin{proof}
  Straightforward induction on the number of top-level definitions in $p_2$, using Lemma~\ref{preduces_diverge_lemma}
  (note that $p_1$ and $p_2$ must have the same number of top-level definitions as $p_1\creducesprog^*p_2$).
\end{proof}

\begin{lemma}\label{test_type_preservation_lemma}
  For any values $v_1$ and $v_2$ such that $v_1\creduces^* v_2$,
  if $v_2\in\tau$ then $v_1\in\tau$.
\end{lemma}
\begin{proof}
  As $v_1$ is a value, every reduction step in $v_1\creduces^* v_2$ can only happen under a $\lambda$-abstraction.
  Given that the $\in$ relation ignores the body of $\lambda$-abstractions ($(\lambda x.e) \in \Empty\to\Any$ for any $e$),
  none of the reduction steps in $v_1\creduces^* v_2$ can change the relation $\cdot\in\tau$.
\end{proof}

\begin{theorem}[Type safety for the source semantics]\label{type_safety_prog}
  For any program $p$, if $\emptyenv\vdashp p:\tau$, then either $p\reducesprog^* v$ with $v\in\tau$
  or $p\reducesprog^\infty$.
\end{theorem}
\begin{proof}
  Straightforward combination of the type safety for the parallel semantics (Theorem~\ref{type_safety_parallel_prog})
  with the Lemmas~\ref{preduces_prog_value_lemma} and \ref{preduces_prog_diverge_lemma}.
  In the case where we get $p\reducesprog^* v$, for some value $v$ such that there exists $v'$
  such that $v\creduces^* v'$ and $\emptyenv\vdashp v':\tau$,
  we can deduce $\emptyenv\vdasha v':\tau$, then $v'\in\tau$ using Proposition~\ref{test_types_prop},
  and finally $v\in\tau$ using Lemma~\ref{test_type_preservation_lemma}.
\end{proof}

\subsection{Algorithmic Type System}

\subsubsection{Maximal Sharing Canonical Form}

This section applies to definitions of Appendix~\ref{sec:msc-appendix}.

As defined in Section~\ref{sec:normalisation}, we consider that expressions of the source language
can contain binding variables.
In particular, the unwinding operator $\eras {.}$ can be used on atoms and
canonical forms containing free binding variables.
An expression without binding variables is called ground expression.

\begin{proposition}
    For any ground expression of the source language $e$,
    $\eras{\term(\ucanon e)}\equiv e$.
    \end{proposition}
\begin{proof}
    Straightorward structural induction on $e$.
\end{proof}

\begin{proposition}\label{eqcan_sound_prop}
    If $\kappa\eqcan\kappa'$, then $\eras{\kappa}\aequiv\eras{\kappa'}$.
\end{proposition}

\begin{proof}
    If a reordering, as defined in Definition~\ref{def:order}, applies at top-level on the expression
    $\bindexp{\xx_1}{a_1}{\bindexp{\xx_2}{a_2}{\kappa}}$, the unwinding remains unchanged:
    as $\xx_1\not\in\fv{a_2}$ and $x_2\not\in\fv{a_1}$, we have
    $\kappa\subs{\xx_1}{a_1}\subs{\xx_2}{a_2} \equiv \kappa\subs{\xx_2}{a_2}\subs{\xx_1}{a_1}$.

    The general case can easily be deduced with the observation that
    $\forall \ct,\kappa_1,\kappa_2.\ \eras{\kappa_1}\aequiv\eras{\kappa_2} \Rightarrow \eras{\ct[\kappa_1]}\aequiv\eras{\ct[\kappa_2]}$
    (with $\ct$ denoting an arbitrary context).
\end{proof}

\begin{proposition}[Equivalence of MSC-forms]\label{prop:mscfequiv}
    If $\kappa_1$ and $\kappa_2$ are two MSC-forms and $\eras {\kappa_1}\aequiv\eras{\kappa_2}$, then $\kappa_1\eqcan\kappa_2$.
\end{proposition}
\begin{proof}
    We will show that $\kappa_2$ can be transformed into $\kappa_1$ just with $\alpha$-renaming and reordering of independent
    bindings (as specified in the definition of $\eqcan$).

    In this proof, we represent \textit{partially unwinded canonical forms}
    by a pair $(\benv, e)$, where $\benv$ is a binding context
    and $e$ an expression. With this representation, the unwinding of $(\benv,e)$ is
    $\bsubs{e}{\benv}$, but for clarity we can also use the notation $\eras{(\benv,e)}$.

    Let $(\benv_1,\xx_1)$ be the representation of $\kappa_1$,
    with $\benv_1$ representing its top-level definitions and $\xx_1$ its final binding variable,
    and $(\benv_2,\xx_2)$ be the representation of $\kappa_2$.
    Formally, we have $\term(\benv_1,\xx_1)\equiv\kappa_1$ and $\term(\benv_2,\xx_2)\equiv\kappa_2$.

    As $\eras {\kappa_1}\aequiv\eras{\kappa_2}$, we have $\eras {(\benv_1,\xx_1)}\aequiv\eras{(\benv_2,\xx_2)}$.
    By $\alpha$-renaming, we can assume that $\xx_1=\xx_2=\xx$ and $\eras {(\benv_1,\xx)}\equiv\eras{(\benv_2,\xx)}$.

    Now, let's prove the property below, from which Proposition~\ref{prop:mscfequiv} can be deduced.
    \textit{
        Let $\benv_1$ and $\benv_2$ two binding contexts, and $e$ an expression such that:
        \begin{itemize}
            \item $\eras{(\benv_1,e)}\equiv\eras{(\benv_2,e)}$
            \item The body of $\lambda$-abstractions in $\benv_1$ and $\benv_2$ are in MSC-form (Definition~\ref{def:maximalsharing})
            \item Both $\benv_1$ and $\benv_2$ satisfy the following properties
            (corresponding to the MSC-form properties applied to the top-level definitions), written here for
            a binding context $\benv$:
            \begin{enumerate}
            \item if $\be{\xx_1}{a_1}$ and $\be{\xx_2}{a_2}$ are distinct
                definitions in $\benv$, then $a_1 \not\eqcan a_2$
            \item for any definition $\be \xx {\lambda z. \kappa}$ in $\benv$, any binding $\bindexp{\yy}{a}{\kappa'}$
                in $\kappa$ is such that $\fv a \not\subseteq\fv{\lambda z. \kappa}$
            \item if $\benv$ contains a definition $\be\xx a$, then $\xx$ is a free variable of
                one of the next definitions in $\benv$ or of $e$
            \end{enumerate}
        \end{itemize}
        Then, we can transform $\benv_2$ into $\benv_1$ just with $\alpha$-renaming, reordering of independent definitions,
        and replacement of an atom by a $\eqcan$-equivalent one.    
    }

    We prove this property by induction on the total number of atoms appearing in $\benv_1$
    (by counting top-level atoms as well as the sub-atoms they contain).

    The base case ($\benv_1=\bempty$) is trivial:
    as $\eras{(\benv_2,e)}\equiv\eras{(\benv_1,e)}\equiv\eras{(\bempty,e)}\equiv e$,
    we deduce with Property 3 of MSC-forms that $\benv_2=\bempty$.

    For the inductive case, let's consider $\benv_1 = \benv_1' \,;\, \be\xx {a_1}$.

    With property 3 of MSC-forms, we know that $\xx$ appears in $e$.
    As $\eras {(\benv_1,e)}\equiv\eras {(\benv_2,e)}$, we can deduce that $\eras{(\benv_1,\xx)}\equiv\eras{(\benv_2,\xx)}$.
    Thus, $\benv_2$ must also feature a definition for $\xx$, let's call $a_2$ the associated atom.
    We move in $\benv_2$ the definition $\be\xx {a_2}$ at the end (if not already),
    it gives $\benv_2 = \benv_2' \,;\, \be\xx {a_2}$.
    We then have $\eras{(\benv_1', a_1)}\equiv\eras{(\benv_2', a_2)}$.
    As every kind of atom introduces a different syntactic construction, we can deduce that
    $a_1$ and $a_2$ are atoms of the same kind.
    \begin{itemize}
        \item If $a_1$ and $a_2$ are atoms that are not $\lambda$-abstractions and that do not contain any binding variable
        (constants, lambda variables), we directly have $a_1 \equiv a_2$.
        \item If $a_1$ and $a_2$ are atoms that are not $\lambda$-abstractions and that contain only one binding variable
        (projections), we can $\alpha$-rename binding variables in $\benv_2$ so that $a_1 \equiv a_2$.
        \item If $a_1$ and $a_2$ are atoms that are not $\lambda$-abstractions and that contain two binding variables
        $\xx$ and $\yy$ (applications, pairs), we consider two cases:
        \begin{itemize}
            \item If $\eras{(\benv_1', \xx)}\aequiv \eras{(\benv_1', \yy)}$,
            let's show that we necessarily have $\xx=\yy$.
            We consider the binding context $\benv_\xx$, which is a cleaned version of $\benv_1'$ where all the definitions
            that are not related (directly or indirectly) to $\xx_1$ have been removed.
            Similarly, we consider the binding context $\benv_\yy$ where the definitions unrelated to $\yy$ have been removed.
            Then, we apply the induction hypothesis to the binding contexts $\benv_\xx$, $\benv_\yy\subs{\yy}{\xx}$ and the expression $\xx$.
            This tells us that $\benv_\xx$ and $\benv_\yy$ are equivalent modulo reordering of the definitions, $\alpha$-renaming and
            $\eqcan$-transformation of atoms. Thus, according to the property 1 of MSC-forms, $\xx$ and $\yy$ cannot have two distinct definitions in
            $\benv_1'$, and thus $\xx=\yy$. The same reasoning can be done for $a_2$, and thus we deduce
            that both $a_1$ and $a_2$ contain the same binding variable twice.
            Thus, we can $\alpha$-rename binding variables in $\benv_2$ so that $a_1 \equiv a_2$.
            \item Otherwise, $\xx$ and $\yy$ must be different, and the same applies to $a_2$.
            Thus, we can $\alpha$-rename binding variables in $\benv_2$ so that $a_1 \equiv a_2$.
        \end{itemize}
        \item If $a_1$ and $a_2$ are typecases (containing 3 binding variables), we proceed similarly to the
        previous case to obtain $a_1 \equiv a_2$.
        \item In the case where $a_1$ and $a_2$ are $\lambda$-abstractions, let's say $\lambda x.\ \kappa_1$
        and $\lambda x.\ \kappa_2$, we note $(\benv_{\xx_1}, \xx_1)$ and
        $(\benv_{\xx_2}, \xx_2)$ the representations of $\kappa_1$ and $\kappa_2$ respectively.
        We now consider the representation $(\benv_1';\benv_{\xx_1}, \lambda x.\ \xx_1)$ and remove from it all the unused
        definitions (i.e. not related to $\xx_1$), it gives us a new representation $(\benv_1'', \lambda x.\ \xx_1)$ .
        We do the same for $(\benv_2';\benv_{\xx_2}, \lambda x.\ \xx_2)$, it gives a new representation $(\benv_2'', \lambda x.\ \xx_2)$.
        Then, we use the induction hypothesis on $\benv_1''$, $\benv_2''\subs{\xx_2}{\xx_1}$ and $\lambda x.\ \xx_1$.
        It gives us that $\benv_1''$ and $\benv_2''$ are equivalent modulo reordering of the definitions, $\alpha$-renaming and
        $\eqcan$-transformation of the atoms. By using the property 2 of MSC-forms, we can deduce that $\benv_{\xx_1}$ and $\benv_{\xx_2}$
        are equivalent modulo reordering of the definitions, $\alpha$-renaming and
        $\eqcan$-transformation of the atoms.
        Thus, we can $\alpha$-rename $\benv_2$ and $\eqcan$-transform some of its atoms so that $\benv_{\xx_1}\equiv \benv_{\xx_2}$,
        and thus $a_1\equiv a_2$.
    \end{itemize}

    In any case, we get $a_1\equiv a_2\equiv a$, thus the last definition of $\benv_1$ is the same as the last definiton of $\benv_2$.
    The same can be proven for the previous definitions by using the induction hypothesis on
    $\benv_1'$, $\benv_2'$ and $e\subs{\xx}{a}$.
\end{proof}
  
\begin{proposition}\label{msfr_sound_prop}
    If $\kappa\msred\kappa'$, then $\eras{\kappa}\aequiv\eras{\kappa'}$.
\end{proposition}
\begin{proof}
    Similar to Proposition~\ref{eqcan_sound_prop}.
\end{proof}

\begin{proposition}[Normalisation]\label{msfr_norm_prop}
    There is no infinite chain $\kappa_1\msred\kappa_2\msred\cdots$
\end{proposition}
\begin{proof}
    Let $n$ be the maximal number of nested $\lambda$-abstractions in $\kappa_1$.
    We call depth of a binding the number of nested $\lambda$-abstractions it is into (the depth of
    a binding of $\kappa_1$ is at most $n$).

    Let $N_{\kappa}(i)$ be the number of bindings of depth $i$ in an expression $\kappa$.
    Let $S(\kappa)$ be the following n-uplet: $(N_{\kappa}(n), N_{\kappa}(n-1), \dots, N_{\kappa}(0))$.

    For any chain $\kappa_1\msred\kappa_2\msred\cdots$,
    the sequence $S(\kappa_1)$, $S(\kappa_2)$, $\dots$ is strictly decreasing with respects to the lexicographic order.
    Thus $\kappa_1\msred\kappa_2\msred\cdots$ cannot be infinite.
\end{proof}

\begin{proposition}\label{msfr_sound2_prop}
    If $\kappa\not\msred$ (i.e., no $\msred$ rule apply on $\kappa$), then $\kappa$ is an MSC-form.
\end{proposition}
\begin{proof}
    We assume $\kappa\not\msred$ and show that all 3 MSC properties are satisfied.

    The property 3 (no unused bindings) is trivially verified: any binding that does not satisfy this property can directly be
    eliminated with the rule \ref{useless}. As the rule \ref{useless} does not apply, this property must be satisfied.

    Now, we focus on the property 2 (extrusion of bindings). We assume that there exists a sub-expression $\lambda x.\ \kappa_1$ of $\kappa$
    and a sub-expression $\bindexp{\yy}{a}{\kappa_2}$ of $\kappa_1$ such that $\fv a\subseteq\fv {\lambda x.\ \kappa_1}$.
    We know that $a$ does not depend on $x$, otherwise $x$ would be in $\fv a$ and not in $\fv{\lambda x.\ \kappa_1}$.
    Thus, we can reorder the binding $y$ (\ref{modeqcan})
    in the first position of its inner-most containing lambda-abstraction, and then apply the rule \ref{extrude} on it,
    which contradicts $\kappa\not\msred$. Consequently, the property 2 is satisfied.

    Finally, we show that the property 1 (sharing) is also satisfied.
    We assume that there are two distinct bindings $\bindexp{\xx_1}{a_1}{\dots}$ and $\bindexp{\xx_2}{a_2}{\dots}$
    such that $a_1\eqcan a_2$. As the property 2 is satisfied, and as $\fv{a_1}=\fv{a_2}$, we know that
    these two bindings are in the same $\lambda$-abstraction.
    Thus, we can reorder them (\ref{modeqcan}) to be the one next to the other so that the rule \ref{sharing} is applicable,
    which contradicts $\kappa\not\msred$. Thus, the property 1 is satisfied.
\end{proof}

\begin{proposition}[Confluence]\label{msfr_confl_prop}
    Let $\kappa_1$, $\kappa_2$ and $\kappa_2'$ such that
    $\kappa_1\msred^*\kappa_2$ and $\kappa_1\msred^*\kappa_2'$.
    Then, there exists $\kappa_3$ and $\kappa_3'$ such that
    $\kappa_2\msred^*\kappa_3$, $\kappa_2'\msred^*\kappa_3'$, and $\kappa_3\eqcan\kappa_3'$.
\end{proposition}
\begin{proof}
    Immediate consequence of Proposition~\ref{msfr_norm_prop} (normalisation), Proposition~\ref{msfr_sound_prop} (preservation of $\eras{.}$),
    Proposition~\ref{msfr_sound2_prop} ($\not\msred$ implies MSC-form), and Proposition~\ref{prop:mscfequiv} (equivalence of MSC-forms).
\end{proof}

\subsubsection{Soundness}

See Appendix~\ref{sec:algo-appendix} for the full algorithmic system,
without the rules for extensions.

\begin{lemma}[Monotonicity]\label{monotonicity_algo_lemma}\ \\
    If $\Gamma\vdashA\ea{\kappa}{\annot}:t$ and $\Gamma'\polyleq\Gamma$, then $\exists\annot',t'.\ \Gamma'\vdashA\ea{\kappa}{\annot'}:t'$ with $t'\polyleq t$.\\
    If $\Gamma\vdashA\ea{a}{\annota}:t$ and $\Gamma'\polyleq\Gamma$, then $\exists\annota',t'.\ \Gamma'\vdashA\ea{a}{\annota'}:t'$ with $t'\polyleq t$.
\end{lemma}

\begin{proof}
    Straightforward structural induction on the derivation.
\end{proof}

\begin{theorem}[Soundness]\label{algo_soundness}
    If $\Gamma\vdashA \ea \kappa \annot:t$, then $\Gamma\vdasha \eras \kappa: t$.
    If $\Gamma\vdashA \ea a \annota:t$, then $\Gamma\vdasha \eras a: t$.
\end{theorem}
\begin{proof}
We proceed by structural induction on the typing derivation of $\Gamma\vdashA \ea \kappa \annot:t$
(resp. $\Gamma\vdashA \ea a \annota:t$) in order to build a derivation
$\Gamma\vdasha \eras \kappa: t$ (resp. $\Gamma\vdasha \eras a: t$).
We consider the root of the derivation:
\begin{description}
    \item[\Rule{Const\AR}] Trivial (we use a \Rule{Const} node).
    \item[\Rule{Ax\AR}] Trivial (we use a \Rule{\Axl} node).
    \item[\Rule{$\to$I\AR}] We have $a\equiv\lambda x.\ \kappa$,
    and thus $\eras a\equiv\lambda x.\ \eras \kappa$.
    
    By induction on the premise, we get $\Gamma, x:\mt\vdasha \eras \kappa: s$.
    By applying the rule \Rule{$\to$I}, we get $\Gamma \vdasha \eras a: \mt\to s$
    (with $t\simeq\mt\to s$).

    \item[\Rule{$\to$E\AR}] We have $a\equiv\xx_1 \xx_2$.
    We pose $t_1\eqdef\Gamma(\xx_1)\Sigma_1$ and $t_2\eqdef\Gamma(\xx_2)\Sigma_2$.

    With an \Rule{\Axv} node, we can derive $\Gamma\vdasha \xx_1: \Gamma(\xx_1)$
    and $\Gamma\vdasha \xx_2: \Gamma(\xx_2)$.
    Using a \Rule{\InstLeq} pattern, we can derive from that
    $\Gamma\vdasha \xx_1: t_1$ and $\Gamma\vdasha \xx_2: t_2$.
    We have $t\simeq t_1\circ t_2$. Thus, according to the definition of $\circ$, we know that $t_1\leq t_2\to t$.
    Thus, with an application of the \Rule{$\leq$} rule on $\Gamma\vdasha \xx_1: t_1$, we can derive $\Gamma\vdasha \xx_1: t_2\to t$.
    We can then conclude with an application of the \Rule{$\to$E} rule.

    \item[\Rule{$\times$I\AR}] We have $a\equiv(\xx_1, \xx_2)$.
    
    With an \Rule{\Axv} node, we can derive $\Gamma\vdasha \xx_1: \Gamma(\xx_1)\renaming_1$
    and $\Gamma\vdasha \xx_2: \Gamma(\xx_2)\renaming_2$
    (with $\renaming_1$ and $\renaming_2$ as in the \Rule{$\times$I\AR} node).
    We can then conclude with an application of the \Rule{$\times$I} rule.

    \item[\Rule{$\times$E$_1$\AR}] We have $a\equiv\pi_1 \xx$.
    We pose $t_\circ=\Gamma(\xx)\Sigma$.

    With an \Rule{\Axv} node, we can derive $\Gamma\vdasha \xx: \Gamma(\xx)$.
    Using a \Rule{\InstLeq} pattern, we can derive from that $\Gamma\vdasha \xx: t_\circ$.
    We have $t\simeq \bpi_1(t_\circ)$. Thus, according to the definition of $\bpi_1$,
    we know that $t_\circ\leq t \times \Any$.
    Thus, with an application of the \Rule{$\leq$} rule, we can derive $\Gamma\vdasha \xx: t \times \Any$.
    We can then conclude with an application of the \Rule{$\times$E$_1$} rule.

    \item[\Rule{$\times$E$_2$\AR}] Similar to the previous case.
    \item[\Rule{$\Empty$\AR}] Similar to the previous case.

    \item[\Rule{$\in_1$\AR}] Similar to the previous case.

    \item[\Rule{$\in_2$\AR}] Similar to the previous case.

    \item[\Rule{Var\AR}] Trivial (we use a \Rule{\Axv} node).

    \item[\Rule{Bind$_1$\AR}]
    We have $\kappa\equiv\bindexp{\xx}{a}{\kappa_\circ}$ and thus
    $\eras \kappa \equiv \eras{\kappa_\circ}\subs{\xx}{\eras a}$.

    By induction on the premise, we get $\Gamma \vdasha \eras {\kappa_\circ}: t$.
    As $\xx\not\in\dom\Gamma$, we know that this derivation does not contain any \Rule{\Axv} node applied on $\xx$.
    We can thus easily transform it into a derivation of $\Gamma \vdasha \eras {\kappa_\circ}\subs{\xx}{\eras a}: t$
    by substituting every occurrence of $\xx$ by $\eras a$.

    \item[\Rule{Bind$_2$\AR}]
    We have $\kappa\equiv\bindexp{\xx}{a}{\kappa_\circ}$ and thus
    $\eras \kappa \equiv \eras{\kappa_\circ}\subs{\xx}{\eras a}$.

    We can suppose without loss of generality that the decomposition $\{\mt_i\}_{i\in I}$
    is a partition of $\Any$: if for any two distinct $i,j\in I$, $\mt_i$ and $\mt_j$ are not disjoint,
    then we can replace the $\mt_j$ part by $\mt_j\setminus\mt_i$ and conclude by monotonicity
    (Lemma~\ref{monotonicity_algo_lemma}).

    By induction on the first premise, we get $\Gamma\vdasha \eras a:s$.
    For any $i\in I$, we apply the induction hypothesis on the corresponding premise.
    It gives $\Gamma, \xx:s\land \mt_i \vdasha \eras{\kappa_\circ} : t_i$.
    With a \Rule{$\leq$} node, we can obtain $\Gamma, \xx:s\land \mt_i \vdasha \eras{\kappa_\circ} : t$
    (with $t\simeq\tbvee_{i\in I} t_i$).
    We conclude with a \Rule{\Vee} node.

    \item[\Rule{$\wedge$\AR}] Trivial induction on the premises.
\end{description}
\end{proof}
  
\subsubsection{Completeness}

\begin{definition}[Atomic source expression]\label{ase_def}
    We say that an expression of the source language is an atomic source expression if it can be constructed with the following syntax:
    \begin{equation*}
        \begin{array}{lrclr}
        \textbf{Atomic source expressions} &\ase &::=& c\alt x\alt \lambda x. e\alt (\xx,\xx)\alt \xx \xx\alt \pi_i \xx\alt \tcase{\xx}{\tau}{\xx}{\xx}
        \end{array}
    \end{equation*}
    and if, for the case $\lambda x. e$,
    all sub-expressions of $e$ are either a binding variable or they contain a lambda variable
    that is not in $\fv{\lambda x. e}$.

    The variable $\ase$ is used to range over atomic source expressions.
\end{definition}

\begin{definition}
    For any atomic source expression $\ase$, we define $\MSCFA{\ase}$ as follows:
    \begin{align*}
        \MSCFA{\lambda x. e} &= \lambda x. \MSCF e\\
        \MSCFA{\ase} &= \ase &\text{ for any $\ase$ that is not a $\lambda$-abstraction}
      \end{align*}      
\end{definition}

\begin{proposition}\label{mscfa_prop}
    For any atomic source expression $\ase$, $\bindexp{\xx}{\MSCFA{\ase}}{\xx}$
    is a valid MSC-form.
\end{proposition}
\begin{proof}
    The extrusion property is ensured by the condition on $\lambda$-abstractions in Definition~\ref{ase_def}.
    The two other properties are trivially satisfied.
\end{proof}

\begin{definition}
  For a binding context $\benv$ and an expression $e$,
  we define the operator $\bdiff e \benv$ as follows: 
    \begin{align*}
        \bdiff{e}{\bempty} &= e\\
        \bdiff{e}{((\xx,a);\benv)} &= \bdiff{(e\esubs{\eras a}{\xx})}{\benv}
    \end{align*}
\end{definition}

\begin{proposition}\label{eras_share_prop}
    For any binding context $\benv$ and expression $e$,
    we have $\bsubs {(\bdiff e \benv)} \benv \aequiv \bsubs e \benv$ and
    $\bdiff {(\bsubs e \benv)} \benv \aequiv \bdiff e \benv$.
\end{proposition}
\begin{proof}
    Straightforward induction on $\benv$.
\end{proof}

In the following, we fix an expression order $\leqexpr$ over expressions (c.f. Desfinition~\ref{def:leqexpr}).
This order should be total, so that for any expression $e$,
it determines a unique MSC-form $\MSCF{e}$ modulo $\alpha$-renaming (and not only modulo $\eqcan$):
independent consecutive bindings must follow the increasing order $\leqexpr$ of their unwinding.
As the order $\leqexpr$ is arbitrarily choosen, the proofs below will work for any MSC-form.

\begin{lemma}[Decomposition of canonical form derivations]\label{form_decomposition}
    If $D$ is a canonical form derivation of $\Gamma\vdasha e: t$, with the root being a \Rule{\Vee} node
    doing the substitution $e'\subs{\xx}{e_\xx}$, then there exists a binding context $\benv$
    such that $\MSCF e\aequiv \term(\benv,\bindexp{\xx}{\MSCFA{\bdiff{e_\xx}\benv}}{\MSCF{\bdiff{e'}\benv}})$.
\end{lemma}
\begin{proof}
    First, we deduce from the fact that $D$ is a canonical form derivation
    that the definition premise, $\Gamma\vdasha e_\xx:s$, is an atomic derivation.

    We know that $e_\xx$ appears in $e$ (as $D$ is canonical, see Definition~\ref{def:acceptable}).
    Thus, we can deduce that $\MSCF{e}$ contains a definition for an atom $a$ that unwinds to $e_\xx$.
    Formally, we know that there exists a binding context $\benv$, an atom $a$ and a canonical form $\kappa$
    such that $\MSCF{e}\aequiv\term(\benv,\bindexp{\xx}{a}{\kappa})$ with $\bsubs{\eras a}\benv \aequiv e_\xx$.

    First, we determine what $\benv$ is.
    The expression $e_\xx$ could contain some sub-expressions that are not binding variables and
    that have no occurrence of a lambda variable defined in $e_\xx$.
    In this case, these sub-expressions should be defined by some bindings in $\benv$ (the unwinding of each such binding
    is necessarily smaller than $e_\xx$ by $\leqexpr$, as $\leqexpr$ is an extension of the sub-expressions order).
    The expression $e'$ could also contain some such sub-expressions. The ones whose unwnding is smaller than $e_\xx$
    according to $\leqexpr$ must be defined by some bindings in $\benv$ too. No other expression should be defined in $\benv$
    or it would contradict the properties of MSC-forms.

    Now, we determine what $a$ is.
    Under the context $\benv$, the expression $\bdiff{e_\xx}\benv$ unwinds to $e_\xx$ (Proposition~\ref{eras_share_prop}).
    Moreover, as $\Gamma\vdasha e_\xx:s$ is an atomic derivation, and given how $\benv$ is constructed,
    $\bdiff{e_\xx}\benv$ must be an atomic source expression.
    Thus, we can deduce from Proposition~\ref{mscfa_prop} that $\MSCFA{\bdiff{e_\xx}\benv}$ can be used in place of the atom $a$ without breaking any property of the MSC-form,
    and thus by unicity of the MSC-form (Proposition~\ref{prop:mscfequiv}) we can conclude that $a\aequiv \MSCFA{\bdiff{e_\xx}\benv}$.

    Finally, we determine what $\kappa$ is. We note $\benv'$ the binding context $\benv;(\xx,\MSCFA{\bdiff{e_\xx}\benv})$.
    The expression $\bdiff{e'}{\benv'}$ unwinds to $e'\subs{\xx}{e_\xx}$
    under the context $\benv'$ (Proposition~\ref{eras_share_prop}).
    As $D$ is a canonical form, $e_\xx$ cannot be a sub-expression of $e'$ (c.f. Definition~\ref{def:acceptable}), thus
    $\bdiff{e'}{\benv}\aequiv\bdiff{e'}{\benv'}$ and thus $\bdiff{e'}{\benv}$ also unwinds to $e'\subs{\xx}{e_\xx}$.
    Thus, $\MSCF{\bdiff{e'}\benv}$ can be used in place of $\kappa$ without breaking any property of the MSC-form
    (note that it only contains top-level bindings for expressions whose unwinding is greater than $e_\xx$ by $\leqexpr$, as the smaller ones have been put in $\benv$).
    By unicity of the MSC-form (Proposition~\ref{prop:mscfequiv}), we conclude that $\kappa\aequiv \MSCF{\bdiff{e'}\benv}$,
    and thus $\MSCF e\aequiv \term(\benv,\bindexp{\xx}{\MSCFA{\bdiff{e_\xx}\benv}}{\MSCF{\bdiff{e'}\benv}})$.
\end{proof}

\begin{lemma}[Completeness]\label{algo_gen_completeness}
    If $D$ is a canonical form derivation of $\Gamma\vdasha e: t$, then $\exists \annot,t'.\ \Gamma\vdashA \ea{\MSCF{e}}{\annot}: t'$ with $t'\polyleq t$.\\
    If $D$ is a canonical atomic derivation of $\Gamma\vdasha \ase: t$ (with $\ase$ an atomic source expression), then $\exists \annota,t'.\ \Gamma\vdashA \ea{\MSCFA{\ase}}{\annota}: t'$ with $t'\polyleq t$.
\end{lemma}
\begin{proof}
    We proceed by induction on the depth of $D$.

    We consider the root of the derivation (the cases up to \Rule{$\in_2$} are for canonical atomic derivations,
    the cases after are for canonical form derivations):
    \begin{description}
        \item[\Rule{Const}] Trivial.
        \item[\Rule{\Axl}] Trivial.
        \item[\Rule{$\to$I}] We have $\ase\equiv\lambda\xx.\ e$ and thus $\MSCFA \ase \aequiv \lambda\xx.\ \MSCF{e}$.

        The premise of this \Rule{$\to$I} node is a canonical form derivation.
        Thus, by induction on this premise, we get $\Gamma,\xx:\mt\vdashA\ea{\MSCF{e}}{\annot}:t'$ (with $t'\polyleq t$).
        We can thus derive $\Gamma\vdashA\ea{\lambda\xx.\ \MSCF{e}}{\annotlambda{\mt}{\annot}}:\mt\to t'$
        and we have $\mt\to t'\polyleq \mt\to t$, which concludes this case.

        \item[\Rule{$\to$E}] We have $\ase\equiv \xx_1 \xx_2$ and thus $\MSCFA \ase \aequiv \xx_1 \xx_2$.

        As $D$ is a canonical atomic derivation, we know that the second premise, $\Gamma\vdasha \xx_2:t_1$, is a \Rule{\InstLeq} pattern
        with no \Rule{$\leq$} node and whose premise is a \Rule{Ax$_{\vee}$} node. Thus, we know that there exists $\Sigma_2$ such that
        $\Gamma(\xx_2)\Sigma_2\simeq t_1$.
        Similarly, the first premise, $\Gamma\vdasha \xx_1:t_1\to t_2$, is a \Rule{\InstLeq} pattern whose premise is a \Rule{Ax$_{\vee}$} node.
        Thus, we know that there exists $\Sigma_1$ such that $\Gamma(\xx_1)\Sigma_1\leq t_1\to t_2$.

        Consequently, and by definition of $\circ$, we know that $(\Gamma(\xx_1)\Sigma_1) \circ (\Gamma(\xx_2)\Sigma_2) \leq t_2$.
        We can thus derive $\Gamma\vdashA\ea{\xx_1 \xx_2}{\annotapp{\Sigma_1}{\Sigma_2}}:t'$
        (with $t'\simeq (\Gamma(\xx_1)\Sigma_1) \circ (\Gamma(\xx_2)\Sigma_2)$) such that $t'\leq t_2$, which concludes this case.

        \item[\Rule{$\times$I}] We have $\ase\equiv (\xx_1, \xx_2)$ and thus $\MSCFA \ase \aequiv (\xx_1, \xx_2)$.

        As $D$ is a canonical atomic derivation, both premises can only be \Rule{Ax$_{\vee}$} nodes.
        Thus, we can deduce that there exists two renamings of polymorphic type variables $\renaming_1$ and $\renaming_2$
        such that $\Gamma(\xx_1)\renaming_1 \simeq t_1$ and $\Gamma(\xx_2)\renaming_2 \simeq t_2$. Thus,
        we can derive $\Gamma\vdashA\ea{(\xx_1, \xx_2)}{\annotpair{\renaming_1}{\renaming_2}}:t_1\times t_2$.

        \item[\Rule{$\times$E$_1$}] We have $\ase\equiv \pi_1 \xx$ and thus $\MSCFA \ase \aequiv \pi_1 \xx$.

        As $D$ is a canonical atomic derivation, we know that the premise, $\Gamma\vdasha \xx:t_1\times t_2$,
        is a \Rule{\InstLeq} pattern whose premise is a \Rule{Ax$_{\vee}$} node.
        Thus, we know that there exists $\Sigma$ such that $\Gamma(\xx)\Sigma\leq t_1\times t_2$.

        Consequently, and by definition of $\bpi_1$, we know that $\bpi_1 (\Gamma(\xx)\Sigma)\leq t_1$.
        We can thus derive $\Gamma\vdashA\ea{\pi_1 \xx}{\annotproj{\Sigma}}:t'$
        (with $t'\simeq \bpi_1 (\Gamma(\xx)\Sigma)$) such that $t'\leq t_1$, which concludes this case.

        \item[\Rule{$\times$E$_2$}] Similar to the previous case.
        \item[\Rule{$\Empty$}] We have $\ase\equiv \tcase{\xx}{\tau}{\xx_1}{\xx_2}$ and thus $\MSCFA \ase \aequiv \tcase{\xx}{\tau}{\xx_1}{\xx_2}$.
 
        As $D$ is a canonical atomic derivation, we know that the premise, $\Gamma\vdasha \xx:\Empty$,
        is a \Rule{\InstLeq} pattern with no \Rule{$\leq$} node and whose premise is a \Rule{Ax$_{\vee}$} node.
        Thus, we know that there exists $\Sigma$ such that $\Gamma(\xx)\Sigma\simeq\Empty$.

        We can thus derive $\Gamma\vdashA\ea{\tcase{\xx}{\tau}{\xx_1}{\xx_2}}{\annotempty\Sigma}:\Empty$.
    
        \item[\Rule{$\in_1$}] We have $\ase\equiv \tcase{\xx}{\tau}{\xx_1}{\xx_2}$ and thus $\MSCFA \ase \aequiv \tcase{\xx}{\tau}{\xx_1}{\xx_2}$.
 
        As $D$ is a canonical atomic derivation, we know that the first premise, $\Gamma\vdasha \xx:\tau$,
        is a \Rule{\InstLeq} pattern whose premise is a \Rule{Ax$_{\vee}$} node.
        Thus, we know that there exists $\Sigma$ such that $\Gamma(\xx)\Sigma\leq\tau$.
        Similarly, the second premise, $\Gamma\vdasha \xx_1:t_1$, can only be a \Rule{Ax$_{\vee}$}.
        Thus, we know that there exists a renaming of polymorphic variables $\renaming$
        such that $\Gamma(\xx_1)\renaming\simeq t_1$.

        We can thus derive $\Gamma\vdashA\ea{\tcase{\xx}{\tau}{\xx_1}{\xx_2}}{\annotthen\Sigma}:\Gamma(\xx_1)$
        with $\Gamma(\xx_1)\polyleq t_1$.
    
        \item[\Rule{$\in_2$}] Similar to the previous case. 
        \item[\Rule{\Axv}] Trivial.
        \item[\Rule{$\leq$}] Straightforward induction on the premise.
        \item[\Rule{Inst}] Straightforward induction on the premise.
        \item[\Rule{$\wedge$}] By induction on the premises, we get $\forall i\in I.\ \Gamma\vdashA \ea{\MSCF{e}}{\annot_i}:t_i'$
        with $t_i'\polyleq t_i$. Thus, we can derive $\Gamma\vdashA \ea{\MSCF{e}}{\annotinter{\{\annot_i\}_{i\in I}}}:\tbwedge_{i\in I}t_i'$
        (with $\tbwedge_{i\in I}t_i'\polyleq \tbwedge_{i\in I}t_i$).

        \item[\Rule{\Vee}] By using Lemma~\ref{form_decomposition}, we know that there exists $\benv$ such that
        $\MSCF e\aequiv \benv[\bindexp{\xx}{\MSCFA{\bdiff{e_\xx}\benv}}{\MSCF{\bdiff{e'}\benv}}]$
        (with $\bdiff{e_\xx}\benv$ being an atomic source expression).

        The unwiding of the top-level definitions in $\benv$ are necessarily smaller than $e_\xx$ by $\leqexpr$.
        Thus, none of them can be defined by a \Rule{\Vee} node in $D$.
        Consequently, and as in canonical form derivation a structural node can only appear in the first premise of a \Rule{\Vee} node,
        none of the expressions defined in $\benv$ are typed in $D$. Thus, these sub-expressions can easily be substituted,
        in $D$, by the associated binding variables in $\benv$. It yields a derivation for
        $\Gamma\vdasha \bdiff{(e'\subs{\xx}{e_\xx})}\benv:t$, or equivalently, for
        $\Gamma\vdasha (\bdiff {e'} \benv)\subs{\xx}{(\bdiff {e_\xx} \benv)}:t$.

        By induction on the premises of this new derivation, we get $\Gamma\vdashA \ea {\MSCFA{\bdiff{e_\xx}\benv}}{\annota} : s'$ (with $s'\polyleq s$)
        and $\forall i\in I.\ \Gamma,\xx:s\land \mt_i\vdashA \ea {\MSCF{\bdiff{e'}\benv}}{\annot_i} : t_i$ (with $t_i\polyleq t$).
        By monotonicity (Lemma~\ref{monotonicity_algo_lemma}), we can derive
        $\forall i\in I.\ \Gamma,\xx:s'\land \mt_i\vdashA \ea {\MSCF{\bdiff{e'}\benv}}{\annot_i'} : t_i'$ (with $t_i'\polyleq t_i\polyleq t$).
        We can thus derive $\Gamma\vdashA \ea{\bindexp{\xx}{\MSCFA{\bdiff{e_\xx}\benv}}{\MSCF{\bdiff{e'}\benv}}}{\annot}:\tbvee_{i\in I}t_i'$ with $\annot=\annotbind{\annota}{\{(\mt_i,\annot_i')\}_{i\in I}}$.

        From that, we can easily derive $\Gamma\vdashA \ea{\MSCF{e}}{\annot'}:\tbvee_{i\in I}t_i'$
        with $\annot'$ obtained by inserting at the root of $\annot$ a $\texttt{skip}$ annotation for each definition in $\benv$.
    \end{description}
\end{proof}

\begin{theorem}[Completeness]\label{algo_completeness}
    If $\Gamma\vdasha e: t$ with $e$ a ground expression, then $\exists \annot,t'.\ \Gamma\vdashA \ea{\MSCF{e}}{\annot}: t'$ with $t'\polyleq t$.
\end{theorem}
\begin{proof}
    Direct application of Lemma~\ref{algo_gen_completeness} after
    using the normalisation theorem (Theorem~\ref{norm_thm})
    on a derivation of $\Gamma\vdasha e: t$.
\end{proof}

\subsection{Annotations Reconstruction System}

This section contains proofs for the reconstruction system (Appendix~\ref{sec:reconstruction-appendix}).

\begin{theorem}[Soundness]
    If $\Gamma\paavdash\epa{\kappa}{\banns}\refines\annot$,
    then $\exists t.\ \Gamma\vdashA\ea{\kappa}{\annot}:t$.\\
    If $\Gamma\paavdash\epa{a}{\lanns}\refines\annota$,
    then $\exists t.\ \Gamma\vdashA\ea{a}{\annota}:t$.
\end{theorem}
\begin{proof}
    We proceed by structural induction on the derivation of $\Gamma\paavdash\epa{\kappa}{\banns}\refines\annot$
    or $\Gamma\paavdash\epa{a}{\lanns}\refines\annota$.

    If the derivation $\Gamma\paavdash\epa{a}{\lanns}\refines\annota$ has a \Rule{App} root,
    we construct a derivation $\Gamma\vdashA\ea{a}{\annota}:t$ (for some $t$) with a \Rule{$\to$E\AR} root.
    In order to satisfy the guard-conditions of the rule \Rule{$\to$E\AR}, we need to prove that
    the annotation $\annotapp{\Sigma_1}{\Sigma_2}$ generated by the \Rule{App} root satisfies
    $\Gamma(\xx_1)\Psubst_1\leq\Empty\to\Any$ and $\Gamma(\xx_2)\Psubst_2\leq\dom{\Gamma(\xx_1)\Psubst_1}$.

    We have, in the premise of the \Rule{App} root, $\Psubst=\tallyf{\Gamma(\xx_1)\renaming_1}{\Gamma(\xx_2)\renaming_2\to \polyvar}$
    and $\Psubst\neq\emptyset$. Let $\psubst\in\Psubst$. By definition of the tallying problem,
    we have $(\Gamma(\xx_1)\renaming_1)\psubst\leq (\Gamma(\xx_2)\renaming_2\to \polyvar)\psubst$,
    which can be rewritten $\Gamma(\xx_1)(\psubst\circ\renaming_1)\leq (\Gamma(\xx_2)(\psubst\circ\renaming_2))\to(\polyvar\psubst)$.
    From that subtyping relation, we can deduce $\Gamma(\xx_1)(\psubst\circ\renaming_1)\leq \Empty\to\Any$,
    and by definition of $\dom{.}$, $\Gamma(\xx_2)(\psubst\circ\renaming_2)\leq\dom{\Gamma(\xx_1)(\psubst\circ\renaming_1)}$.
    As $(\psubst\circ\renaming_2)\in\Psubst_2$ and $(\psubst\circ\renaming_1)\in\Psubst_1$, we
    deduce $\Gamma(\xx_1)\Psubst_1\leq \Empty\to\Any$ and $\Gamma(\xx_2)\Psubst_2\leq\dom{\Gamma(\xx_1)\Psubst_1}$
    (we use the fact that $\dom{.}$ is monotonically non-increasing).

    The other cases are similar or straightforward.
\end{proof}

Now, we propose a sketch of proof justifying that the deduction rules
for the reconstruction system define a terminating algorithm.
The idea of this proof is similar to the proof of termination of
the \textit{Kirby-Paris hydra} game\footnote{\url{https://en.wikipedia.org/wiki/Hydra_game}}:
we can associate an ordinal number weight to each node,
and this weight can only decrease as the game (or derivation) advances.
Intuitively, this non-negative weight represents the advancement of
the game (or derivation).
Though sub-trees can sometimes be duplicated,
their weight is always lowered before being duplicated,
resulting in a lower weight overall.

\newcommand{\rw}[3]{\texttt{w}(#1,#2,#3)}
\newcommand{\resw}[3]{\texttt{w}(#1,#2,#3)}

For any type environment $\Gamma$, canonical form or atom $\kappaa$,
and intermediate annotation $\aanns$ compatible with the structure of $\kappaa$, the weight $\rw\Gamma\kappaa\aanns$
is the ordinal number defined as follows:

\begin{minipage}{.95\textwidth}
\begin{align*}
    \rw \Gamma \kappaa \annotityp & = 1\\
    \rw \Gamma \kappaa \annotiuntyp & = 1\\
    \rw \Gamma {\kappaa} {\annotiinter{S_1}{S_2}} & = \textstyle\sum \{\rw {\Gamma} \kappaa {\aanns}\,\alt\,\aanns\in S_1\}
\end{align*}
\end{minipage}

\begin{minipage}{.45\textwidth}
\begin{align*}
    \rw \Gamma c \annotiinfer & = \omega\\
    \rw \Gamma x \annotiinfer & = \omega
\end{align*}
\end{minipage}
\begin{minipage}{.45\textwidth}
\begin{align*}
    \rw \Gamma {\lambda x.\kappa} \annotiinfer & = \omega^{\rw \Gamma \kappa \annotiinfer}\\
    \rw \Gamma {\lambda x.\kappa} {\annotilambda{\mt}{\banns}} & = \omega^{\rw \Gamma \kappa \banns}
\end{align*}
\end{minipage}

\begin{align*}
    \rw \Gamma {\pi_i \xx} \annotiinfer & = \omega&\text{if }\xx\in\dom\Gamma\\
    \rw \Gamma {\pi_i \xx} \annotiinfer & = \omega^2&\text{otherwise}
\end{align*}
\begin{align*}
    \rw \Gamma {\xx_1\xx_2} \annotiinfer & = \omega&\text{if } \{\xx_1,\xx_2\}\subseteq\dom\Gamma\\
    \rw \Gamma {\xx_1\xx_2} \annotiinfer & = \omega^2&\text{otherwise, if } \xx_1\in\dom\Gamma\\
    \rw \Gamma {\xx_1\xx_2} \annotiinfer & = \omega^2&\text{otherwise, if } \xx_2\in\dom\Gamma\\
    \rw \Gamma {\xx_1\xx_2} \annotiinfer & = \omega^3&\text{otherwise}
\end{align*}
\begin{align*}
    \rw \Gamma {(\xx_1,\xx_2)} \annotiinfer & = \omega&\text{if } \{\xx_1,\xx_2\}\subseteq\dom\Gamma\\
    \rw \Gamma {(\xx_1,\xx_2)} \annotiinfer & = \omega^2&\text{otherwise, if } \xx_1\in\dom\Gamma\\
    \rw \Gamma {(\xx_1,\xx_2)} \annotiinfer & = \omega^2&\text{otherwise, if } \xx_2\in\dom\Gamma\\
    \rw \Gamma {(\xx_1,\xx_2)} \annotiinfer & = \omega^3&\text{otherwise}
\end{align*}
\begin{align*}
    \rw \Gamma {\tcase{\xx_0}{\tau}{\xx_1}{\xx_2}} \annotithenelse & = \omega\\
    \rw \Gamma {\tcase{\xx_0}{\tau}{\xx_1}{\xx_2}} \annotiinfer & = \omega^2&\text{if }\Gamma(\xx_0)\leq\tau\\
    \rw \Gamma {\tcase{\xx_0}{\tau}{\xx_1}{\xx_2}} \annotiinfer & = \omega^2&\text{otherwise, if }\Gamma(\xx_0)\leq\neg\tau\\
    \rw \Gamma {\tcase{\xx_0}{\tau}{\xx_1}{\xx_2}} \annotiinfer & = \omega^3&\text{otherwise, if }\xx_0\in\dom\Gamma\\
    \rw \Gamma {\tcase{\xx_0}{\tau}{\xx_1}{\xx_2}} \annotiinfer & = \omega^4&\text{otherwise}
\end{align*}
\begin{align*}
    \rw \Gamma {\bindexp{\xx}{a}{\kappa}} {\annotiskip{\banns}} & = \omega^{\rw \Gamma \kappa \banns}\\
    \rw \Gamma {\bindexp{\xx}{a}{\kappa}} {\annotibind{\lanns}{\sanns_1}{\sanns_2}} & = \omega^{\alpha}\\
    &\hspace*{-3.5cm}\text{with }\alpha=\textstyle\sum \{\rw {(\Gamma,\xx:\mt)} \kappa {\banns}\alt(\mt,\banns)\in\sanns_1\}\\
    \rw \Gamma {\bindexp{\xx}{a}{\kappa}} {\annotiprop{\lanns}{\Gammas}{\sanns_1}{\sanns_2}} & =
    \omega^{\alpha+\card\Gammas}\\
    &\hspace*{-3.5cm}\text{with }\alpha=\textstyle\sum \{\rw {(\Gamma,\xx:\mt)} \kappa {\banns}\alt(\mt,\banns)\in\sanns_1\}\\
    \rw \Gamma {\bindexp{\xx}{a}{\kappa}} {\annotitrybind{\lanns}{\banns_1}{\banns_2}} & = \omega^{\alpha}\\
    &\hspace*{-3.5cm}\text{with }\alpha=\textstyle\sum \{\rw \Gamma a \lanns,\ \rw {(\Gamma,\xx:\Any)} \kappa {\banns_1},\ \rw {\Gamma} \kappa {\banns_2}\}\\
    \rw \Gamma {\bindexp{\xx}{a}{\kappa}} {\annotitryskip{\banns}} & = \omega^{\alpha}\\
    &\hspace*{-3.5cm}\text{with }\alpha=\textstyle\sum \{\rw{\Gamma}{a}{\annotiinfer},\ \rw \Gamma \kappa \banns\}
\end{align*}
\begin{align*}
    \rw \Gamma \xx \annotiinfer & = \omega &\text{if } \xx\in\dom\Gamma\\
    \rw \Gamma \xx \annotiinfer & = \omega^2 &\text{otherwise}
\end{align*}
where, for any multiset $\{\alpha_1, \alpha_2, \dots, \alpha_n\}$,
$\textstyle\sum \{\alpha_1, \alpha_2, \dots, \alpha_n\} \eqdef \alpha_1 + \alpha_2 + \dots + \alpha_n$
with $\alpha_1\geq\alpha_2\geq\dots\geq\alpha_n$.

Then, we define a weight $\resw{\Gamma}{\kappaa}{\results}$
for any type environment $\Gamma$, canonical form or atom $\kappaa$,
and result $\results$ compatible with the structure of $\kappaa$:
\begin{align*}
    \resw{\Gamma}{\kappaa}{\resok{\aanns}} & = 1\\
    \resw{\Gamma}{\kappaa}{\resfail} & = 1\\
    \resw{\Gamma}{\kappaa}{\respart{\Gamma'}{\aanns_1}{\aanns_2}} & =
    \textstyle\sum\ \{\rw{\Gamma\land\Gamma'}{\kappaa}{\aanns_1}\}\cup\{\rw{(\Gamma,\xx:\neg \mt)}{\kappaa}{\aanns_2}\,\alt\,(\xx:\mt)\in\Gamma'\}\\
    \resw{\Gamma}{\kappaa}{\ressubst{\{\msubst_i\}_{i\in I}}{\aanns_1}{\aanns_2}} & =
    \textstyle\sum\ \{\rw{\Gamma}{\kappaa}{\aanns_2}\}\cup\{\rw{\Gamma\msubst_i}{\kappaa}{\aanns_1\msubst_i}\,\alt\,i\in I\}\\
    \resw{\Gamma}{\kappaa}{\resvar{\xx}{\aanns_1}{\aanns_2}} & = \textstyle\sum \{\rw{(\Gamma,\xx:\Any)}{\kappaa}{\aanns_1},\ \rw{\Gamma}{\kappaa}{\aanns_2}\}
\end{align*}

\begin{lemma}
    For any $\Gamma, \kappaa, \aanns$, and $\Gamma'$ such that $\Gamma'\leq\Gamma$,
    we have $\rw{\Gamma'}{\kappaa}{\aanns}\leq\rw{\Gamma}{\kappaa}{\aanns}$. 
\end{lemma}
\begin{proof}
    Straightforward induction.
\end{proof}

\begin{lemma}
    For any $\Gamma, \kappaa, \aanns$, and $\msubst$,
    we have $\rw{\Gamma\msubst}{\kappaa}{\aanns\msubst}\leq\rw{\Gamma}{\kappaa}{\aanns}$. 
\end{lemma}
\begin{proof}
    Straightforward induction (we recall that test types $\tau$ do not contain type variables).
\end{proof}

\begin{lemma}
    If $\Gamma\aavdash\epa{\kappaa}{\aanns}\refines\results$
    or $\Gamma\raavdash\epa{\kappaa}{\aanns}\refines\results$,
    then $\rw{\Gamma}{\kappaa}{\aanns}\gneq \resw{\Gamma}{\kappaa}{\results}$.
\end{lemma}
\begin{proof}
    Structural induction on the derivation of $\Gamma\aavdash\epa{\kappaa}{\aanns}\refines\results$
    or $\Gamma\raavdash\epa{\kappaa}{\aanns}\refines\results$.
\end{proof}

\begin{theorem}[Termination]\label{app:rec_term}
    The deduction rules $\raavdash$ and $\aavdash$ define a terminating algorithm:
    it can either fail (if no rule applies at some point) or return a result $\results$.
\end{theorem}
\begin{proof}
    There can only be finitely many \Rule{Iterate$_1$}
    and \Rule{Iterate$_2$} nodes applied on a given canonical form or atom,
    otherwise, according to the previous lemmas,
    we could extract from them an infinite decreasing chain of ordinal numbers.
\end{proof}

\ifpersonalcopy
\else
\setcounter{TotPages}{30}
\fi
\end{document}